\documentclass{article}
\usepackage[utf8]{inputenc}
\usepackage[a4paper, top=2.2cm, bottom=2.5cm, left=2cm, right=2cm]
{geometry}
\usepackage{amsmath,amssymb,amsfonts,bm,amscd,mathtools}
\usepackage{graphicx}
\usepackage{caption, subcaption}
\usepackage{comment}
\usepackage{cite}
\usepackage{float}
\usepackage{appendix}
\usepackage{physics}
\usepackage{array}
\usepackage{xcolor}
\usepackage{amsthm}
\usepackage{hyperref}
\hypersetup{hidelinks}
\usepackage{authblk} 
\usepackage{microtype}

\newcommand{\Z}{{\mathbb Z}}
\newcommand{\R}{{\mathbb R}}
\newcommand{\C}{{\mathbb C}}
\newcommand{\Q}{{\mathbb Q}}

\newcommand{\m}{{\mathrm{m}}}

\renewcommand{\d}{\partial}
\def\be{\begin{equation}}
\def\ee{\end{equation}}
\def\bea{\begin{align}}
\def\eea{\end{align}}

\newtheorem{theorem}{Theorem}
\newtheorem{corollary}{Corollary}
\newtheorem{lemma}[theorem]{Lemma}
\newtheorem{proposition}{Proposition}[section]
\newtheorem{conjecture}{Conjecture}
\newtheorem{definition}{Definition}
\newtheorem{example}{Example}[section]
\newtheorem{remark}{Remark}[section]

\numberwithin{corollary}{theorem}
\numberwithin{corollary}{proposition}

\title{$\hat{Z}$-TQFT, Surgery Formulas, and New Algebras}

\author{Pedro Guicardi}
\author{Mrunmay Jagadale}
\affil{California Institute of Technology, Pasadena, CA 91125, USA}
\date{\vspace{-1cm}}

\begin{document}
\maketitle

\begin{abstract}
    The $\hat{Z}$ invariants of three-manifolds introduced by Gukov--Pei--Putrov--Vafa have influenced many areas of mathematics and physics. However, their TQFT structure remains poorly understood. In this work, we develop a framework of decorated $\mathrm{Spin}$-TQFTs and construct one based on Atiyah--Segal-like axioms that computes the $\hat{Z}$ invariants. Central to our approach is a novel quantization of $SL(2,\mathbb{C})$ Chern--Simons theory and a $\mathbb{Q}$-extension of the algebra of observables on the torus, from which we obtain the torus state space of the $\hat{Z}$-TQFT. Using the torus state space and topological invariance, we uniquely determine the $\hat{Z}$ invariants for negative-definite plumbed manifolds. Within this TQFT framework, we establish gluing, rational surgery, partial surgery, satellite, and cabling formulas, as well as explicit closed-form expressions for Seifert manifolds and torus link complements. We also generalize these constructions to higher-rank gauge groups.
\end{abstract}

\tableofcontents

\section{Introduction}

A long-standing problem in mathematical physics has been the categorification of the $\mathrm{WRT}$ (Witten--Reshetikhin--Turaev) invariants\cite{Witten89, RT91, CraneFrenkel}. That is, the construction of a $4d$ TQFT that promotes the numerical invariants of 3-manifolds coming from Chern--Simons theory to $3d$ homology theories. This envisioned categorification would also contain powerful new 4-manifold invariants, which could be used to probe smooth phenomena in four-manifold topology. This categorification would be analogous to the categorification of the Jones polynomial $J_K(q)$ by Khovanov Homology \cite{Khovanovhomology, BarNatanKhhomology}.

The Jones polynomial is a Laurent polynomial in $ q = e^{\frac{2 \pi i }{k}}$, with integer coefficients, where $ k$ denotes the level. The integrality of these coefficients naturally suggests a categorification, as they can be interpreted as ranks of some homology groups. Indeed, the coefficient of $q^{n}$ in the Jones polynomial is precisely the graded rank of Khovanov homology in $q$-grading $n$. 

In contrast, the $\mathrm{WRT}$ invariant of a three-manifold is a function of the integer level $k$ with no evident grading structure and no immediate integrality property. These invariants are typically expressed as a finite sum with summation ranges that depend on the integer level $k$, rather than as an analytic function of $q$ \cite{Jeffrey92, FreedGompf}. To interpret $q$ as an honest grading of a homology theory, one must first re-express the three-manifold invariant as an analytic function of $q$ rather than $k$. This reparametrization requires analytic continuation in $k$, and one needs this analytic continuation to produce invariants with integer coefficients, which may then be interpreted as ranks of (co)homology groups. 

The $q$-series valued invariants introduced in \cite{GPPV}, $\hat{Z}_{\alpha}(M,q)$, have integer coefficients and allow one to express the $\mathrm{WRT}$ invariant as a finite linear combination of these $\hat{Z}_{\alpha}$ invariants evaluated at $q= e^{\frac{2 \pi i}{k}}$. Moreover, the invariants $\hat{Z}_\alpha(M,q)$ have a physical definition via M-theory, where they are defined as supersymmetric indices of $M5$-branes wrapping $S^1\times D^2\times M$ in the presence of a certain supersymmetric background in the 11-dimensional spacetime. Crucially, these supersymmetric indices are graded traces over BPS cohomologies $\mathcal{H}_\alpha (M)$. In other words, the $\hat{Z}_\alpha(M,q)$ invariants can be integrated as graded Euler characteristics of an underlying cohomology theory of three-manifolds whose vector spaces are spanned by BPS bound states of $M2$-$M5$-branes. Although the $\hat{Z}_\alpha(M,q)$ invariants are computable, the BPS vector spaces are difficult to compute. To develop techniques for computing these BPS vector spaces, a structural understanding of $\hat{Z}_\alpha(M,q)$ invariants is required, particularly the TQFT structure that computes them.

Preliminary versions of the $\hat{Z}$-TQFT were studied in \cite{GM, MJ}. In \cite{GM}, two-variable series associated to knot complements were introduced as $\hat{Z}$ invariants for the knot complement, along with their gluing rules. In \cite{MJ},  the vector space associated with $T^{2}$ and the $SL(2,\Z)$ action on it was proposed. However, many questions about the structure of the $\hat{Z}$-TQFT remain open. 

 The $\hat{Z}$ invariant for closed three-manifolds requires an additional label whose interpretation has evolved over time. Originally, this label was understood in \cite{GPPV} as an abelian $SL(2,\C)$ flat connection on the closed three-manifold. Later, \cite{GM} proposed that this label should instead be a $\mathrm{Spin}^{c}$-structure. While the sets of $\mathrm{Spin}^{c}$-structures and abelian flat connections on a three-manifold have the same cardinality, their behavior under diffeomorphisms is different. A $\mathrm{Spin}^{c}$-structure on a three-manifold can be decomposed into a $\mathrm{Spin}$-structure together with an element of $H^{1}(\cdot,\Z)$. In this paper, we present this additional label as a $\mathrm{Spin}$-structure and an element of $H^{1}(\cdot,\Q/\Z)$. For rational homology spheres, this distinction does not show up since the groups are isomorphic. However, for three-manifolds with boundary, the $H^{1}(\cdot, \Q/\Z)$ formulation enables us to write down functorial cutting and gluing rules. 

Along with developing the cutting and gluing rules for the $\hat{Z}$-TQFT, we describe the TQFT assignments to various manifolds. The vector space assigned to the torus plays a particularly important role due to its close connection with the category of line operators in the $\hat{Z}$-TQFT. We construct this vector space from a novel quantization of the moduli space of $SL(2,\C)$ flat connections on the torus. The resulting vector space is given by an appropriate completion of the following algebra,
$$ \frac{\langle X^{\lambda}, Y^{\mu}\vert \lambda, \mu \in \Q \rangle}{\{ X^{\lambda} Y^{\mu} = q^{-2\lambda \mu} Y^{\mu} X^{\lambda} \} } . $$
This vector space naturally comes with an $SL(2,\Z)$ action inherited from the moduli space of flat connections and a $H^{1}(T^{2}, \Q/\Z)$ grading. 

Another important TQFT assignment is what the theory assigns to tree-link complements, as a large class of three-manifolds can be built by gluing tree-link complements to each other. By a tree-link, we mean a link obtained by placing unknots on the vertices of a tree-graph and linking the unknot components according to the adjacency of the vertices. To a $N$-component tree-link with linking matrix $Q$, $S^{3} \setminus L_{Q}$, the $\hat{Z}$-TQFT assigns,
$$ \hat{Z}_{\mathbf{a}}(S^{3} \setminus L_{Q};, X, Y )  = \sum_{n \in \mathbf{a} + \Z^N } q^{-(n,Q n)}  \prod^N_{i=1} Y_{i}^{n_{i}} (X_i^{1/2} - X_i^{-1/2})^{1-\mathrm{deg}(v_{i}) } X^{(Q n)_{i}}_{i}, $$
where $\mathbf{a}$ is a combination of a $\mathrm{Spin}$-structure and an element of $H^{1}(S^{3} \setminus L_{Q},\Q/\Z)$ as explained in section $\ref{spinstructures_review}$. 

One can also construct a large class of three-manifolds by gluing two knot complements. Using the techniques developed in \cite{GM, S2}, one can write down the two-variable series, $F_{K}(X,q) = \sum_{m \in \frac{1}{2}+ \Z}  f_{K}^{m}(q) X^{m} $, for a large class of knot complements. We give a gluing formula for $\hat{Z}$ in terms of these two-variable series. For a three manifold $M$ obtained by gluing $S^3 \setminus \nu(K_1)$ and $S^3 \setminus \nu(K_2)$ along $T^{2}$ with mapping class group element $\gamma = \begin{pmatrix}
    r & a\\
    p & b
\end{pmatrix}$, we have,
$$  \hat{Z}_{\alpha}(M,q) = \sum_{m_1,m_2 \in \mathbb{Z}+\frac{1}{2}} f^{m_1}_{K_1}(q)\cdot f^{m_2}_{K_2}(q) \cdot q^{-\frac{1}{p}(b m_{1}^{2}+2 m_{1} m_{2}+ r m_{2}^{2})} \cdot \delta^\Z\left(\frac{m_{1}+m_{2} r}{p}+ \frac{\alpha}{p}  \right),  $$
where $f^{m_i}_{K_i}(q)$  are coefficients in the $X$-expansion of the two-variable series and  $\delta^\Z$ is the indicator function of $\Z$. For more details, see Proposition \ref{splicing_theorem}.

More generally, we explain and develop a formalism for gluing any two states along a torus boundary, given a $Spin^c$ structure and a mapping class group element (not just knot complements). This allows us to go even further and derive new formulas associated with the $\hat{Z}$ invariants. For instance, while in \cite{GM} a general formula was conjectured for 3-manifolds resulting from surgery on a knot, here we extend this to surgery on a link. That is, suppose $L$ is an $N$-component link, and $M = S^3_{\frac{p_1}{r_1},\ldots,\frac{p_N}{r_N}}(L)$. Suppose additionally that $Q$ is the linking matrix of $L$ where the diagonals have the rational framings $p_i/r_i$. Define,
$$\mathcal{L}^\alpha_{Q}  (X^\mathbf{\mu}) = \delta^{Q \mathbb{Z}^N} (\mathbf{\mu} - \frac{\alpha}{\mathbf{r}}) \cdot q^{-(\mu,Q^{-1} \mu)}.$$    
Here, $Q$ is an invertible $N\times N $ matrix, $\mu$ and $\alpha$ are $N$-component vectors and vector division is taken element wise, $(\frac{\alpha}{\mathbf{r}})_i = (\frac{\alpha_i}{r_i})$. Then, we conjecture, 
 $$\Hat{Z}_\alpha (S^3_{\frac{p_1}{r_1},\ldots, \frac{p_N}{r_N} } (L)) \cong \mathcal{L}^\alpha_{Q_f} \left( \prod^N_{i=1}(X_i^{\frac{1}{2r_i}} - X_i^{-\frac{1}{2r_i}}) \cdot F_L\right), $$
 where $\alpha$ takes values in $Spin^c(M)$, whose structure we describe in Section \ref{partial_surg_section}, and $F_L$ is the (Gukov--Manolescu) GM series of $L$. We prove this for the case of weakly negative definite plumbings. 

We use this to derive a formula for the $\hat{Z}$ invariants of all Seifert manifolds fibered over $S^2$ (of negative Euler number $\eta$),
$$\Hat{Z}_{\alpha}(M(b,0;\frac{p_1}{r_1},\ldots,\frac{p_d}{r_d});q) \cong \sum_{n\in \mathbb{Z}}q^{-\eta n^2 -2n \xi_\alpha} \cdot \Psi^n_{\Vec{p},\Vec{r},\alpha,0},$$
where 
$$\Psi^n_{\Vec{p},\Vec{r},\alpha,0} = \underset{i=1,\ldots,d}{\sum_{\epsilon_i=\pm1}} \prod_i \epsilon_i \delta^\mathbb{Z}\left( \frac{r_i n +\alpha_i+\epsilon_i/2}{p_i}\right)\cdot \mathrm{sgn}(m)^d \binom{\frac{d}{2}-2+|m|}{d-3}\delta^\mathbb{Z}\left(m+\frac{d}{2}\right)\bigg|_{m = \eta n_0 + \xi_\alpha- \sum_i\frac{\epsilon_i}{2p_i}}.$$
We also generalize this to more general Riemann surfaces of genus $g$ in the form of a conjecture. 

Similarly, using our newly developed framework, we find a simple expression for the GM series of all torus links $T(sp,tp)$,
    $$F_{T_{tp,sp}}(X_1,\ldots,X_p;q) = \sum_{n\in \mathbb{Z}+\frac{st(p-1) +1}{2}} f^n_{T_{sp,tp}}(q) \prod^p_{i=1}X_i^n,$$
where    
$$f^n_{T_{sp,tp}}(q) = q^{\frac{n^2}{st}}\sum_{\epsilon,\epsilon' = \pm1} \epsilon \epsilon' \mathrm{sgn}(-2n +\epsilon t +\epsilon' s )^p \binom{\frac{|\epsilon t+\epsilon' s -2n|+pst}{2st}-1}{ p-1} \delta^\mathbb{Z}\left(\frac{\epsilon t+\epsilon' s+pst -2n}{2st}\right) .$$
We use this expression, along with our novel satellite formula, to derive an expression for the GM series for the cabling of a knot or link.

We use our framework of decorated Spin TQFT's to prove and derive many novel formulas in the theory of $\hat{Z}$ invariants. Among these, we emphasize the following: 
\begin{enumerate}
    \item Gluing of two knot complements with an arbitrary $MCG$ element;
    \item Partial surgery on links ($-\frac{1}{r}$ and $\infty$ on unknot component and more generally as well);
    \item $\hat{Z}$ amplitude for arbitrary Seifert manifold fibered over $\Sigma_g$;
    \item Laplace transform for rational surgery on links;
    \item Satellite formula for GM series;
    \item GM series for all torus links; 
    \item Cabling formula for GM series;
    \item Whitehead doubling formula for GM series;
    \item Generalizations to an arbitrary semi-simple Lie group.
\end{enumerate}

\subsection*{Structure of the paper}
In Section 2, we review concepts regarding $\mathrm{Spin}$-structure and $\mathrm{Spin}^c$-structure in the context of 3-manifolds with and without boundary. Then, we define the notion of a `decorated $\mathrm{Spin}$ TQFT' with our version of Atiyah--Segal-like axioms. In Section 3, we describe the vector space associated with the torus and point out the $MCG$ representation on it that reproduces the $\hat{Z}$ invariants. In this section, we also define a bilinear form on this space that corresponds to the gluing of vectors along $T^2$.

In Section 4, we study the consistency conditions on the vectors assigned to the solid torus and point out the choice that leads to $\hat{Z}$ invariants. Additionally, we show that this choice of vector assignment together with the structure assigned to the torus uniquely determines the amplitudes assigned to all weakly negative definite plumbed manifolds, thus recovering the plumbing formula of \cite{GPPV}. Finally, we prove that the mapping class group action determines gluing formulas for two knot complements, which reduce to the surgery formula (Laplace transform) of \cite{GM}. We also provide a closed formula for the amplitude associated with any Seifert manifold fibered over $S^2$ with negative Euler number.

In Section 5, we use our newly developed framework to derive new surgery formulas for $\hat{Z}$ and Gukov--Manolescu invariants. We prove that Gukov--Manolescu series for links obey a partial surgery formula as well as a more general Laplace transform for rational surgery on all components. We use these results to find closed expressions for the GM series of all torus links. Furthermore, we derive a formula for the GM series of satellite knots and use it to write down a general cabling formula. Finally, we conjecture formulas for amplitudes of Seifert manifolds fibered over any Riemann surface and Whitehead doubling formulas.  

In Section 6, we repeat all of the above Sections for the case of general semi-simple Lie algebras. We conclude with Section 7, emphasizing future directions and open questions.

%%%%%%%%%%%%%%%%%%%%%%%%%%%%%%%%%%%%%%%%%%%%%%%%%%%%%%%%%%%%%%%%%%%%%%%%%%%%%%%%

\section{TQFT Structure}
In \cite{MJ} and \cite{CGP}, the notion of a `decorated' TQFT (or `non-semi-simple' TQFT) was studied in the context of the $\hat{Z}$ invariant. This is the structure we attribute to TQFT, which computes three-manifold invariants that depend on additional data other than the topology of the three-manifold. In the case of the $\hat{Z}$ invariant, this extra decoration is given by $\mathrm{Spin}^{c}$-structures. 

The set of $\mathrm{Spin}^{c}$-structures on a three-manifold $M$ is affinely isomorphic to $H^{2}(M,\Z)$. Thus, choosing a $\mathrm{Spin}^{c}$-structure on $M$ allows us to identify the set of decorations for $\hat{Z}(M)$ with $H^{2}(M,\Z)$. In fact, it suffices to choose a $\mathrm{Spin}$-structure on $M$ to make this identification, since there exists a canonical map from the set of $\mathrm{Spin}$-structures to the set of $\mathrm{Spin}^{c}$-structures.

%%%%%%%%%%%%%%%%%%%%%%%%%%%%%%%%%%%%%%%%%%%%%%%%%%%%%%%%%%%%%

\subsection{A Review of \texorpdfstring{$\mathrm{Spin}$-Structures}{Spin-Structures}}\label{spinstructures_review}
To better understand how $\mathrm{Spin}$-structures appear in the $\hat{Z}$ invariants, it is helpful to review the properties of $\mathrm{Spin}$-structures in low dimensions. In this sub-section, we will briefly review $\mathrm{Spin}$-structures on low-dimensional manifolds. We will use the definition of $\mathrm{Spin}$-structure from \cite{Milnor1965}. For more details about $\mathrm{Spin}$-structures on low-dimensional manifolds, see \cite{Milnor1965, AtiyahSpin, KirbyTaylor, Johnson, Beliakova1996, BM1996 }.
\begin{definition}
    Let $F(M)$ denote the frame bundle of the manifold $M$. A $\mathrm{Spin}$-structure on $M$ is a first cohomology class $\mathfrak{s}\in H^{1}(F(M),\Z_{2})$, whose restriction to each fibre is non-zero.
\end{definition}

For $n$-dimensional manifold $M$ we have the following exact sequence \cite{Milnor1965},
\be 
0 \rightarrow H^{1}(M, \Z_{2}) \xrightarrow{\pi^{*}} H^{1}(F(M), \Z_{2}) \xrightarrow{i^{*}} H^{1}(SO(n),\Z_{2}) \xrightarrow{\delta} H^{2}(M,\Z_{2}).
\ee
The group $ H^{1}(SO(n),\Z_{2})\cong \Z_{2}$, and the second Stiefel-Whitney class $\omega_{2}= \delta(1)$. Therefore, $\omega_{2}=0$ if and only if there exists $\mathfrak{s}\in H^{1}(F(M), \Z_{2}) $ such that $i^{*}(\mathfrak{s})=\frac{1}{2} \mod \Z$. Thus, we have a $\mathrm{Spin}$-structure on $M$ if and only if the second Stiefel-Whitney class $\omega_{2}$ vanishes. The exact sequence also tells us that the set of $\mathrm{Spin}$-structures $\{\mathfrak{s}\in H^{1}(F(M), \Z_{2}) \vert i^{*}(\mathfrak{s})= \frac{1}{2} \mod \Z  \}$ is affinely isomorphic to $H^{1}(M,\Z_{2})$. When evaluating a $\mathrm{Spin}$-structure, we will write its value lifted to $\frac{1}{2}\Z$, though it should be understood modulo $\Z$.

Let us now build intuition dimension by dimension.

We begin with $\mathrm{Spin}$-structures on $1$-dimensional manifolds. There are two $\mathrm{Spin}$-structures on a circle. One of the two $\mathrm{Spin}$-structures can be extended to the disc $D^{2}$ bounding the circle, while the other one can not be extended. We call the $\mathrm{Spin}$-structure that can be extended to the disc the bounding $\mathrm{Spin}$-structure on $S^{1}$.

Moving on to two-dimensional manifolds, a closed orientable surface of genus $g$, denoted $\Sigma_{g}$, admits $2^{2g}$ $\mathrm{Spin}$-structures. We can characterize these $\mathrm{Spin}$-structures based on whether or not they are bounding on the canonical homology basis (the $a$ and $b$ cycles) of $\Sigma_{g}$. A $\mathrm{Spin}$-structure is bounding on a cycle if it can be extended to a disc bounding that cycle.

In the rest of the paper, we will be concerned only with 3-manifolds with boundary tori, or the disjoint union of tori, so our discussion will be restricted to link complements. 

Now, let us look at the $\mathrm{Spin}$-structures on three-manifolds. Suppose we have a framed knot in a three-manifold $M$. Using the framing of the knot and the orientation on $M$, we can construct a frame of $M$ at each point on the knot.  Thus, the framed knot gives us a 1-cycle on $F(M)$, and we can evaluate a $\mathrm{Spin}$-structure on it. We can characterize a $\mathrm{Spin}$-structure on a three-manifold by its values on framed knots. If $U$ is a zero-framed un-knot in $M_{3}$, so that $U$ bounds a disc in $M_{3}$, then evaluation of any $\mathrm{Spin}$-structure $\mathfrak{s}\in \mathrm{Spin}(M)$ on $U$ gives us $\frac{1}{2} $. That is $\mathfrak{s}(U)=\frac{1}{2} $ for all $\mathfrak{s}\in \mathrm{Spin}(M)$. We will use this property to characterize $\mathrm{Spin}$-structures on three-manifolds.

\begin{example}[Zero-framed unknot complement]
   Let us consider the zero-framed unknot complement $S^{3}\setminus U$. Any $\mathrm{Spin}$-structure on $S^{3}\setminus U$ is bounding on the longitude. Therefore, $\forall \mathfrak{s} \in \mathrm{Spin}(S^{3}\setminus U)$, $\mathfrak{s}(\ell)=\frac{1}{2}$. A $\mathrm{Spin}$-structure on $S^{3}\setminus U$ may or may not be bounding on the meridian. We denote the $\mathrm{Spin}$-structure bounding on the meridian by $\mathfrak{s}_{\frac{1}{2}}$, and the non-bounding $\mathrm{Spin}$-structure by $\mathfrak{s}_{0}$.
\begin{align}
    \mathfrak{s}_{\frac{1}{2}}(m) & = \frac{1}{2} & \mathfrak{s}_{\frac{1}{2}}(\ell) & = \frac{1    }{2}, \\
    \mathfrak{s}_{0}(m) & = 0 & \mathfrak{s}_{0}(\ell) & = \frac{1}{2},
\end{align}
where $\ell$ denotes the longitude and $m$ denotes the meridian. 
\end{example}
\begin{example}[Zero-framed Hopf link complement]
     We will denote the longitude and meridians of the Hopf link by $\ell_{i}$ and $m_{i}$, respectively. A $\mathrm{Spin}$-structure $\mathfrak{s}$ is bounding on $m_{1}$ if and only if it is bounding on $\ell_{2}$. Similarly, a $\mathrm{Spin}$-structure $\mathfrak{s}$ is bounding on $m_{2}$ if and only if it is bounding on $\ell_{1}$. Thus, a $\mathrm{Spin}$-structure on the zero-framed Hopf link complement is completely determined by its value on the meridians. Let us denote the four $\mathrm{Spin}$-structures on the zero-framed Hopf link complement by $\mathfrak{s}_{ij}$, where $\mathfrak{s}_{ij}(m_{1})= i$, and $\mathfrak{s}_{ij}(m_{2})= j$. Therefore, we have 
\begin{align}
   \mathfrak{s}_{ij}(m_{1}) &= i & \mathfrak{s}_{ij}(\ell_{1}) &= j & \mathfrak{s}_{ij}(m_{2}) &= j & \mathfrak{s}_{ij}(\ell_{2}) &= i.
\end{align}
\end{example}
\begin{example}[Zero-framed unlink complement]
    Let us contrast this with the zero-framed unlink complement, that is, the complement of the disjoint union of two unknots. Any $\mathrm{Spin}$-structure on the zero-framed unlink complement is bounding on $\ell_{1}$ and $\ell_{2}$. Therefore, we have
\begin{align}
    \mathfrak{s}^{UL}_{ij}(m_{1}) &= i & \mathfrak{s}^{UL}_{ij}(\ell_{1}) &= \frac{1}{2} & \mathfrak{s}^{UL}_{ij}(m_{2}) &= j & \mathfrak{s}^{UL}_{ij}(\ell_{2}) &= \frac{1}{2}.
\end{align}
\end{example}

In the three examples discussed above, the $\mathrm{Spin}$-structures were determined by their value on the meridians. This is true in general. $\mathrm{Spin}$-structures on a link complement are entirely determined by their value on the meridians of the link components. Let $\mathcal{L}$ be a $n$-component framed link, let $\ell_{i}$ and $m_{i}$ denote the longitude and meridian of the $i$th component of $\mathcal{L}$, and let $L$ be the linking matrix of $\mathcal{L}$. The value of $\mathrm{Spin}$-structure $\mathfrak{s}$ on the longitude $\ell_{i}$ is given by \cite{Beliakova1996, BM1996}, 
\be 
\mathfrak{s}(\ell_{i}) = \frac{1}{2} + \frac{1}{2} L_{ii} + \sum_{j=1}^{n} L_{i j}\left(\frac{1}{2}+ \mathfrak{s}(m_{j})\right) \mod \Z.
\ee 
Suppose $S^{3}_{\mathcal{L}_{1},\mathcal{L}}$ is a three-manifold obtained by performing a Dehn surgery on a framed sub-link $\mathcal{L}_{1}$ of $\mathcal{L}$. Suppose the components of $\mathcal{L}_{1}$ are labeled by $\{1,2, \ldots , n_{1} \}$, then for a $\mathrm{Spin}$-structure $\mathfrak{s} \in \mathrm{Spin}(S^{3}_{\mathcal{L}_{1},\mathcal{L}})$ for all $i \in \{1,2, \ldots , n_{1} \}$,
\be 
\frac{1}{2} L_{ii} = \sum_{j=1}^{n} L_{i j}\left(\frac{1}{2}+ \mathfrak{s}(m_{j})\right) \mod \Z.
\ee 
Thus the set of $\mathrm{Spin}$-structures on $S^{3}_{\mathcal{L}_{1},\mathcal{L}}$ is given by 
\be 
\mathrm{Spin}(S^{3}_{\mathcal{L}_{1},\mathcal{L}})= \left\{ \mathfrak{s} \in \left( \frac{1}{2}\Z^{n} \right)\bigg/ \Z^{n} \Bigg\vert \hspace{0.2cm} \frac{1}{2} L_{ii} = \sum_{j=1}^{n} L_{i j}\left(\frac{1}{2}+ \mathfrak{s}_{j} \right) \mod \Z \hspace{0.2cm} \forall i \in \{1,2, \ldots , n_{1} \} \right\}.
\ee 

We will now describe the map from the set of $\mathrm{Spin}$-structures to the set of $\mathrm{Spin}^{c}$-structures on the three-manifold $M_{\mathcal{L}}$ obtained by a surgery on an $n$-component framed link. For further details, see \cite{DM2002}. A combinatorial description of the set of $\mathrm{Spin}^{c}$-structures on $M_{\mathcal{L}}$ is given by 
$$\mathrm{Spin}^{c}(M_{\mathcal{L}}) = \delta + 2 \Z^{n} / 2 L \Z^{n}, $$
where $L$ is the linking matrix, and 
$\delta = \sum_{\substack{j =1 \endline j \neq i}}^{n} L_{ij} .$
The map from $\mathrm{Spin}(M_{\mathcal{L}})$ to $\mathrm{Spin}^{c}(M_{\mathcal{L}})$ is given by 
\begin{align}
    \mathrm{Spin}(M_{\mathcal{L}}) & \rightarrow \mathrm{Spin}^{c}(M_{\mathcal{L}}) \endline
    \mathfrak{s} & \mapsto 2 L \mathfrak{s}.
\end{align}
Suppose we fix a $\mathrm{Spin}$-structure $\mathfrak{s}_{0}$ on $M_{\mathcal{L}}$, then any $\mathrm{Spin}^{c}$-structure $b$ on $M_{\mathcal{L}}$ can be written as $$b = 2 L \mathfrak{s}_{0} + 2 h_{b},$$ where $h_{b}\in \Z^{n}/ L \Z^{n} \cong H^{2}(M_{\mathcal{L}},\Z) $.

To describe the TQFT structure of $\hat{Z}$, it is convenient to work with $\hat{Z}^{\Q/\Z}$ invariants. The invariants $\hat{Z}^{\Q/\Z}(M,q)$ are labeled by a pair consisting of a $\mathrm{Spin}$-structure on $M$ and a first cohomology class in $H^{1}(M,\Q/\Z)$. For the three-manifold $M_{\mathcal{L}}$, the invariant $\hat{Z}^{\Q/\Z}(M_{\mathcal{L}})$ is related to $\hat{Z}(M_{\mathcal{L}})$ by:
\be 
\hat{Z}^{\Q/\Z}_{(\mathfrak{s},\beta)}(M_{\mathcal{L}},q) = \hat{Z}_{ 2 L \mathfrak{s} + 2 Bk(\beta) }(M_{\mathcal{L}},q),
\ee 
where $Bk$ denotes the Bockstein homomorphism $Bk: H^{1}(M_{\mathcal{L}},\Q/\Z) \rightarrow H^{2}(M_{\mathcal{L}},\Z)$ associated with the short exact sequence $0 \rightarrow \Z \rightarrow \Q \rightarrow \Q / \Z \rightarrow 0$. In the case of three-manifold $M_{\mathcal{L}}$, the Bockstein homomorphism takes the explicit form:
\begin{align}
    Bk: H^{1}(M_{\mathcal{L}},\Q/\Z) \cong (L^{-1}\Z^{n}) / \Z^{n} &\rightarrow \Z^{n} / L \Z^{n} \cong H^{2}(M_{\mathcal{L}},\Z) \endline
    \beta &\mapsto L \beta.
\end{align}
Thus, we may write:
\be 
\hat{Z}^{\Q/\Z}_{(\mathfrak{s},\beta)}(M_{\mathcal{L}},q) = \hat{Z}_{2L( \mathfrak{s}+ \beta)}(M_{\mathcal{L}},q).
\ee
The label $(\mathfrak{s}_{1},\beta)$ is equivalent to the label $ (\mathfrak{s}_{2},\beta + \mathfrak{s}_{1}- \mathfrak{s}_{2} )$. Let us unpack what $\beta + \mathfrak{s}_{1}- \mathfrak{s}_{2}$ means. The difference $\mathfrak{s}_{1}- \mathfrak{s}_{2}$ can be thought of as an element of $H^{1}(M_{\mathcal{L}},\Z_{2})$, since the set of $\mathrm{Spin}$-structure is afiinely isomorphic to $H^{1}(M_{\mathcal{L}},\Z_{2})$. In fact, $\mathfrak{s}_{1}- \mathfrak{s}_{2}$ can be thought of as an element of $H^{1}(M_{\mathcal{L}},\Q/\Z)$, since there is a natural inclusion $H^{1}(M_{\mathcal{L}},\Z_{2}) \subset H^{1}(M_{\mathcal{L}},\Q/\Z)$. Therefore, it makes sense to add $\mathfrak{s}_{1}- \mathfrak{s}_{2}$ to $\beta \in H^{1}(M_{\mathcal{L}},\Q/\Z)$. Because of this equivalence we will write the pair $(\mathfrak{s}_{1},\beta)$ as a single element $ \mathbf{a} =\mathfrak{s}_{1} + \beta \in (\frac{1}{2} L^{-1}\Z^{n} )/ \Z^{n}  $. Concretely, the wavefunction labels $(\mathfrak{s},\beta)$ may be written as a single element in the space,
$$(\mathfrak{s},\beta) \longrightarrow \mathbf{a} \in \mathcal{A} = \bigg( (\frac{1}{2} L^{-1}\Z^{n} )/ \Z^{n} \bigg)/\sim$$
$$(\mathfrak{s}_{1},\beta)  \sim (\mathfrak{s}_{2},\beta + \mathfrak{s}_{1}- \mathfrak{s}_{2}) .$$
Note that the distinction between $\hat{Z}^{\Q/\Z}$ and $\hat{Z}$ is notational. They represent the same physical invariants with different labeling conventions, and we will suppress this distinction in our subsequent analysis.

%%%%%%%%%%%%%%%%%%%%%%%%%%%%%%%%%%%%%%%%%%%%%%%%%%%%%%%%%%

\subsection{Rules of \texorpdfstring{$H^{1}(\cdot, \Q/\Z)$}{H1(.,Q/Z)} decorated \texorpdfstring{$\mathrm{Spin}$}{Spin} TQFT}\label{tqftrules}
The $\hat{Z}$ invariants are decorated by $Spin^c$-structures. In this paper, we will show that they are computed by a $H^{1}(\cdot, \Q/\Z)$ decorated $\mathrm{Spin}$ TQFT. Here, we define what we mean by this. Below, $K$ will denote some background field, but for our discussions, in the remainder of the paper, we take $K$ to be the algebraic closure of the field of formal Laurent series in $q$, $\mathbb{C}((q)) $, so that $K =\mathbb{C}((q))^{alg}$. 
\begin{definition} We say that $Z$ is a decorated, oriented 3d $\mathrm{Spin}$ TQFT decorated by $H^{1}(\cdot,\Q/\Z)$ if $Z$ is the following collection of data subject to the following rules: 
\begin{enumerate}
    \item \textbf{State spaces.} To every oriented 2-manifold $\Sigma$ (possibly with punctures), $Z$ assigns a (possibly infinitely generated) $K$-module $Z(\Sigma)$.
    \item \textbf{Disjoint unions.} For disjoint surfaces, the $K$-module factorizes as, $Z(\Sigma_1 \bigsqcup \Sigma_2) = Z(\Sigma_1) \otimes Z(\Sigma_2)$.
    \item \textbf{Empty surface.} The $K$-module associated with the empty surface is the ground field. That is $Z(\Sigma) = K$.
    \item \textbf{Grading.} $Z(\Sigma)$ is graded by $H^{1}(\Sigma, \Q/\Z)$.
    \item \textbf{States.} To the tuple $(M,\mathfrak{s},h, \varphi)$, where $M$ is a three-manifold with boundary, $\mathfrak{s}$ is a $\mathrm{Spin}$-structure on $M$, $h \in H^{1}(M,\Q/\Z) $, and $\varphi$ is a choice of framing on $M$, we assign a vector (equivalently wavefunction), $$Z(M,\mathfrak{s},h, \varphi) = \ket{M,\mathfrak{s},h, \varphi} \in Z(\partial M).$$
    \item \textbf{Grading of states.} $Z(M,\mathfrak{s},h, \varphi)$ belongs to the $(h\vert_{\partial M} + \varphi\vert_{\partial M}^{*}\circ i^{*}(\mathfrak{s}))$-graded subspace of $Z(\partial M)$, where $i$ is the inclusion map from the frame bundle of $\partial M$ to the frame bundle of $M$. That is, the grading of the state is shifted by the $\mathrm{Spin}$ contribution.
    \item \textbf{Mapping class group action.} To every oriented 2-manifold $\Sigma$, $Z$ assigns a representation of $MCG(\Sigma)$ on $Z(\Sigma)$: $$Z(\cdot):MCG(\Sigma) \rightarrow GL(Z(\Sigma)).$$ In fact, $$Z(M,\mathfrak{s}',h', \gamma \cdot \varphi) = Z(\gamma) \cdot Z(M,\mathfrak{s},h, \varphi), $$ 
    where $\mathfrak{s}', h'$ are the corresponding $\mathrm{Spin}$-structure and first cohomology element under the attachment of the mapping cylinder $M_\gamma$ to $M$. Due to this property, we will often ignore the framing of wavefunctions and consider only 0-framed wavefunctions:
    $$\ket{M,\mathfrak{s},h} = \ket{M,\mathfrak{s},h,0}.$$
    \item \textbf{Pairing.} $Z$ is endowed with a bilinear map $\langle \cdot |\cdot \rangle_\Sigma: Z(\Sigma)\times Z(\Sigma) \rightarrow K$. 
    \item \textbf{Gluing rule.} Suppose $M_{1}$ and $M_{2}$ are three-manifolds with boundary. Suppose $\Sigma \subset \partial M_{1}$ and $\Sigma \subset \partial M_{2}$. Suppose $M$ is obtained by gluing $M_{1}$ and $M_{2}$ along $\Sigma$ using a diffeomorphism in mapping class group element $\gamma$. Then 
$$\ket{M, \mathfrak{s}, h, \varphi} = \langle M_{1}, \mathfrak{s}\vert_{M_1}, h\vert_{M_1}|  \gamma | M_{2}, \mathfrak{s}\vert_{M_2}, h\vert_{M_2}\rangle_\Sigma,$$
where $\mathfrak{s}\vert_{M_i}$ and $h\vert_{M_i}$ are restrictions of the $\mathrm{Spin}$-structure $\mathfrak{s}$ and the cohomology element $h$ to $M_{i}$.  
\end{enumerate}
\end{definition}

In the following sections, we will explicitly describe various elements of the $H^{1}(\cdot, \Q/\Z)$ decorated TQFT that computes $\hat{Z}$. More precisely, we will define the vector space associated with the torus (and disjoint unions thereof), identify the wavefunction associated with the solid torus, and show that the above rules uniquely determine a large number of amplitudes and wavefunctions, including those explored in \cite{GM} and \cite{GPPV}. Additionally, we will use the decorated TQFT structure to derive several novel rational surgery formulas for $\hat{Z}$ and satellite formulas for $F_K$.

%%%%%%%%%%%%%%%%%%%%%%%%%%%%%%%%%%%%%%%%%%%%%%%%%%%%%%%%%%%%%%%%%%%%%%%%%%%%%%%%

\section{\texorpdfstring{$\Q$}{Q}-extended Quantization of \texorpdfstring{$SL(2,\C)$}{SL(2,C)} Chern--Simons Theory}\label{torus_hilbertspace}
The quantization of $SL(2,\C)$ Chern--Simons theory represents a fundamental challenge in theoretical physics, serving as a prototype for the quantization of gauge theories with non-compact gauge groups. Non-compact gauge groups introduce novel quantization difficulties that remain poorly understood in general. The $\hat{Z}$ invariant of three-manifolds is believed to provide a non-perturbative quantization of complex Chern--Simons theory. In this section, we start with the phase space of complex Chern--Simons theory on the torus and construct from it the vector space associated with the torus in the $\hat{Z}$-TQFT.

The moduli space of the $SL(2,\C)$ flat connections on $T^{2}$ is given by,
\be 
\mathcal{M}_{\text{Flat}}(SL(2,\C),T^{2}) =
\frac{\C^{\times}\times \C^{\times} \text{ with } \mathbb{CP}^{1} \text{ attached at points } (\pm1,\pm 1) }{ \Z_{2} },
\ee 
where $\Z_{2}$ is the Weyl group of $SL(2,\C)$. We will only consider the abelian flat connections and ignore the $\mathbb{CP}^{1}$s attached at $(\pm1,\pm 1)$. Thus, the phase space of the $SL(2,\C)$ Chern--Simons theory on $T^{2}$ can be taken as $\C^{\times}\times\C^{\times} / \Z_{2}$. If we coordinatize $\C^{\times}\times\C^{\times}$ by $X$ and $Y$, the symplectic form on $\C^{\times}\times\C^{\times}$ is given by $\omega_{T^{2}} = \dd \log{X} \wedge \dd \log{Y}$. Following the canonical quantization prescription, we promote the classical variables to operators satisfying:
$$[\log{X},\log{Y}] = -2\hbar.$$
This immediately implies that the holonomy operators themselves satisfy a $q$-deformed commutation relation:
$$XY = q^{-2} Y X,$$
where $q=e^{\hbar}$. One might therefore expect that the algebra of observables in $SL(2,\C)$ Chern--Simons theory is given by, 
\be 
\mathcal{O}_{\Z}=\frac{\langle X, Y \rangle}{\{ X Y = q^{-2} Y X \} } .
\ee
However, in constructing the $\hat{Z}$-TQFT, we find that extending the algebra to include rational exponents of $X$ and $Y$ leads to well-behaved functorial cutting and gluing rules. Thus, the following algebra is a better approximation of the algebra of observables in $SL(2,\C)$ Chern--Simons theory.
\be 
\mathcal{O}_{\Q}=\frac{\langle X^{\lambda}, Y^{\mu}\vert \lambda, \mu \in \Q \rangle}{\{ X^{\lambda} Y^{\mu} = q^{-2\lambda \mu} Y^{\mu} X^{\lambda} \} } .
\ee 
We can further extend the algebra by including real exponents. 
\be 
\mathcal{O}_{\R}=\frac{\langle X^{\lambda}, Y^{\mu}\vert \lambda, \mu \in \R \rangle}{\{ X^{\lambda} Y^{\mu} = q^{-2\lambda \mu} Y^{\mu} X^{\lambda} \} }.
\ee
Extending $\mathcal{O}_{\Z}$ to $\mathcal{O}_{\Q}$ or $\mathcal{O}_{\R}$ corresponds to analytically continuing in color.

At this stage, it is unclear whether the correct algebra of operators is given by $\mathcal{O}_{\Q}$ or $\mathcal{O}_{\R}$. A deeper understanding of $\hat{Z}$, particularly for manifolds with first Betti number greater than zero, is needed to resolve this. We will mainly work with the algebra $\mathcal{O}_{\Q}$, but most of our discussion applies to both algebras. 

Having established the algebra, let us now look at some structural properties of these algebras. The Weyl symmetry acts on the generators of the algebras by sending $X \rightarrow X^{-1}$ and $Y \rightarrow Y^{-1}$. This $\Z_{2}$-action splits the algebras into symmetric and antisymmetric subalgebras.
\be 
\mathcal{O}_{R} = \mathcal{O}_{R}^{+} \oplus \mathcal{O}_{R}^{-},
\ee 
where $R= \R $ or $\Q$. The algebra $\mathcal{O}_{\Z}^{+}$ is isomorphic to the Kauffman Skein algebra of $T^{2}$ ($KBS_{q}(T^{2})$), the $\mathfrak{sl}_{2}$ Skein Algebra of $T^{2}$ ($SkAlg_{\mathfrak{sl}_{2}}(T^{2})$), and the Spherical Double Affine Hecke Algebra of type $A_{1}$ at $t=1$. This suggests that it might be interesting to look at analogous extensions of the above three algebras isomorphic to $\mathcal{O}_{\Z}^{+}$. 

Along with the $\Z_{2}$ grading induced by the Weyl group action, the algebra $\mathcal{O}_{\Q}$ also has a $(\Q/\Z)^{2}$ grading.
\begin{align}
 \mathcal{O}_{\Q} &= \bigoplus_{(\alpha,\beta)\in (\Q/\Z)^{2} }  y^{\beta} x^{\alpha} \mathcal{O}_{\Z}.
\end{align}
The Weyl group action takes the $(\alpha,\beta)$ graded subspace to the $(-\alpha,-\beta)$ graded subspace. Thus the Weyl group action splits the vector space $\mathcal{V}_{(\alpha,\beta)} = y^{\beta} x^{\alpha}  \mathcal{O}_{\Z} \oplus y^{-\beta} x^{-\alpha}  \mathcal{O}_{\Z} $, into symmetric and anti-symmetric parts.
\be \label{vecab}
\mathcal{V}_{(\alpha,\beta)} = \mathcal{V}^{+}_{(\alpha,\beta)} \oplus \mathcal{V}^{-}_{(\alpha,\beta)}.
\ee

%%%%%%%%%%%%%%%%%%%%%%%%%%%%%%%%%%%%%%%%%%%%%%%%%%%%%%%%%%%%%%%%%%

\subsection{Vector space associated with \texorpdfstring{$T^{2}$}{T2}}
Rules of quantum mechanics tell us that the Hilbert space should be a representation of the algebra of observables. A quantization scheme takes the classical phase space of observables as input and produces a Hilbert space associated with the surface. These quantization schemes typically yield a Hilbert space isomorphic to an appropriate space of functions on a Lagrangian submanifold of the phase space. In our case, the quantization schemes produce the space of functions in the variable $X$. However, we find that this vector space must be appropriately enlarged in the following way. 

As above, we let $R = \R$, $\Q,$ or $\C$ (mainly $\Q$) and
$$\mathcal{O}_R =\frac{ \langle X^m, Y^n\rangle_{n,m \in R}}{(X^m Y^n - q^{-2nm} Y^n X^m) },$$
denote the $R$-extended quantum torus. Let $\mathbb{C}_q$ be the algebraic closure of $\mathbb{C}((q))$, the field of formal Laurent series in $q$. We can consider the $\mathcal{O}_R$-module, 
$$\Hat{\mathcal{O}}_R = \mathbb{C}_q \langle \{  Y^n X^m\}_{n,m\in R}\rangle.$$
In the above module, elements are formal sums $\sum_{n,m\in R} c_{n,m} Y^n X^m$, where we allow $c_{n,m} \in \mathbb{C}_q$ to be non-zero for infinitely many values of $n,m \in R$. The Weyl group of $\mathfrak{sl}_2$ has a non-trivial automorphism on $\Hat{\mathcal{O}}_R$,
$$w: X\mapsto X^{-1}  \hspace{7mm}Y\mapsto Y^{-1}.$$

We begin the construction of our decorated $\mathrm{Spin}$ TQFT by declaring the vector space associated with the torus to be, 
\be \mathcal{H}(T^2) = \Hat{\mathcal{O}}_\Q . \ee

\begin{remark}
    Intuitively, we choose to see this space as the vector space generated by Wilson lines wrapping $X$ and $Y$ cycles of the torus, in Verma module representations of $\mathfrak{sl}_2$ with highest weights being the $X$, $Y$ exponents.

    Concretely, we think of $\mathcal{H}(T^2)$ as the vector space over $\C_q$ spanned by the vectors, $Y^n X^m \ket{0} = \ket{n,m}$, with $m, n \in \Q$.
\end{remark}

%%%%%%%%%%%%%%%%%%%%%%%%%%%%%%%%%%%%%%%%%%%%%%%%%%%%%%%%%%%%%%%%%%

\subsection{Bilinear Form on \texorpdfstring{$\mathcal{H}(T^2)$}{H(T2)}}
We wish to define closed 3-manifold invariants as partition functions of a decorated TQFT. To this end, we endow the vector space associated with the torus $\mathcal{H}(T^2)$ with a bilinear form (which we will refer to as an ``inner product'' in the physics sense). We will now define this inner product on $\mathcal{O}_{\Q}$, which can then be extended to a suitable subspace of $\mathcal{H}(T^2) = \Hat{\mathcal{O}}_\Q$.

Any element in the algebra $\mathcal{O}_{\Q}$ is a linear combination of terms of the form $X^{\lambda_{1}} Y^{\mu_{1}} \cdots X^{\lambda_{r}} Y^{\mu_{r}}$ for some positive integer $r$. Using the $q$-commutation relation, we can bring all the $Y$s to the left to get,
\be 
X^{\lambda_{1}} Y^{\mu_{1}} \cdots X^{\lambda_{r}} Y^{\mu_{r}} = q^{-2\sum_{i=1}^{r} \mu_{i} \sum_{j=1}^{i} \lambda_{j}} Y^{\sum_{i=1}^{r} \mu_{i}} X^{\sum_{j=1}^{r} \lambda_{j}}.
\ee 
We adopt the convention that the operators in $\mathcal{O}_{\Q}$ are written with $Y$ appearing to the left of $X$. For an element $\psi \in \mathcal{O}_{\Q}$ we can express it as a sum,
\begin{align*}
    \psi(X,Y) &=  \sum_{m,n \in \Q}  \psi_{m,n}(q) Y^{n} X^{m}.
\end{align*}
We define a bilinear form on $\mathcal{O}_{\Q}$ via, 
\be 
\langle \psi | \phi \rangle := \sum_{m,n \in \Q} \psi_{-n,-m}(q) \phi_{n,m}(q).
\ee 
This gives us a well-defined bilinear map,
$$\bra{\cdot}\ket{\cdot}: \mathcal{O}_{\Q} \times \mathcal{O}_{\Q} \rightarrow \C_q .$$
Formally, we will also write the inner product of $\psi, \phi \in \mathcal{O}_{\Q}$ as a double contour integral,
\be \label{innerprod}
\langle \psi | \phi \rangle = \oint \frac{\dd X}{2\pi i X} \oint \frac{\dd Y}{2\pi i Y} \psi(X,Y)^\dagger \phi(X,Y),
\ee
where
\be 
\psi(X,Y)^\dagger = \left( \sum_{n,m \in \Q} \psi_{m,n}(q) Y^n X^m \right)^\dagger := \sum_{n,m \in \Q} \psi_{m,n}(q)X^{-m} Y^{-n}.
\ee 

We can justify the notation of integrals as follows. Let $X= e^{i u}$ and $Y= e^{i v}$. The process of q-commuting all $Y$s to the left of all $X$ gives us a map from the algebra $\mathcal{O}_{\Q}$ to the space of functions $\C(q,u,v)$.  
\be 
\mathcal{F}: \mathcal{O}_{\Q} \rightarrow \C(q,u,v).
\ee 
The above map $\mathcal{F}$ can be restricted to the sub-algebras $\mathcal{O}_{\Q}$ or $\mathcal{O}_{\Z}$. The image of $\mathcal{F}$ is the space of functions such that there exists a pair of integers $(p_{u},p_{v})$ such that the function is periodic in $u$ with period $2 \pi p_{u}$ and periodic in $v$ with period $2 \pi p_{v}$. When restricted to $\mathcal{O}_{\Z}$, the image $\mathcal{F} (\mathcal{O}_{\Z})$ is the space of periodic functions in $u$ and $v$ with period $2 \pi$. Suppose $f, g \in \mathcal{O}_{\Q}$ are such that $\mathcal{F}(f), \mathcal{F}(g)$ are periodic in $u$ with period $2 \pi p_{u}$ and periodic in $v$ with period $2 \pi p_{v}$. Then the inner product of $f$ and $g$ is given by,
\be 
(f,g) := \frac{1}{4 \pi^{2} p_{u} p_{v} } \int_{\mathcal{C}_{p_{u}}} \dd u \int_{\mathcal{C}_{p_{v}}} \dd v  \mathcal{F}(f)(-u,-v) \mathcal{F}(g)(u,v),
\ee 
where $\mathcal{C}_{p}$ is a sum of contour from $-\pi p+i \epsilon$ to $\pi p+i \epsilon$ and the contour from $-\pi p - i \epsilon$ to $\pi p -i \epsilon$ for a small positive real number $\epsilon$. The $i \epsilon$ prescription originates from the principal value prescription used in the definition of the $\hat{Z}$-invariant in \cite{GPPV}. 

On $\mathcal{H}(T^2) = \hat{\mathcal{O}}_\Q$, the infinite sums prevent the above map from being well-defined everywhere. Nevertheless, the subset $\mathbb{L}$ of $\mathcal{H}(T^2) \times \mathcal{H}(T^2)$ that converges under this bilinear form will be our main interest. Therefore, when we allude to the bilinear form on $\mathcal{H}(T^2)$, we are in fact referring to the map,
$$\bra{\cdot}\ket{\cdot}: \mathbb{L} \rightarrow \C_q , $$
which is defined in the same way as above.

%%%%%%%%%%%%%%%%%%%%%%%%%%%%%%%%%%%%%%%%%%%%%%%%%%%%%%%%%%%%%%%%%%%%

\subsection{\texorpdfstring{$SL(2,\mathbb{Z})$}{SL(2,Z)} Action}\label{sl2z_section}
The vector spaces associated with $T^{2}$ are endowed with an action of the mapping class group of the torus $SL(2,\Z)$. This $SL(2,\Z)$ action plays a crucial role in surgery formulas for three-manifold invariants computed by a TQFT. We will now discuss the $SL(2,\Z)$ action on $\mathcal{H}(T^{2})$. 

The algebra $\mathcal{O}_{\Q}$ has an automorphism subgroup $SL(2,\Z)$ associated with the mapping class group of $T^{2}$.
Recall that $X,Y$ obey the commutation relation, $XY=q^{-2}YX$, in our conventions. The representation of $SL(2,\mathbb{Z})$ is specified by the generators, which act as:
\begin{align}
  \begin{pmatrix}
1 & 1 \\
0 & 1 
\end{pmatrix} & \leftrightarrow \tau_{+}: \hspace{2mm} X \mapsto q X Y \hspace{4mm} Y\mapsto Y ,  \\
\begin{pmatrix}
1 & 0 \\
1 & 1 
\end{pmatrix} &\leftrightarrow \tau_{-}: \hspace{2mm} X \mapsto X \hspace{4mm} Y\mapsto q^{-1} Y X .
\end{align}
Similarly, the automorphism subgroup acts on the logarithmic generators as follows: if $\gamma = \begin{pmatrix} a & b \\ c & d  \end{pmatrix} \in SL(2,\Z)$,
\be 
\begin{pmatrix} a & b \\ c & d  \end{pmatrix} \begin{pmatrix}
    u \\ v
\end{pmatrix} = \begin{pmatrix}
    a u + b v \\ c u + d v
\end{pmatrix}.
\ee 
\begin{lemma}\label{prlemma}
The automorphism $\gamma$ acts on the basis element $Y^n X^m$, for $m,n \in \Q$ of the algebra $\mathcal{O}_{\Q}$ as follows:
\be 
\gamma (Y^n X^m) = Y^{b m+ d n} X^{a m+ c n} q^{- a b m^{2}- 2 b c m n - c d n^{2}}.
\ee 
\end{lemma}
\begin{proof}
Note that for $\alpha, \beta \in \Q$ using the Baker-Campbell-Hausdorff formula we can show that $Y^{\beta} X^{\alpha} = q^{\alpha \beta} e^{i \beta v + i \beta u}$. Under the automorphism described above $Y^n X^m$ goes to, 
\begin{align}
    \gamma(Y^n X^m) &=\gamma( q^{m n}e^{ i n v+ i m u}) = q^{m n} e^{ i n ( c u + d v)+ i m(a u+ b v)} =q^{m n} e^{ i(b m + d n)v+ i( a m + c n) u}  \endline \gamma(Y^n X^m) &= Y^{b m+ d n} X^{a m+ c n} q^{- a b m^{2}- 2 b c m n - c d n^{2}}. 
\end{align} 
\end{proof}
For example, the generators of $SL(2,\Z)$, $\tau_{+} = \begin{pmatrix}
1 & 1 \\
0 & 1 
\end{pmatrix}$ and $ \tau_{-} = \begin{pmatrix}
1 & 0 \\
1 & 1 
\end{pmatrix} $ act on the generators of the algebra $\mathcal{O}_{\Q} $ as,
\begin{align}
    \tau_{+}( Y^{\mu} X^{\lambda} ) &= q^{-\lambda^{2}}Y^{\lambda+\mu} X^{\lambda} \\
    \tau_{-}( Y^{\mu} X^{\lambda} ) &= q^{-\mu^{2}}Y^{\mu} X^{\lambda + \mu }.
\end{align}

On the vector space associated with the torus $\mathcal{H}(T^2) = \hat{\mathcal{O}}_\Q$, this $SL(2,\Z)$ action obviously extends by declaring the action on each generator $Y^\mu X^\nu$ is as in Lemma \ref{prlemma}. It immediately follows that this representation is faithful. 

\begin{proposition}
    The representation of $MCG(T^2) = SL(2,\Z)$ on $\mathcal{H}(T^2)$, specified by the action on the basis of $\mathcal{H}(T^2)$ in Lemma \ref{prlemma} is faithful.
\end{proposition}

%%%%%%%%%%%%%%%%%%%%%%%%%%%%%%%%%%%%%%%%%%%%%%%%%%%%%%%%%%%%%%%%%%%%%%%%%%%%%%%%

\section{Building Blocks in \texorpdfstring{$\hat{Z}$}{Z}-TQFT} 

To describe the $\hat{Z}$-invariants of general three-manifolds, it is useful to first identify a set of elementary ingredients whose wavefunctions can be used to produce $\hat{Z}$ for a large class of three-manifolds. In this section, we develop such building blocks. We begin by analyzing the vector associated with the solid torus, the simplest three-manifold with torus boundary. Next, we consider tree-link complements, using which we reproduce the formula for $\hat{Z}$ of plumbed manifolds from \cite{GPPV}. We also give a neat formula for $\hat{Z}$ for a Seifert manifold over $S^{2}$. Finally, we discuss gluing rules for knot complements, which provide an alternative construction for a large class of closed three-manifolds.

\subsection{Solid Torus Wavefunction} \label{solid_torus_sec}

One of the basic building blocks of three-manifolds is the solid torus $\mathbb{S} = S^1\times D^2$. In the $\hat{Z}$-TQFT, the corresponding wavefunction is realized as a vector in the vector space associated with the torus. Recall that since $H^1(\mathbb{S},\Q/\Z) = \Q/\Z$, the wavefunctions are labeled by $h\in \Q/\Z$ and $\mathfrak{s} \in  \mathrm{Spin}(\mathbb{S})$,
$$\ket{\mathbb{S},h,\mathfrak{s}}\in \mathcal{H}(T^2).$$
By the rule regarding grading of states (Rule 6) from Section \ref{tqftrules}, we can assume the wavefunctions to be of the form, 
\be \label{solidtorus_decorated_sectors} \ket{\mathbb{S},h,\mathfrak{s}} =  \sum_{n,m\in \Z} c^{h,\mathfrak{s}}_{n,m}(q) \cdot Y^{n+h+\mathfrak{s}(\m)}X^{m+\mathfrak{s}(\ell)} . \ee 
Here, we have chosen the convention that under,
$$H^1(\mathbb{S},\Q/\Z) \rightarrow H^1(T^2,\Q/Z)$$
$$h\mapsto (h,0). $$
$Y$ takes the place of the meridian variable. As reviewed in Section \ref{spinstructures_review}, the solid torus $\mathbb{S}$ admits two $\mathrm{Spin}$-structures, $\mathfrak{s}_0$ and $\mathfrak{s}_1$, which, when evaluated on the meridian and longitude, yield, 
\begin{align*}
    \mathfrak{s}_1 (\m)& = \frac{1}{2} & \mathfrak{s}_1 (\ell) &=\frac{1}{2}, \\
    \mathfrak{s}_0 (\m) & = 0 & \mathfrak{s}_0 (\ell) &=\frac{1}{2}.
\end{align*}

Now, we will deduce this wavefunction through topological invariance. First, we note that $\ket{\mathbb{S},h,\mathfrak{s}}$ should be annihilated by the A-polynomial $(Y-1)$. 
\be
(Y-1)\ket{\mathbb{S},h,\mathfrak{s}} =  \sum_{n,m\in \Z} (c^{h,\mathfrak{s}}_{n-1,m}(q)-c^{h,\mathfrak{s}}_{n,m}(q)) \cdot Y^{n+h+\mathfrak{s}(\m)}X^{m+\frac{1}{2}} = 0 ,
\ee 
which implies $c^{h,\mathfrak{s}}_{n-1,m}(q)= c^{h,\mathfrak{s}}_{n,m}(q)$ or in other words $c^{h,\mathfrak{s}}_{n,m}(q)$ is independent of $n$ and we can write $c^{h,\mathfrak{s}}_{n,m}(q) = c^{h,\mathfrak{s}}_{m}(q)$. 
Also, note that $\mathbb{S}$, $h$, and $\mathfrak{s}$ are invariant (up to framing) under $\tau_+^{2}$. Therefore we require $\tau_{+}^{2 a} \ket{\mathbb{S},h,\mathfrak{s}} = q^{f(2a, h,\mathfrak{s})} \ket{\mathbb{S},h,\mathfrak{s}} $, where $q^{f(2a, h,\mathfrak{s})}$ is the framing anomaly. Since $\tau_{+}^{2 a} \ket{\mathbb{S},h,\mathfrak{s}} = \tau_{+}^{2} \cdots \tau_{+}^{2}\ket{\mathbb{S},h,\mathfrak{s}}$, we have $f(2a, h,\mathfrak{s}) = a f(2, h,\mathfrak{s}) $. The action of $\tau_{+}^{2 }$ on $\ket{\mathbb{S},h,\mathfrak{s}}$ is given by
\begin{align}
 \tau_{+}^{2} \ket{\mathbb{S},h,\mathfrak{s}} = &    \sum_{m,n \in \Z } c^{h,\mathfrak{s}}_{m}(q) q^{- \frac{(2m+1)^{2}}{4}}  Y^{n+h+\mathfrak{s}(\m) + 2 m + 1 }X^{m+\frac{1}{2}} = \sum_{m,n \in \Z } c^{h,\mathfrak{s}}_{m}(q) q^{- \frac{(2m+1)^{2}}{4}}  Y^{n+h+\mathfrak{s}(\m)  }X^{m+\frac{1}{2}}.
\end{align}
Therefore, 
\be 
 c^{h,\mathfrak{s}}_{m}(q) (q^{f(2, h,\mathfrak{s})} - q^{- \frac{(2m+1)^{2}}{4}}) =  0,
\ee 
for all $m \in \Z$, and $h$, and $\mathfrak{s}$. Thus for all $m \in \Z$, $c^{h,\mathfrak{s}}_{m}(q)$ can be non-zero only if $q^{f(2, h,\mathfrak{s})} = q^{- \frac{(2m+1)^{2}}{4}}$. Equivalently, this requires $ m = -\frac{1}{2} \pm \sqrt{-f(2, h,\mathfrak{s}) } $, so $f(2, h,\mathfrak{s})$ must be such that $ -\frac{1}{2} \pm \sqrt{-f(2, h,\mathfrak{s}) }$ is an integer and $m$ is equal to that integer. Hence, $\ket{\mathbb{S},h,\mathfrak{s}}$ is of the following form, 
\be 
\ket{\mathbb{S},h,\mathfrak{s}}^{(m)}   =   \sum_{n \in \Z} Y^{n+h+\mathfrak{s}(\m)  } ( c^{h,\mathfrak{s}}_{-}(q) X^{- (m + \frac{1}{2})} + c^{h,\mathfrak{s}}_{+}(q) X^{(m + \frac{1}{2})}    ) .
\ee
We also expect that the solid torus wavefunction is either symmetric or anti-symmetric under the Weyl group action. This implies 
\be 
\ket{\mathbb{S},h,\mathfrak{s}}^{(m, \pm )}   =  \mathcal{N}^{h,\mathfrak{s}}(q) \sum_{n \in \Z} Y^{n+h+\mathfrak{s}(\m)  } (  X^{- (m + \frac{1}{2})} \pm  X^{(m + \frac{1}{2})}    ), 
\ee 
for some function $\mathcal{N}^{h,\mathfrak{s}}(q) $ and a non-negative integer $m$. Among these possibilities, we select the anti-symmetric wavefunction under Weyl action. Furthermore, we take $m=0$, the smallest non-negative integer, and fix the normalization $\mathcal{N}^{h,\mathfrak{s}}(q) = 1$. Other choices of the non-negative integer $m$ correspond to the wavefunction of the solid torus with Wilson line insertions. With these choices, we define our solid torus wavefunction as follows:
\begin{definition}
In the $\hat{Z}$-TQFT, the solid torus $\mathbb{S} = S^1 \times D^2$ is assigned the following wavefunction (or vector)
    $$\ket{\mathbb{S},h,\mathfrak{s}} = \sum_{n\in \Z}Y^{n+h +\mathfrak{s}(\m)}(X^\frac{1}{2}-X^{-\frac{1}{2}} ).$$
\end{definition}

We note here that a different choice of this wavefunction that satisfies the constraints imposed by the rules in Section \ref{tqftrules} would also give a consistent decorated $\mathrm{Spin}$ TQFT. For instance, the `refined' $\hat{\hat{Z}}$ invariant of \cite{AJK} would be given by setting $c^{h,\mathfrak{s}}_{\pm}(q) = t^{\mp} $ instead (where $t$ is some generic complex variable). It is worth noting that due to the framing anomaly discussed above, our theory will only produce topological invariants up to an overall $\pm q^c$ factor with $c\in \mathbb{Q}$. 

As noted in Section \ref{tqftrules}, the solid torus above (equivalently, the unknot complement) is 0-framed. To frame it by an arbitrary rational number $\frac{p}{r}$, the natural next step is to consider how this vector behaves under action by elements of the mapping class group $\gamma_{p/r}\in SL(2,\mathbb{Z})$. Suppose it is of the form, $\gamma_{p/r}= \begin{pmatrix}
b & a \\
p & r 
\end{pmatrix}$, then the wavefunctions of the $\frac{p}{r}$ framed solid torus $\mathbb{S}_{p/r}$ are given by, $$\ket{\mathbb{S}_{p/r} ,\gamma_{p/r}h,\gamma_{p/r}\mathfrak{s}} = \gamma_{p/r}\ket{\mathbb{S},h,\mathfrak{s}} = \sum_{\substack{n \in h + \mathfrak{s} + \Z  \\ \epsilon = \pm 1}}  \epsilon q^{- \tfrac{a b}{4} - \epsilon a p n - p r n^{2} } Y^{\tfrac{\epsilon a }{2} + r n} X^{\tfrac{ \epsilon b}{2} + p n } .$$
as dictated by Rule 5 of Section \ref{tqftrules}.

\subsection{Links and Plumbings}
Thus far, we have studied the structure of $\mathcal{H}(T^2)$, endowed with the $SL(2,\mathbb{Z})$ action and the decorated $\mathrm{Spin}$ TQFT structure. It turns out this structure also determines the amplitudes of Seifert Manifolds (and more generally Graph/Plumbed Manifolds) and wavefunctions for more general knots and links. Throughout this section, we will show that our formalism above and topological invariance are sufficient to uniquely determine the plumbing formula of \cite{GPPV}. That is, if $M$ is a closed 3-manifold acquired by link surgery described by a weakly negative definite plumbing graph with linking matrix $Q$, then the $\hat{Z}$ invariants are defined by, 
\be \label{plumbingformula}\Hat{Z}_\alpha(M;q) = \pm q^c \prod_v\oint \frac{dx_v}{2\pi i x_v} \left[ \prod_v (x_v^{1/2} -x_v^{-1/2})^{2-\mathrm{deg}( v)} \Theta^{-Q}_\alpha  \right], \ee
where,
$$\Theta^{-Q}_\alpha = \sum_{\ell \in \alpha + 2Q\mathbb{Z}^n} q^{-\frac{( \ell,Q^{-1} \ell)}{4}} x^{\frac{\ell}{2}},$$
and consequently all statements derived from it, including $F_K$ for torus links. This will also allow us to go further and prove more general surgery and satellite formulas in the next section. 

We will build up towards general plumbed manifolds by first considering the Hopf link complement, slowly increasing the complexity of links, and we will derive the wavefunction assigned to tree link complements. Using this wavefunction, the wavefunction assigned to the solid torus and the gluing rule (rule 8) in section \ref{tqftrules}, we will determine the $\hat{Z}$ for plumbed manifolds. 

We will first find the state that corresponds to the Hopf link, $\ket{Hopf}\in \mathcal{H}(T^2)^{\otimes2}$. The most general form allowed by the rules of $H^{1}(\cdot,\Q/\Z)$ decorated $\mathrm{Spin}$ TQFT for $\ket{Hopf}$ is given by,
$$\ket{Hopf, \mathbf{a}} = \sum_{\substack{ n \in (a_{1}, a_{2}) + \Z^{2} \\ m \in (a_{2}, a_{1}) + \Z^{2} }}  H^{(n_1,m_1),(n_2,m_2)} \ket{n_1,m_1}\otimes\ket{n_2,m_2}.$$
The Hopf link complement is topologically equivalent to mapping cylinder $T^{2} \times_{S} [0,1] $, where $S = \begin{pmatrix}
    0 & -1 \\ 1 & 0
\end{pmatrix} \in SL(2,\Z)$. Therefore, gluing a Hopf link complement to a manifold with torus boundary is equivalent to attaching the mapping cylinder $T^{2} \times_{S} [0,1] $. 

Suppose we take a vector $\ket{v} = \sum_{ \substack{n \in a_{2} + \Z \\ m \in - a_{1} + \Z} } c_{n,m} \ket{n,m}  \in \mathcal{H}(T^{2})$ and glue it with Hopf link complement, we get, 
$$ \sum_{\substack{n \in (a_{1}, a_{2}) + \Z^{2} \\ m \in  (a_{2}, a_{1}) + \Z^{2} }} H^{(n_1,m_1),(n_2,m_2)} c_{n_{2}, -m_{2}} \ket{n_{1},m_{1}}. $$
For any $c_{n,m}$ this should be equal to, 
$$ S \ket{v} = \sum_{ \substack{n_{1} \in a_{1} + \Z \\ m_{1} \in a_{2} + \Z }} c_{m_{1},-n_{1}} q^{-2 m_{1} n_{1}} \ket{n_{1},m_{1}},$$
that is,
$$\sum_{\substack{n_{2} \in a_{2}+ \Z \\ m_{2} \in a_{1} + \Z }}  H^{(n_1,m_1),(n_2,m_2)} c_{n_{2}, -m_{2}} =  c_{m_{1},-n_{1}} q^{-2 m_{1} n_{1}}  . $$
In particular, choosing $c_{n,-m} = \delta_{n,k} \delta_{m,\ell}$, we get, 
\be 
H^{(n_1,m_1),(n_2,m_2)} = q^{- 2 m_{1} n_{1}}  \delta_{ n_{1}, m_{2} } \delta_{n_{2}, m_{1}}.
\ee 
Therefore, we can write down $\ket{Hopf}$ as,
$$\ket{Hopf,\mathbf{a}} = {\sum_{n\in \Z^2+\mathbf{a}}} q^{-2n_1 n_2} Y_1^{n_1}Y_2^{n_2} X_1^{n_2} X_2^{n_1} .$$

Now, we shall use the newly found wavefunction for the complement of the Hopf link to construct a more general class of link complements. An intermediate step in this venture will be finding the state, $\ket{H} \in \mathcal{H}(T^2)^{\otimes3}$, associated with the complement of a 3-component link $H$, which is depicted in the figure \ref{Hwavefn}.
\begin{figure}[H] 
\centering
\includegraphics[width=0.3\textwidth]{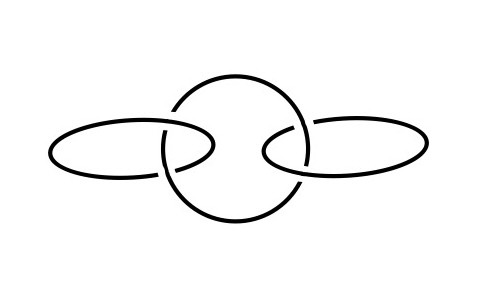}
\caption{Kirby diagram for $\ket{H}\in \mathcal{H}(T^2)^{\otimes3}$. }
\label{Hwavefn}
\end{figure}
Let us assume $\ket{H}$ to be of the most general form:
$$\ket{H,\mathbf{a}} = \sum_{n\in \Z^3+\mathbf{a}} Y_1^{n_1}Y_2^{n_2}Y^{n_3}_3 \cdot H^{n_1, n_2, n_3} (X_1,X_2,X_3).$$
We will also take $1$ to be the index of the central 2-valency vertex, so that the torus corresponding to the boundary of the central unknot is denoted $T^2_1$. Now, we point out that gluing two Hopf links along one component with mapping group element $\tau_-^p$ is topologically equivalent to doing $p$-surgery on the central unknot of $H$ (see Figure \ref{Hopf_gluing}).
\begin{figure}[H] 
\centering
\includegraphics[width=0.7\textwidth]{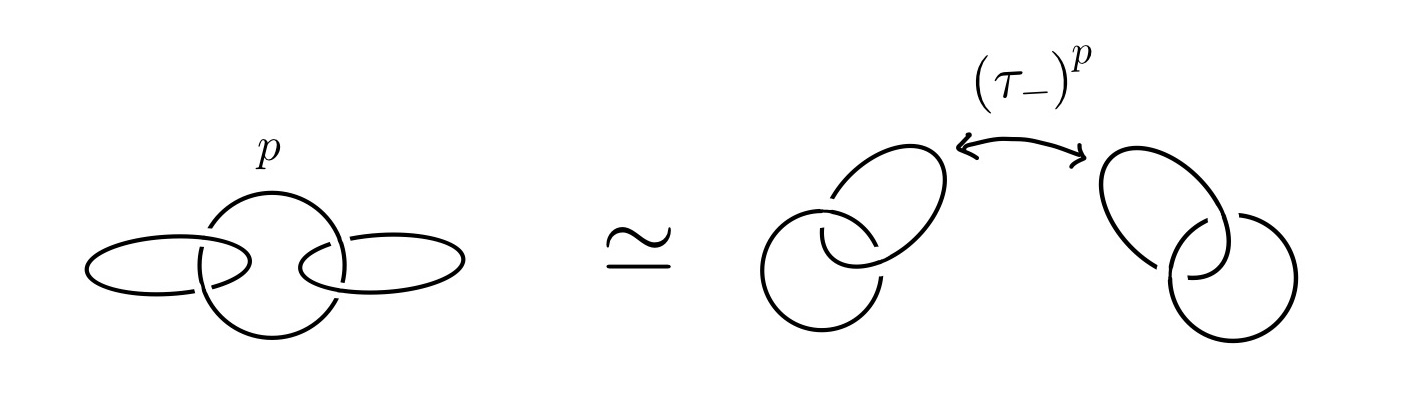}
\caption{The equivalence of $p$-surgery on the central unknot of $H$ and the gluing of two Hopf links.}
\label{Hopf_gluing}
\end{figure}
In equation form, this translates to, 
\be \label{Hopf_gluing_eq}
\bra{\mathbb{S},\mathbf{a}_1}( \tau_-^p s)^\dagger \ket{H,\mathbf{a}}_{T^2_1} = \bra{Hopf, (\mathbf{a}_2,\mathbf{a}_1)}\tau_-^p\ket{Hopf,(\mathbf{a}_1,\mathbf{a}_3)}_{T^2_1}.
\ee
The left-hand side amounts to the quantity,
$$\bra{\mathbb{S},\mathbf{a}_1}( \tau_-^p s)^\dagger \ket{H,\mathbf{a}}_{T^2_1} = \sum_{n\in \Z^3+\mathbf{a}}Y_2^{n_2}Y^{n_3}_3 \cdot q^{-p n_1^2} C.T._X\left[X_1^{n_1 p} (X_1^{\frac{1}{2}}-X_1^{-\frac{1}{2}})H^{n_1, n_2, n_3} (X_1,X_2,X_3)\right].$$
Note that the element $(X_1^{\frac{1}{2}}-X_1^{-\frac{1}{2}})$ is an antisymmetric element of $\mathcal{O}_\Q$, independent of $Y$. Therefore, per appendix \ref{high_rank_expansions}, its inverse, compatible with the Weyl group action, is well-defined and unique. We will denote as $(X_1^{\frac{1}{2}}-X_1^{-\frac{1}{2}})^{-1}$ and for the convenience of calculations, we set,
$$H^{n_1, n_2, n_3} (X_1,X_2,X_3) = (X_1^{\frac{1}{2}}-X_1^{-\frac{1}{2}})^{-1}\Tilde{H}^{n_1, n_2, n_3} (X_1,X_2,X_3).$$
So that, 
\be \label{Hopf_LHS} 
\bra{\mathbb{S},\mathbf{a}_1}( \tau_-^p s)^\dagger \ket{H,\mathbf{a}}_{T^2_1} = \sum_{n\in \Z^3+\mathbf{a}}Y_2^{n_2}Y^{n_3}_3 \cdot q^{-p n_1^2} C.T._X\left[X_1^{n_1 p} \tilde{H}^{n_1, n_2, n_3} \right].
\ee
For the right-hand side, we have, 
\be \label{Hopf_RHS}
\bra{Hopf, (\mathbf{a}_2,\mathbf{a}_1)}s^\dagger \tau_-^p\ket{Hopf,(\mathbf{a}_1,\mathbf{a}_3)}_{T^2_1} = \sum_{n\in \Z^3+\mathbf{a}} Y_2^{n_2} Y_3^{n_3}q^{-2n_1 (n_2+n_3)-pn_1^2 } X_2^{n_1} X_3^{n_1} \cdot \delta_{-n_3, n_2+pn_1}.
\ee
The equation \eqref{Hopf_gluing_eq} therefore tells us that the expressions \eqref{Hopf_LHS} and \eqref{Hopf_RHS} are the same, implying that for all $p\in \Z$, we have,  
$$\sum_{n_{1} \in a_{1} +\Z}C.T._{X_{1}}\left[X_1^{n_1 p} \tilde{H}^{n_1, n_2, n_3}(X_{1}, X_{2}, X_{3}) \right] = \sum_{n_{1} \in a_{1} +\Z} q^{-2n_1 (n_2+n_3)} X_2^{n_1} X_3^{n_1} \delta_{n_2+n_3,- n_1p}.$$
Writing 
$$\tilde{H}^{n_1,n_2,n_3} = \sum_{m_1,m_2,m_3}\tilde{H}^{n_1,n_2,n_3}_{m_1,m_2,m_3} X_1^{m_1} X_2^{m_2} X_3^{m_3}, $$
where the sum over $(m_{1}, m_{2}, m_{3})$ is over $(a_{2}+ a_{3}, a_{1}, a_{1})+\Z^{3}$, the equation above becomes, 
\be \label{prior} 
\sum_{n_1\in \Z+a_1} \tilde{H}^{n_1,n_2,n_3}_{-pn_1,m_2,m_3} = \sum_{n_{1} \in a_{1} +\Z} q^{-2n_1 (n_2+n_3)} \delta_{n_1, m_2} \delta_{n_1, m_3}\delta_{n_2+n_3,- n_1p}. 
\ee
Before we solve this equation, notice that we may take advantage of elements of the mapping class group that leave $H$ invariant. These are easy to spot, as Dehn twists around the meridian of the central unknot should be canceled by Dehn twists around either of the peripheral unknots. More precisely, we have that the elements 
$$\tau_+\otimes \tau_-\otimes 1$$
$$1\otimes \tau_-\otimes \tau_+$$
leave the wavefunction $\ket{H,\mathbf{a}}$ invariant up to framing anomaly and label transformation. Solving for this constraint, we find that topological invariance requires that, 
\begin{align*}
   \tilde{H}^{n_1,n_2,n_3}_{m_1,m_2,m_3} =& q^{2an_1^2}\tilde{H}^{n_1,n_2+am_2,n_3}_{m_1+an_1,m_2,m_3} \\
   \tilde{H}^{n_1,n_2,n_3}_{m_1,m_2,m_3}= &q^{2bn_1^2}\tilde{H}^{n_1,n_2,n_3+bm_3}_{m_1+bn_1,m_2,m_3}.
\end{align*}
Using these recursion relations, one can easily show that equation \eqref{prior} implies, 
$$\sum_{n_1\in \Z+a_1} q^{2pn_1^2}\tilde{H}^{n_1,n_2,n_3}_{0,m_2,m_3} = \sum_{n_{1} \in a_{1} +\Z} q^{2pn_1^2} \delta_{n_1, m_2} \delta_{n_1, m_3}\delta_{n_2+n_3,0}.$$
Matching powers of $q^p$, we find a solution to the $m_1 =0$ case,
$$\tilde{H}^{n_1,n_2,n_3}_{0,m_2,m_3} = \delta_{m_2,n_1}\delta_{m_3,n_1}\delta_{n_2+n_3,0}.$$
The recursion relations then allow us to solve the general case, which yields, 
$$\tilde{H}^{n_1,n_2,n_3}_{m_1,m_2,m_3} = q^{-2m_1n_1}\delta_{m_2,n_1}\delta_{m_3,n_1}\delta_{n_2+n_3,m_1}.$$
So that, 
$$\tilde{H}^{n_1, n_2, n_3}  = q^{-2n_1 (n_2+n_3)} X_1^{n_2+n_3} X_2^{n_1} X_3^{n_1}.$$
Putting everything together, we find:
\be \label{Hform} 
\ket{H,\mathbf{a}} = \sum_{n\in \Z^3+\mathbf{a}} q^{-2n_1(n_2+n_3)} (\prod_{i=1,2,3} Y^{n_i}_i) (X_1^{1/2} - X_1^{-1/2})^{-1} X^{n_2+n_3}_1 X^{n_1}_2X^{n_1}_3 .
\ee

Our next goal is to derive the wavefunction for the ``trinion'' depicted in the figure \ref{trinion}.
\begin{figure}[H] 
\centering
\includegraphics[width=0.5\textwidth]{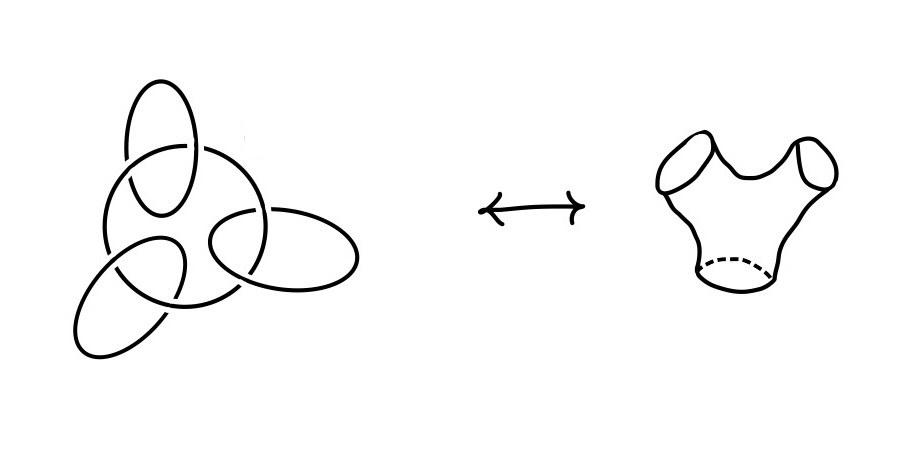}
\caption{Kirby diagram for the trinion wavefunction and its 2d TQFT analogue, the pair of pants. Here, knots with framing indicate glued tori and knots without framing indicate incoming/outgoing states. }
\label{trinion}
\end{figure}
To get the vector corresponding to the trinion $\ket{T} \in \mathcal{H}(T^2)^{\otimes 4}$ (figure \ref{trinion}), or, for that matter, more general tree configurations, we can glue $H$ vectors together using the $S = \begin{pmatrix}
    0&-1\\1&0
\end{pmatrix}$ matrices. 
\begin{figure}[H] 
\centering
\includegraphics[width=0.7\textwidth]{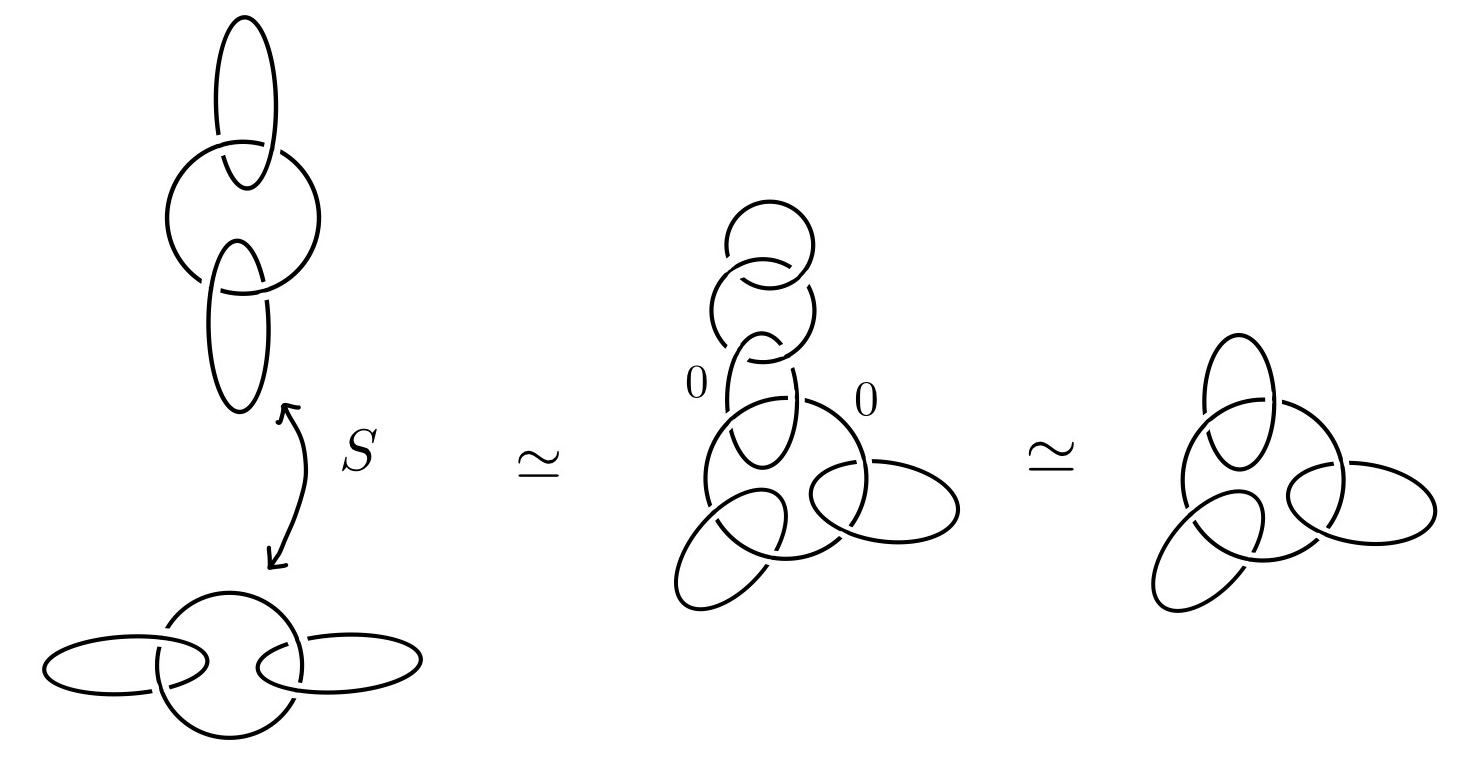}
\caption{Gluing of the 3-manifolds $H$ to yield a 3 valency vertex. Here, the gluing is between the central unknot complement of one $H$ to one of the outermost unknot complements. }
\label{Trinion_gluing}
\end{figure}
In this case, figure \ref{Trinion_gluing} depicts the exact operation needed, and one finds, 
$$\ket{T,\mathbf{a}} = \bra{H,\mathbf{a}_1}S\ket{H,\mathbf{a}_2}_{T^2}.$$
Therefore,
\be \label{trinion_eq} \boxed{\ket{T,\mathbf{a}} = \sum_{\mathbf{n}\in \Z^4+\mathbf{a}} q^{-2n_1(n_2+n_3+n_4)} 
  Y_1^{n_1} Y_2^{n_2} Y_3^{n_3} Y_4^{n_4}(X_1^{1/2} - X_1^{-1/2})^{-2} X_1^{n_2+n_3+n_4} X^{n_1}_2 X^{n_1}_3 X^{n_1}_4.}  \ee

Just as with its 2d analogue, the pair of pants, we can cut and glue the trinion, or more fundamentally, the 3-manifold $H$, arbitrarily many times to acquire more general wavefunctions and amplitudes. Let us define what we mean by this general class of wavefunctions. 
\begin{definition}
    By a \textbf{Tree Link}, we mean a link where every component is the unknot and they are linked as specified by the data of a tree graph such that any two components have either $lk(L_i,L_j)=1$ or $0$, depending on whether their respective vertices in the graph are connected by an edge or not.
\end{definition} 
Then, equations \eqref{trinion_eq} and \eqref{Hform} uniquely specify the form of any tree link, and we have the following result:
\begin{proposition} \label{treelinks_wvfn}
For an $N$-component tree link $L_Q$, the invariant $Z(S^3 \setminus \nu(L_Q),\mathbf{a}) =\ket{L_Q,\mathbf{a}}$ exists and is given by, 
\be \label{treelinks} \boxed{\ket{L_Q,\mathbf{a}} = \sum_{n \in \mathbf{a} + \Z^N } q^{-(Qn,n)} Y^n \prod^N_{i=1}(X_i^{1/2} - X_i^{-1/2})^{1-\mathrm{deg}v_i} X^{Q n}.} \ee
Here $Q$ has $0$ diagonal elements, reflecting that all components are $0$-framed. Additionally, we use the shorthand $Y^n = Y_1^{n_1} Y_2^{n_2}\ldots Y_N^{n_N} $ and $X^n =X_1^{n_1} X_2^{n_2}\ldots X_N^{n_N} $. These are unique and are invariant under Neumann moves (up to an overall framing factor). 
\end{proposition}
\begin{proof}
    Let $Q$ be the linking matrix of a tree graph representing a tree link, and let $Q'$ be the linking matrix corresponding to the tree link acquired from $Q$ by adding a vertex. We proceed by induction by noticing that $\ket{L_{Q'},\mathbf{a}}$ is given by gluing the state $\ket{H,\mathbf{a}_1}$ to the state $\ket{L_Q,\mathbf{a}_2}$ using the $S= \begin{pmatrix}
    0&1\\-1&0
\end{pmatrix}$ mapping class group element of $SL(2,\Z)$. We depict an example of this ``adding a vertex'' operation in figure \ref{Trinion_gluing}. It is easy to check that \eqref{treelinks} satisfies the induction step. Uniqueness is implied by the preceding discussion, and Neumann move invariance will be checked more generally in Proposition \ref{plumbed_link_wavfn}.
\end{proof}

 We can now easily check that integer surgery in every component of a tree link $L_Q$ recovers the plumbing formula \eqref{plumbingformula} of \cite{GPPV}. That is, let $p=(p_1,p_2,\ldots,p_N)$ denote the set of integer framings on any tree link with $N$ components. Let $M$ be the closed 3-manifold resulting from $p$-surgery on $L_Q$. Now, consider the amplitude,
$$Z(M,\mathbf{a})=\bigotimes^N_{i=1}\bra{\mathbb{S},\mathbf{a}_i} \bigotimes^N_{i=1}(\tau^{p_i}_-)^\dagger \ket{L_Q,\mathbf{a}}.$$
We see that,
\begin{align*}
&\bra{\mathbb{S},\mathbf{a}_i} (\tau^{p_i}_-)^\dagger \sum_{n_i\in \mathbf{a}_i + Z}Y^{n_i}_i (X_i^{1/2} - X_i^{-1/2})^{1-\mathrm{deg}v_i} X^{(Q\mathbf{n})_i}\ket{0}   \\
= & \sum_{n_i\in \mathbf{a}_i + \Z} q^{-p_i n_i^2}\oint \frac{dX_i}{2\pi i X_i} (X_i^{1/2} - X_i^{-1/2})^{2-\mathrm{deg}v_i} X_i^{(Q\mathbf{n})_i + p n_i}.
\end{align*}
We let $Q^f = Q + diag(p_1,p_2,\ldots,p_N)$ denote the framed linking matrix associated to this surgery. Note that above we slightly abuse notation by writing $\mathbf{a}$ for both the label for $M$, and its restriction to $S^3\setminus\nu(L_Q)$. Then, we see that:
$$Z(M,\mathbf{a}) = \prod^N_{i=1}\oint \frac{dX_i}{2\pi i X_i} (X_i^{1/2} - X_i^{-1/2})^{2-\mathrm{deg}v_i} \sum_{\mathbf{n}\in \mathbf{a} + \mathbb{Z}^N} q^{-(Q^f \mathbf{n},\mathbf{n})}X^{Q^f\mathbf{n}}.$$

In light of the plumbing formula definition of $\hat{Z}$ in equation \eqref{plumbingformula} from \cite{GPPV}, we have the following result:
\begin{proposition}\label{Zequalsplumbing}
Let $M$ be a closed 3-manifold acquired by surgery on a weakly negative definite plumbing graph with linking matrix $Q^f$. Then, we have, 
$$Z(M,\mathbf{a}) = \hat{Z}_{Bk(\mathbf{a})}(M;q),$$
where, from Section \ref{spinstructures_review}, $Bk(\cdot)$ is the Bockstein map induced by the Bockstein homomorphism, $Bk(\mathbf{a}) = Bk(h,\mathfrak{s}) = Q^f(h+\mathfrak{s})\in Spin^c(M)$.

\end{proposition}

%%%%%%%%%%%%%%%%%%%%%%%%%%%%%%%%%%%%%%%%%%%%%%%%%%%%

\subsubsection{Examples: Seifert Manifolds over \texorpdfstring{$S^2$}{S2}}\label{seifert_section}

We now demonstrate one of the uses of our formalism by computing the $\hat{Z}$ invariants of general Seifert manifolds over $S^2$ analytically, something that would be otherwise tricky using only plumbing graphs with integer coefficients.  

Let $\Sigma_g$ be an oriented Riemann surface of genus $g$. Then we denote by $M\left(b,g; \frac{p_1}{r_1},\frac{p_2}{r_2},\ldots,\frac{p_d}{r_d}\right)$ the Seifert fibration over $\Sigma_g$ with degree $b$ and $d$ exceptional fibers given by the data $\{\frac{p_i}{r_i}\}_{i=1,..,d}$. Though we focus on $\Sigma_g = S^2$ for now, in Section \ref{seifert_sigma_g_section}, we will return to more general Riemann surfaces in some detail. 

The surgery diagram we will be interested in is displayed in the figure \ref{preseifert_manifold}. The wavefunction of the 0-framed link in question is given by,
$$\ket{L_Q,\mathbf{a}} = \sum_{\mathbf{n}\in \Z^{d+1}+\mathbf{a}}q^{-(Q\mathbf{n},\mathbf{n})} Y^\mathbf{n} (X^{\frac{1}{2}}_0 - X^{-\frac{1}{2}}_0)^{1-d} X^{Q\mathbf{n}} = \underset{m_0\in \frac{1}{2}\Z}{\sum_{\mathbf{n}\in \Z^{d+1}+\mathbf{a}}}q^{-(Q\mathbf{n},\mathbf{n})} f^{m_0}_{d+1} Y^\mathbf{n} X^{-m_0}_0 X^{Q\mathbf{n}},$$
where $f^{m_0}_{d+1}$ coefficients of the $X_{0}$ expansion of $(X^{\frac{1}{2}}_0 - X^{-\frac{1}{2}}_0)^{1-d}$, and
$$Q = \begin{pmatrix}
    0&1&1&\\
    1&0&0 &\cdots\\
    1&0&0 &\\
    &\vdots & &\ddots
\end{pmatrix}.$$
\begin{figure}[H] 
\centering
\includegraphics[width=0.4\textwidth]{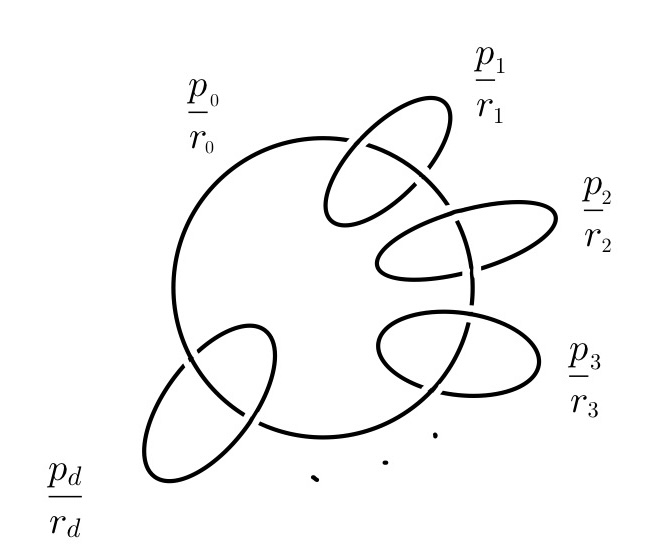}
\caption{A surgery diagram yielding $M\left(b,0; \frac{p_1}{r_1},\frac{p_2}{r_2},\ldots,\frac{p_d}{r_d}\right) $ whenever $\frac{p_0}{r_0} = -b$. Notice that the fiber data $\frac{p_i}{r_i}$ is sometimes written $\frac{r_i}{p_i}$ in the literature, so this figure fixes our conventions.}
\label{preseifert_manifold}
\end{figure}
The $Z(\cdot)$ invariant associated to the resultant 3-manifold decorated by the label $\mathbf{a} = (\mathbf{a}_0,\mathbf{a}_1,\ldots,\mathbf{a}_d)$ is given by,
$$Z(M\left(b,0; \frac{p_1}{r_1},\frac{p_2}{r_2},\ldots,\frac{p_d}{r_d}\right),\mathbf{a})= \left(\otimes^d_{i=0} \bra{\mathbb{S},\mathbf{a}_i}\gamma_{p_i/r_i} ^\dagger \right) \ket{L_Q}.$$
At the vector space factors on the leaves of the graph in figure \ref{preseifert_manifold}, we can repeat the same manipulations in the proof of Theorem \ref{splicing_theorem} and see that,
$$\bra{\mathbb{S},\mathbf{a_i}} \gamma_{p_i/r_i}^\dagger Y^{-n_i}_i X^{n_0}_i\ket{0} \cong  \sum_{\epsilon_i =\pm1} \epsilon_i q^{-n_i^2\frac{p_i}{r_i} -2\frac{ a_i n_i}{r_i}} \delta_{-n_0,\frac{p_i n_i+\alpha_i +\epsilon_i/2}{r_i}} \cdot \delta^\mathbb{Z} \left( \frac{r_i n_0 +\alpha_i+\epsilon_i/2}{p_i}\right),$$
where $(\alpha_0,\alpha_1,\ldots,\alpha_d) = \alpha = Bk(\mathbf{a})$. Similarly, at the central vertex, we have, 
$$\bra{\mathbb{S},\mathbf{a_0}}\gamma_{p_0/r_0}^\dagger Y^{-n_0}_0 X^{k}_0\ket{0} = \sum_{\epsilon_0 =\pm1} \epsilon_0 q^{-n_0^2\frac{p_0}{r_0} -2\frac{a_0 n_0}{r_0}} \delta_{-k,\frac{p_0 n_0+\alpha_0 +\epsilon_0/2}{r_0}} \delta^\mathbb{Z}( n_0 +\frac{r_0+1}{2}),$$
where $k$ is a stand-in variable for $-m_0 + n_1+n_2+\cdots+n_d$. Then, using the Kronecker symbols, we collapse the sums over $n_1,n_2,\ldots,n_d$ by writing,
$$n_i = -\frac{r_i n_0}{p_i}-\frac{\alpha_i}{p_i} -\frac{\epsilon_i}{2p_i},$$
and consequently,
\begin{align*}
    m_0 = & -k +\sum_i n_i = \frac{p_0 n_0+\alpha_0 +\epsilon_0/2}{r_0} - \sum_i \frac{r_i n_0}{p_i}+\frac{\alpha_i}{p_i} +\frac{\epsilon_i}{2p_i} \endline
    = & \eta \cdot n_0  +\xi_\alpha + \frac{\epsilon_0}{2r_0} -\sum_i\frac{\epsilon_i}{2p_i},
\end{align*}
where we have defined,
$$\eta = \frac{p_0}{r_0} - \sum^d_{i=1} \frac{r_i}{p_i}, \hspace{3mm} \xi_\alpha = \frac{\alpha_0}{r_0} - \sum^d_{i=1}\frac{\alpha_i}{p_i}.$$
Computing the $q$ exponent in terms of $n_0$, we find, 
$$ -\eta \cdot n^2_0 - 2n_0 \xi_\alpha + \text{constant terms}.$$
For the Seifert manifold in question, we set $r_0=1$ and the result is then,
\begin{align*}
 & Z(M\left(b,0; \frac{p_1}{r_1},\frac{p_2}{r_2},\ldots,\frac{p_d}{r_d}\right),\mathbf{a}) \endline
 =&  \sum_{n_0\in\mathbb{Z}}q^{-\eta \cdot n_0^2 - 2n_0 \xi_\alpha } \cdot \underset{i=0,1,..,d}{\sum_{\epsilon_i = \pm1}}\epsilon_0 \prod^d_{i=1}\epsilon_i\delta^\mathbb{Z} \left( \frac{r_i n_0 +\alpha_i+\epsilon_i/2}{p_i}\right)  sgn(m)^{d+1} \binom{\frac{d-3+2|m|}{2}}{d-2} \endline  & \hspace{5cm} \times  \delta^\mathbb{Z}(m+\frac{d+1}{2})|_{m = \eta n_0 + \xi_\alpha +\epsilon_0 /2 - \sum^d_{i=1}\epsilon_i/2p_i}.
\end{align*}
Recall that if $A = \prod_{i=1,\ldots,d}p_i$, then the integer homology of $M(b,0;\frac{p_1}{r_1},\ldots,\frac{p_d}{r_d})$ is given by,
\begin{align*}
  H_1 &\simeq \mathbb{Z}^{\delta_0,b}\oplus TorH_1   &  |Tor H_1| &= |A| |b|.
\end{align*}
Then, we have the following result: 
\begin{proposition}
    Let $M= M(b,0;\frac{p_1}{r_1},\ldots,\frac{p_d}{r_d})$ denote the orientable Seifert fibration over $S^2$ of degree $b$ with singular fiber data $(\frac{p_1}{r_1},\ldots,\frac{p_d}{r_d})$. Let $\eta = -b -\sum^d_{i=1}\frac{r_i}{p_i}$ denote its Euler number and suppose $\eta<0$ Fix a $Spin^c$ structure representative, $$(0,\alpha) \in \mathbb{Z}^{\delta_0,b}\oplus TorH_1+\frac{1}{2}, $$
    with $\alpha = (\alpha_0,\ldots,\alpha_d)$ and define $\xi_\alpha = \alpha_0 - \sum_i \frac{\alpha_i}{p_i}$. Then, the $\hat{Z}$ invariant of $M$ is given by, 
    $$\Hat{Z}_{\alpha}(M(b,0;\frac{p_1}{r_1},\ldots,\frac{p_d}{r_d});q) \cong \sum_{n\in \mathbb{Z}}q^{-\eta n^2 -2n \xi_\alpha} \cdot \Psi^n_{\Vec{p},\Vec{r},\alpha,0}$$
    where
    $$\Psi^n_{\Vec{p},\Vec{r},\alpha,0} = \underset{i=1,\ldots,d}{\sum_{\epsilon_i=\pm1}} \prod_i \epsilon_i \delta^\mathbb{Z}\left( \frac{r_i n +\alpha_i+\epsilon_i/2}{p_i}\right)\cdot \mathrm{sgn}(m)^d \binom{\frac{d}{2}-2+|m|}{d-3}\delta^\mathbb{Z}\left(m+\frac{d}{2}\right)\bigg|_{m = \eta n_0 + \xi_\alpha- \sum_i\frac{\epsilon_i}{2p_i}}.$$
\end{proposition}
\begin{remark}
    For $d=3$, the above $q$-series have proved to be interesting modular objects (see, e.g, \cite{3dmod}). It is tempting to speculate that the closed formula above could be used to demonstrate more general modular properties of Seifert manifolds with any number of singular fibers. It would be interesting to examine this hypothesis in the future.
\end{remark} 
\begin{example}[Seifert Manifolds] Some familiar examples include, 
\begin{align*}
  \hat{Z}(M(2,0;-2,-3,-2)) & \cong \begin{aligned}
  &  -1+q^4-q^8+q^{20}-q^{28}+q^{48}-q^{60}+q^{88}-q^{104}+q^{140}-q^{160}+\ldots \\
& -\frac{1}{2} q \left(1-q+q^2-q^5+q^7-q^{12}+q^{15}-q^{22}+q^{26}-q^{35}+q^{40}-q^{51}+q^{57}+\ldots\right) \\
  & 1-q^2+q^{10}-q^{16}+q^{32}-q^{42}+q^{66}-q^{80}+q^{112}-q^{130}+q^{170}-q^{192}+\ldots \\
 & \frac{1}{2}  \left(-1+q-q^2+q^5-q^7+q^{12}-q^{15}+q^{22}-q^{26}+q^{35}-q^{40}+q^{51}-q^{57} +\ldots\right) \\ & \vdots
\end{aligned}   \endline
\hat{Z}(M(2,0;-2,-3,-\frac{3}{2})) &\cong \begin{aligned}
& 1+q^2-q^7-q^{13}+q^{23}+q^{33}-q^{48}-q^{62}+q^{82}+q^{100}-q^{125}+\ldots \\ 
  & \frac{1}{2} q^2 \left(-1+2 q-q^3+q^6-2 q^{10}+q^{15}-q^{21}+2 q^{28}-q^{36}+q^{45}-2 q^{55}+q^{66}+\ldots\right)\\
   & 1-2 q+q^3-q^6+2 q^{10}-q^{15}+q^{21}-2 q^{28}+q^{36}-q^{45}+2 q^{55}-q^{66}+q^{78}+\ldots\\
 &  -\frac{1}{2} q  \left(1+q^2-q^7-q^{13}+q^{23}+q^{33}-q^{48}-q^{62}+q^{82}+q^{100}-q^{125}+\ldots\right)\\  & \vdots 
   \end{aligned} \endline 
   \hat{Z}(M(2,0;-2,-3,-\frac{4}{3})) &\cong \begin{aligned}
& -1+q-q^3-q^{14}+q^{15}+q^{34}-q^{42}+q^{49}-q^{71}+q^{80}-q^{92}-q^{133}+\ldots\\
&   \frac{1}{2} q^2   \left(-1+q+q^2+q^9-q^{22}-q^{39}-q^{44}+q^{53}-q^{67}+q^{78}+q^{85}+\ldots\right)\\
&   1-q+q^3+q^{14}-q^{15}-q^{34}+q^{42}-q^{49}+q^{71}-q^{80}+q^{92}+\ldots.\\
&    q^3 \left(-1- q^{10}+ q^{15}+ q^{35}- q^{85}-   q^{125}+\ldots\right)\\ & \vdots \end{aligned}
\end{align*}
We can also compute any example with any number of fibers, for instance, 
\begin{align*}
    \hat{Z}(M(2,0;-2,-\frac{3}{2},-\frac{4}{3}, \frac{5}{7})) &\cong  \begin{aligned}
     &  \frac{1}{2} \left(4-12 q^{81}+18 q^{202}-22 q^{319}+30 q^{578}-34 q^{769}+37 q^{917}+52q^{1800}+\ldots\right) \\
 & \frac{1}{2} q^{21} \left(-7+11 q^{43}+12 q^{56}-18 q^{187}+23 q^{309}-30 q^{547}-41q^{1064}+\ldots\right) \\
& -4-8 q^{30}-18 q^{212}+22 q^{311}+26 q^{431}-30 q^{598}+41 q^{1130}+48q^{1533}+\ldots \\
& \frac{1}{2} q^{216} \left(18-18 q^5-23 q^{124}+30 q^{390}+36 q^{666}+41 q^{887}-41 q^{898}+\ldots\right). 
    \end{aligned}
\end{align*}
\end{example}

%%%%%%%%%%%%%%%%%%%%%%%%%%%%%%%%%%%

\subsection{Inner Product and Surgery Formulas}
The rules in Section \ref{tqftrules} tell us how surgery should be performed in this theory. That is, suppose $M$ is a 3-manifold acquired through gluing two 3-manifolds, $M_1$ and $M_2$, along a common closed boundary $\Sigma$ with an element of the mapping class group $\gamma \in MCG(\Sigma)$. Then, 
$$\ket{M,h,\mathfrak{s}} = \bra{M_1,h_1,\mathfrak{s}_1}\gamma\ket{M_2,h_2,\mathfrak{s}_2}_\Sigma \in \mathcal{H}\bigg((\d M_1 \setminus \Sigma) \sqcup (\d M_2 \setminus \Sigma)\bigg),$$
where $h_i$ and $\mathfrak{s}_i$ are the restrictions $h\vert_{M_i}$ and $\mathfrak{s}\vert_{M_i}$ respectively. As discussed in section \ref{spinstructures_review}, the labels $(h,\mathfrak{s})$ may be recast into a single label $\mathbf{a}\in \mathcal{A}$, whose image under the Bockstein map is a $Spin^c$-structure:
$$Bk(\mathbf{a}) \in Spin^c(M).$$
As such, we write the gluing rule above as,
$$\ket{M,\mathbf{a}} = \bra{M_1,\mathbf{a}\vert _{M_1}}\gamma\ket{M_2,\mathbf{a}\vert _{M_2}}_\Sigma.$$

In \cite{GM}, a two-variable series for a knot complement in $S^3$ was defined, so long as the knot complement could be described by a plumbing graph with one vertex (or unknot) left unsurgered. We shall call these series ``GM series,'' and for a knot $K$, they will be denoted by $F_K(X,q)$. In light of Proposition \ref{Zequalsplumbing}, it follows that for such a knot complement, we have, 
$$\ket{K,\mathbf{a}} = \sum_{n\in \Z +\mathbf{a}}Y^n F_K(X,q).$$

More generally, if $M$ is a plumbed knot complement in a general 3-manifold $M_3$, $M = M_3 \setminus \nu(K)$, then our results imply that the precise relation between the GM invariants from \cite{GM} (or $F_K$ invariants) of plumbed knot complements, and our wavefunctions is given by, 
$$|M,\mathbf{a} \rangle = \sum_{n\in \mathbb{Z}} Y^{n+\mathbf{a}} \psi^{\mathbf{a}}_{n}(M;X,q) \in \mathcal{H}(T^2),$$
where
$$\psi^{\mathbf{a}}_{n}(M;X,q) = \hat{Z}_{Bk(\mathbf{a})}  (M;X,n,q).$$

Suppose we wish to glue two knot complements in 3-manifolds, $M_1$ and $M_2$, to make a closed 3-manifold, $M$. We need a gluing matrix that reverses the orientation of the meridian ($X\rightarrow X^{-1} $) and keeps the orientation of longitude. We will denote this element as $s$. Our inner product \eqref{innerprod} simply tells us this is:
\begin{align*}
    \ket{M,\mathbf{a}} = \bra{M_1,\mathbf{a}_1} s \ket{M_2,\mathbf{a}_2} = & \oint \frac{dX}{2\pi i X} \oint \frac{dY}{2\pi i Y} \sum_{n_i\in \Z} \psi^\mathbf{a_1}_{n_1}  (M_1;X^{-1},q) Y^{-n_1-\mathbf{a}_1} Y^{n_2+\mathbf{a}_2} \psi^\mathbf{a_2}_{n_2}  (M_2;X^{-1},q)  \\
    = & \oint \frac{dX}{2\pi i X} \sum_{n\in \mathbb{Z}} \psi^{\mathbf{a}_1}_{n}  (M_1;X^{-1},q) \cdot \psi^{\mathbf{a}_1}_{n} (M_2;X^{-1},q).
\end{align*}
In the second line, we have used that $\mathbf{a}_2-\mathbf{a}_1 =0$. Seeing as this is the constant term of the integrand, the result is invariant under $X\leftrightarrow X^{-1}$. Therefore, this is exactly the gluing theorem (90) in \cite{GM} (up to a framing factor).

We now specialize to the case of gluing two knot complements along their torus boundaries. This leads to an explicit formula for the $\hat{Z}$ invariant of the resulting closed three-manifold.

Let us consider a special case of the situation above, where both knot complements are in $S^3$. For clarity, we will write $|S^3 \setminus N(K),\mathbf{a}\rangle = |K,\mathbf{a}\rangle $. Again, we take $\Tilde{\gamma}_{p/r} = \tau_+ ^{k_1}\tau_- ^{k_2}\tau_+ ^{k_3}\cdots $ to be the $SL(2,\mathbb{Z})$ element representing the gluing matrix:
\be
 \begin{pmatrix}
1 & k_1 \\
0 & 1 
\end{pmatrix} \begin{pmatrix}
1 & 0 \\
k_2 & 1  \end{pmatrix} \cdots = \begin{pmatrix}
b & a \\
p & r 
\end{pmatrix}.
\ee
These elements can then be seen as ``adding'' a certain framing to the 0-framed knot complement $|K,\mathbf{a}\rangle$ by acting via $SL(2,\mathbb{Z})$ on the boundary. At this point, it is also convenient to introduce the following notation, 
\begin{definition}
    Let $S$ be any set, then $\delta^S(\cdot)$ denotes the indicator function, 
    $$\delta^S(x)  = \begin{cases} 
       1 & x\in S\\
      0 & \text{otherwise.}
   \end{cases}$$
\end{definition}
Wavefunctions of knot complements are always of the form,
$$\ket{K_i,\mathbf{a}_i} = \underset{n\in \Z+\mathbf{a_i}}{\sum_{m\in \Z+\frac{1}{2}}} f^m_{K_i}(q)  Y^n X^m.$$
Then, the gluing of two knot complements with a general mapping class group element $\gamma = \begin{pmatrix}
    r&a\\
    p&b
\end{pmatrix}$ should be given by:
$$\bra{K_1,\mathbf{a}_1} \gamma \ket{K_2,\mathbf{a}_2}.$$
To this end, we find the following result:
\begin{proposition} \label{splicing_theorem}
Let $M$ be a closed 3-manifold obtained by gluing $S^3 \setminus \nu(K_1)$ and $S^3 \setminus \nu(K_2)$ along $T^2$ with mapping class group element $\gamma = \begin{pmatrix}
    r&a\\
    p&b
\end{pmatrix}$, then, 
    $$Z(M,\mathbf{a}) =\bra{K_1,\mathbf{a}_1}\gamma \ket{K_2,\mathbf{a}_2} = \sum_{m_1,m_2 \in \mathbb{Z}+\frac{1}{2}} f^{m_1}_{K_1}(q)\cdot f^{m_2}_{K_2}(q) \cdot q^{-\frac{1}{p}(b m_{1}^{2}+2 m_{1} m_{2}+ r m_{2}^{2})} \cdot \delta^\Z\left(\frac{m_{1}+m_{2} r}{p}+ \frac{\alpha}{p}  \right).$$
    Whenever the right-hand side converges. Above, we have used that $\delta^\mathbb{Z} (x) = \begin{cases} 
       1 & x\equiv0\hspace{2mm}mod\hspace{2mm}\mathbb{Z}\\
      0 & \text{otherwise.}
   \end{cases}$ and $\alpha = Bk(\mathbf{a})\in Spin^c(M) \cong \frac{p \delta_{0, r \hspace{-0.2cm} \mod 2}}{2}+ \Z/p \Z $. If $M$ is obtained via a weakly negative plumbing, then $Z(M,\mathbf{a}) = \hat{Z}_{Bk(\mathbf{a})}(M)$. 
\end{proposition}
\begin{proof}
    The $SL(2,\Z)$ representation from Section \ref{sl2z_section} tells us that,
$$\gamma \cdot \ket{K_2,\mathbf{a}_2} = \underset{n_2\in \Z+\mathbf{a_2}}{\sum_{m_2\in \Z+\frac{1}{2}}} q^{-n_2^2 p b -2 n_2m_2 a p-m_2^2 a r}f^{m_2}_{K_2}(q)  Y^{n_2 b +m_2 a} X^{n_2 p+m_2 r}.$$
After using orientation reversal and taking the bilinear form with $$\bra{K_1,\mathbf{a}_1} = \underset{n_1\in \Z+\mathbf{a_1}}{\sum_{m_1\in \Z+\frac{1}{2}}} f^{m_1}_{K_1}(q)   X^{-m_1} Y^{-n_1} ,$$ we can eliminate the sums over $n_1,n_2$ by setting, 
\begin{align*}
   n_1 &= - \frac{ b m_{1}  + m_{2}}{p} &  n_2 &= - \frac{m_1 + r m_2}{p},
\end{align*}
to find,
\begin{align*}
& \bra{K_1,\mathbf{a}_1}\gamma \ket{K_2,\mathbf{a}_2} \\
= & \sum_{m_1,m_2\in \Z+\frac{1}{2}}f^{m_2}_{K_2}(q)f^{m_1}_{K_1}(q) q^{-\frac{1}{p}(b m_{1}^{2}+2 m_{1} m_{2}+ r m_{2}^{2})}  \delta^\Z\left(\frac{m_2+m_1b}{p}+\mathbf{a}_1\right)\delta^\Z\left(\frac{r m_2+m_1}{p}+\mathbf{a}_2\right).
\end{align*}
The $\delta^\Z$'s come from the fact that the $n_i$ are restricted to lie in $\Z+\mathbf{a}_i$. Since $(\frac{1}{2},\mathbf{a}_{1})$ and $(\frac{1}{2},\mathbf{a}_{2})$ are restrictions of $\mathbf{a}$, they satisfy the following relation,
\begin{align*}
  -\frac{1}{2} & = \mathbf{a}_{2} p + \frac{r}{2} \mod \Z  & \mathbf{a}_{1} &= \frac{a}{2} + \mathbf{a}_{2} b \mod \Z   .
\end{align*}
These equations tell us that $\mathbf{a}_{1}$ and $\mathbf{a}_{2}$ are rational numbers of the form, 
\begin{align*}
    \mathbf{a}_{1} &= - \frac{1+b}{2 p} - \frac{b k }{p}  \mod \Z & \mathbf{a}_{2} & = - \frac{1+r}{2p }- \frac{ k }{p} \mod \Z, 
\end{align*}
for some integer $k$. We can write them as 
\begin{align}
   \mathbf{a}_{1} & = \frac{a+ b \delta_{0, r \hspace{-0.2cm} \mod 2} }{2}  + \frac{b \alpha^{\prime} }{p} \mod \Z \endline
   \mathbf{a}_{2} & =  \frac{\delta_{0, r \hspace{-0.2cm} \mod 2}}{2}  + \frac{\alpha^{\prime}}{p} \mod \Z,
\end{align}
where $\alpha^{\prime} =  -\frac{1+ r + p\delta_{0, r \hspace{-0.1cm}\mod 2} }{2} - k $ is an integer. Note that $\alpha$ is equivalent to $\alpha+ p$, therefore, $\alpha \in \Z_{p}$, further $\mathbf{a}_{1}$ and $\mathbf{a}_{2}$ are nicely split into $\frac{1}{2}\Z /\Z$ and $\frac{1}{p}\Z /\Z$, which is the split into the $\mathrm{Spin}$-structure on $M$ and $H^1(M,\Q/\Z) \cong \Z_p$. Using the relation between $\mathbf{a}_{1}$ and $\mathbf{a}_{2}$, we can write, 
\begin{align*}
   \delta^\Z\left(\frac{m_2+m_1b}{p}+ \mathbf{a}_{1} \right) \delta^\Z\left(\frac{r m_2+m_1}{p}+\mathbf{a}_2\right) = & \delta^\Z\left(\frac{m_2 + m_1 b }{p}+\frac{a}{2} + \mathbf{a}_{2} b \right)\delta^\Z\left(\frac{r m_2 + m_1}{p} +\mathbf{a}_{2} \right) \endline
   = & \delta^\Z\left(\frac{m_2 + m_1 b }{p}+\frac{a}{2} -\left( \frac{r m_2 + m_1}{p}  \right) b \right)\delta^\Z\left(\frac{r m_2 + m_1}{p} +\mathbf{a}_{2} \right).
\end{align*}
Using $b r - a p =1 $, we can further simplify the above expression and get,
\begin{align*}
   \delta^\Z\left(\frac{m_2 + m_1 b}{p}+\mathbf{a}_1\right) \delta^\Z\left(\frac{r m_2+m_1}{p}+\mathbf{a}_2\right)  = & \delta^\Z\left(\frac{r m_2 + m_1}{p} +\mathbf{a}_{2} \right)\delta^\Z\left( a (m_{2}- \frac{1}{2}) \right) .
\end{align*}
Since $m_{2} \in \frac{1}{2} + \Z$, $ a (m_{1}- \frac{1}{2})$ is always an integer, we have, 
\begin{align*}
    \bra{K_1,\mathbf{a}_1}\gamma \ket{K_2,\mathbf{a}_2} = \sum_{m_1,m_2\in \Z+\frac{1}{2}}f^{m_2}_{K_2}(q)f^{m_1}_{K_1}(q) q^{-\frac{1}{p}(b m_{1}^{2}+2 m_{1} m_{2}+ r m_{2}^{2})}  \delta^\Z\left(\frac{m_{1}+m_{2} r}{p}+ \frac{\alpha}{p}  \right),
\end{align*}
where $\alpha = \frac{p \delta_{0, r \hspace{-0.2cm} \mod 2}}{2} +\alpha^{\prime} $ is the Bockstein image of $\mathbf{a}$ as described earlier. This concludes the proof.
\end{proof}

A noteworthy special case is when $K_2$ is the unknot. In this case, we recover the surgery formula of \cite{GM}. That is, define the Laplace transform via, 
$$\mathcal{L}^\alpha_{p/r}(X^u) = q^{-u^2\frac{r}{p}}\cdot \delta^\Z(\frac{ru-a}{p}).$$
Then, we have the following result: 
\begin{corollary}
    Let $M$ be the 3-manifold resulting from Dehn surgery along a knot $K$ with coefficients $\frac{p}{r}$. Suppose $K$ is a plumbed knot with a GM series, 
    $$F_K(X,q) = \sum_{m\in \Z+\frac{1}{2}}f^m_K(q) X^m,$$
    then,
    $$Z(M,\mathbf{a}) = \bra{K,\mathbf{a}_1}\gamma\ket{\mathbb{S},\mathbf{a}_2} = q^{-\frac{b}{4p}+\frac{1}{4rp}}\cdot \mathcal{L}^\alpha_{p/r}\left((X^{\frac{1}{2r}}-X^{-\frac{1}{2r}})F_K(X;q)\right).$$
    So long as the right-hand side converges. 
\end{corollary}

As with most results in this paper, we also expect the gluing formula to hold whenever we have well-defined GM series, not just for the plumbed case. 

\begin{conjecture}[Gluing Formula]
    Let $K_1,K_2$ be knots with a well-defined GM series. Let $M = (S^3 \setminus\nu(K_2))\cup_\gamma (S^3 \setminus\nu(K_1))$ be the manifold resulting from gluing with $\gamma = \begin{pmatrix}
    r&a\\
    p&b
\end{pmatrix}$ as described above. Then,
    
    $$\hat{Z}_{\alpha}(q) = \sum_{m_1,m_2 \in \mathbb{Z}+\frac{1}{2}} f^{m_1}_{K_1}(q)\cdot f^{m_2}_{K_2}(q) \cdot  q^{-\frac{1}{p}(b m_{1}^{2}+2 m_{1} m_{2}+ r m_{2}^{2})} \cdot \delta^\Z\left(\frac{m_{1}+m_{2} r}{p}+ \frac{\alpha}{p}  \right). $$
So long as the right-hand side converges. 
\end{conjecture}

It would be interesting to verify this in some examples when one of the knots is not plumbed. However, this is difficult, as we are not aware of many results on knot gluing.

%%%%%%%%%%%%%%%%%%%%%%%%%%%%%%%%%%%%%%%%%%%%%%%%%%%%%%%%%%%%%%%%%%%%%%%%%%%%%%%%%%%%%%

\section{Link Surgeries and States}\label{amplitudesandstates}

\subsection{Surgeries on Links}\label{partial_surg_section}

In the previous section, we derived \eqref{treelinks}, hence re-deriving the plumbing formula \eqref{plumbingformula} for negative-definite 3-manifolds and the surgery formula for knot complements (Laplace transform) from more fundamental principles. In particular, we showed that the results of \cite{GM} follow from the module associated with the torus $\mathcal{H}(T^2)$ and the associated mapping class group representation. 

In this work, we are largely interested in the generalization to links. To this end, we introduce the following notion, 
\begin{definition} Suppose $Z$ assigns well-defined vectors for a link complement $\ket{L,\mathbf{a}}$. We define the GM series for the corresponding $N$-component link as, 
$$F_L(X_1,X_2,\ldots,X_N;q) = \sum_{m\in \frac{1}{2}\Z^N} \ket{0,m} \braket{0,m}{L,\mathbf{0}},$$
where $\mathbf{0}$ above means $(0,\mathfrak{s}_0)$ (see appendix \ref{notationsandconventions} for the definition of $\mathfrak{s}_{0}$). 
\end{definition}
Notice that in the language of \cite{GM}, we can equivalently define the above objects in terms of unintegrated vertices. That is, the plumbing formula \eqref{plumbingformula} with $N$ vertices left unintegrated produces objects of the form,
$$\hat{Z}_{\alpha} (M;X_1,\ldots,X_N,n_1,\ldots,n_N;q).$$
The GM series for the $N$-component link complement in this language is simply, 
$$F_L(X_1,\ldots,X_N;q) = \hat{Z}_{0} (M;X_1,\ldots,X_N,0,\ldots,0;q).$$
For instance, the GM series for tree links follows from Proposition \ref{treelinks_wvfn}:
$$F_{L_Q}(X;q) = \prod^N_{i=1}(X_i^{\frac{1}{2}} - X_i^{-\frac{1}{2}})^{1-\mathrm{deg} v_i}.$$
If we hope to leave the realm of `tree links,' a natural next step is to consider what happens to the wavefunctions \eqref{treelinks} under partial surgeries. To this end and its consequences, we dedicate this section. 

Let us first briefly recall the rational calculus of Rolfsen \cite{Rolf}. Let $L=\cup^N_{i=1}L_i$ be any link with an unknot component $L_1$. Suppose all components $L_i$ have some rational surgery framing $\gamma_i$. Then we may replace $L$ by $L'$ as specified by figure \ref{Rolfsencalc} provided we change the surgery framings via,
\begin{align*}
 \gamma_1 & \longrightarrow \gamma_1'=\frac{1}{\frac{1}{\gamma_1}+ \tau }    &  \gamma_i &\longrightarrow \gamma_i+\tau \cdot lk(L_1,L_i)^2
\end{align*}
\begin{figure}[H] 
\centering
\includegraphics[width=0.7\textwidth]{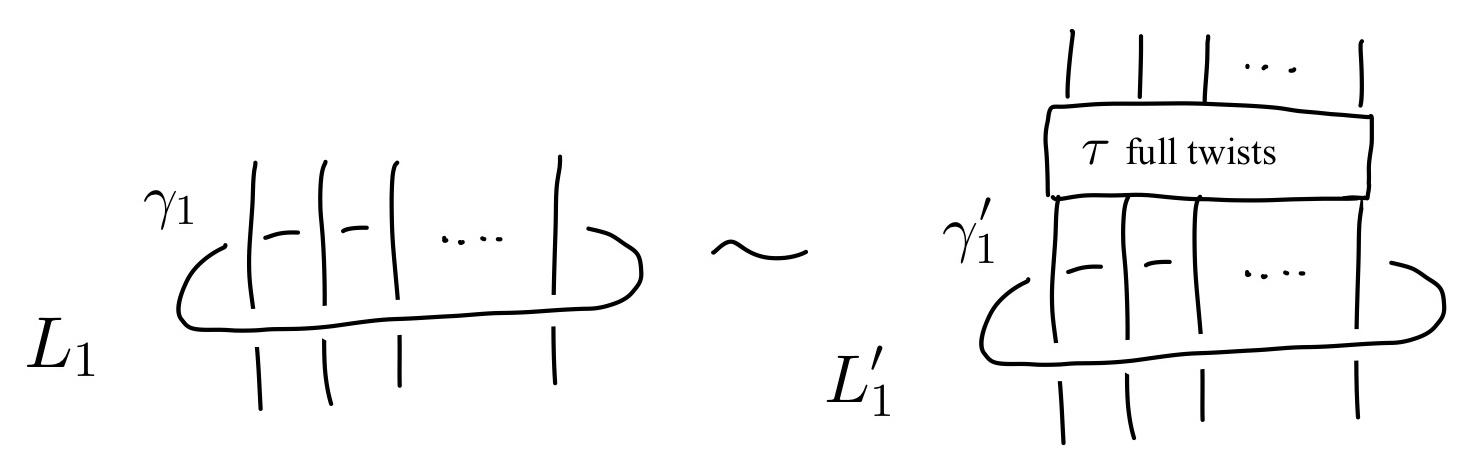}
\caption{Partial surgery of unknot component $L_1$.}
\label{Rolfsencalc}
\end{figure}
This means that if a link, $L$, whose complement in $S^3$ has a surgery diagram given by some partial assignment of rational surgeries on a tree link of the form \eqref{treelinks}, then we are able to acquire its wavefunction $\ket{L}$ through performing the appropriate partial surgeries. We shall call such links plumbed, thus warranting the following definition.
\begin{definition}
Let $L$ be an $N$-component link and $M_L = S^3 -\nu(L)$. We will call $L$ plumbed if $M_L$ has a surgery presentation as a plumbing graph with $N$ vertices viewed as solid tori complements.  
\end{definition}
For instance, torus links and satellites thereof are all plumbed links. With this in mind, we find the following result,
\begin{proposition}\label{plumbed_link_wavfn}
    Let $M$ denote the complement of an $N$-component plumbed link, $L$, in $S^3$. Additionally, let $Q$ denote the linking matrix of $L$ with $0$ on the diagonals (reflecting that all components are taken to be 0-framed). $Z(M) \in \mathcal{H}(T^2)^{\otimes N}$ is given by the following vector,
    $$\ket{L,\mathbf{a}} = \underset{m\in \frac{1}{2}\Z^N}{\sum_{\mathbf{n}\in \Z^N+\mathbf{a}}} q^{-(Q\mathbf{n},\mathbf{n})} \cdot f^m_L(q) \prod^N_{i=1} Y^{n_i}_i X_i^{m_i+(Q n)_i} $$
    where $F_L(X_1,X_2,\ldots,X_N;q) = \sum_{\mathbf{m}}f^\mathbf{m}_L(q) \cdot X^\mathbf{m}$ is the GM series for the link $L$.
\end{proposition}
\begin{proof}
    By induction, we may assume $L$ is a tree link. The main moves to check are the Rolfsen twists. The others are straightforward. 

Therefore, we are interested in $\frac{1}{r}$ surgeries on a unknot $L_1$, inside of $L$. Let $Q_{i,j} = lk(L_i,L_j)$ be the linking matrix associated with $L$ (again, the diagonals are framing coefficients). We can immediately deduce the change in $Q$ under the Rolfsen move above:

$$\{Q_{i,j}\}_{i,j=1,2,..,N} \longrightarrow \{Q'_{k,l}\}_{k,l=2,..,N} $$
$$Q'_{k,l} =  \begin{cases} 
      Q_{k,l} & Q_{1,k} = 0 \hspace{2mm} or\hspace{2mm}Q_{1,l} =0 \\
      Q_{k,l} + r Q_{1,k}Q_{1,l} & Q_{1,k} \neq 0 \hspace{2mm} and \hspace{2mm}Q_{1,l} \neq0\\
      Q_{k,l} + r Q_{1,k}^2 & j=k  
   \end{cases}
$$
or, put more simply:
$$Q'_{k,l} = Q_{k,l} + r Q_{k,1}Q_{1,l}.$$
We frame a solid torus state by $\gamma_{-1/r} = \begin{pmatrix}
    0&1\\
    -1&r
\end{pmatrix}$ to find
$$\gamma_{-1/r} \ket{\mathbb{S},\mathbf{a}_1} = \sum_{n'\in \Z+\mathbf{a}_1,m'=\pm1/2} \delta_{|m'|,1/2} \mathrm{sgn}(m') q^{2m'n' +n'^2r} Y^{n'r+m'} X^{-n'}. $$
Gluing this along $L_1$, we see that:
$$\bra{\mathbb{S},\mathbf{a}_1}\gamma_{-1/r}^\dagger\ket{{n_1},{m_1+(Q \mathbf{n})_1}} = \sum_{m'} \delta_{|m'|,1/2} \mathrm{sgn}(m')\cdot  q^{m_1^2 r - 2m'm_1+(Q\mathbf{n})_1 \cdot (2m' - 2m_1 r -(Q\mathbf{n})_1 r )} \cdot \delta_{n_1, r(m_1+(Q \mathbf{n})_1)+m'}.$$
Let $k>1$, after collapsing the sum over $n_1$, we see that:
$$(Q \mathbf{n})_l = n_1 Q_{l,1} + \sum_{k=2}^N Q_{l,k}n_k = (r(m_1+(Q \mathbf{n})_1)+m') Q_{l,1} + \sum_{k=2}^N Q_{l,k}n_k. $$
Note that, $Q_{1,1} =0$ by assumption (that is, $(Q\mathbf{n})_1 = \sum_{k>1} Q_{1,k}n_k$), so: 
\begin{align*}
   (Q \mathbf{n})_l  =& (rm_1 +m')Q_{l,1} +\sum^N_{k=2} (Q_{l,k}+r Q_{1,k} Q_{l,1})n_k  \endline
   =& (rm_1 +m')Q_{l,1} + (Q' \mathbf{n'})_l ,
\end{align*}
where if $\mathbf{n}=(n_1,n_2,\ldots,n_N)$, we let $\mathbf{n'} = (n_2,\ldots,n_N)$. These will be part of the exponents of $X_i$. Similarly, for the $q$ exponents, we note that: 
\begin{align*}
    & -(Q\mathbf{n},\mathbf{n}) + m_1^2 r - 2m'm_1+(Q\mathbf{n})_1 \cdot (2m' - 2m_1 r -(Q\mathbf{n})_1 r ) \\
    = & -2n_1Q^{1k}n_k -n_kQ^{kl}n_l + m_1^2 r - 2m'm_1+(Q\mathbf{n})_1 \cdot (2m' - 2m_1 r -(Q\mathbf{n})_1 r ) \\
    = & -2[r(m_1+(Q \mathbf{n})_1)+m']Q^{1k}n_k -n_kQ^{kl}n_l + m_1^2 r - 2m'm_1+(Q\mathbf{n})_1 \cdot (2m' - 2m_1 r -(Q\mathbf{n})_1 r ) \\
    = & -n_kQ^{kl}n_l - rn_kn_lQ^{1k}Q^{1k}+m_1^2r - 2m' m_1 \\
    = & -(Q'\mathbf{n'},\mathbf{n'})+ m_1^2r - 2m' m_1.
\end{align*}
Therefore, we may write:
\begin{align}\label{partialsurgwavfn}
   &  \bra{\mathbb{S},\mathbf{a}_1}\gamma_{-1/r}^\dagger\ket{L,\mathbf{a}} \endline
   = & \underset{m\in \frac{1}{2}\Z^N,m'=\pm 1/2}{\sum_{\mathbf{n}'\in \Z^{N-1}+\mathbf{a}'}}  \delta_{|m'|,1/2} \mathrm{sgn}(m')q^{-(Q'\mathbf{n'},\mathbf{n'})} \cdot f^\mathbf{m}_L(q) \left[q^{m_1^2r +2m' m_1} \cdot \prod_{k\neq 1} X^{(rm_1+m')lk(L_1,L_k)}_k\right] \endline & \hspace{2.5cm} \times \prod_{k\neq1} Y^{n_k}_k X_k^{m_k+(Q' n)_k},
\end{align}
where $\mathbf{a}' = (\mathbf{a}_2,\ldots,\mathbf{a}_{N})$. This implies the result.
\end{proof}

Of course, we expect the form of these wavefunctions to hold more generally for any link with a well-defined GM series. 
\begin{conjecture}\label{general_link_wavefun_conj}
Let $L$ be any link with a well-defined GM series,
$$F_L(X_1,X_2,\ldots,X_N;q) = \sum_{\mathbf{m}\in \frac{1}{2}\Z^N}f^\mathbf{m}_L(q) \cdot X^\mathbf{m}.$$
Then its wavefunction in $\mathcal{H}(T^2)^{\otimes N}$ exists and is of the form,
 \be \ket{L,\mathbf{a}} = \underset{m\in \frac{1}{2}\Z^N}{\sum_{\mathbf{n}\in \Z^N+\mathbf{a}}} q^{-(Q\mathbf{n},\mathbf{n})} \cdot f^\mathbf{m}_L(q) \prod^N_{i=1} Y^{n_i}_i X_i^{m_i+(Q n)_i} \label{general_link_wavfn} \ee
 where $Q$ is the linking matrix of $L$, with $0$ diagonal.  
\end{conjecture}

As a side remark, we note that just as in the knot case, taking the constant term in the $Y$'s gives us back $F_L$, as expected. 

As an immediate consequence, we may take the $Y$ constant term of equation \eqref{partialsurgwavfn} to find, 

\begin{theorem}[Partial Surgery Theorem] \label{partialsurgerytheorem}  Let $L$ be any plumbed link with unknot component $L_1$. If $L'$ is the link resulting from partial $-1/r$ ($r>0$) on $L_1$, then the GM series of $L$ and $L'$ are (up to an overall $\pm q^c$ factor) related by:
$$F_{L'}(X_2,\ldots,X_N;q) = \mathcal{L}_{X_1,-1/r}\left((X_1^{\frac{1}{2r}} - X_1^{-\frac{1}{2r}})\cdot F_L(X_1,X_2,\ldots,X_N;q)\right),$$
where the Laplace transform here acts as:
$$\mathcal{L}_{X_1,-1/r}: X_1^m \mapsto q^{m^2 r}\prod_{k\neq 1} X^{r\cdot  m \cdot lk(L_1,L_k)}, $$
For $r=0$, we find the $\infty$-surgery formula for $L$:
$$F_{L'}(X_2,\ldots,X_N;q) = \sum_{m'=\pm \frac{1}{2}}  \mathrm{sgn}(m') \prod_{k\neq 1} X^{m' lk(L_1,L_k)}_k  F_L(X_1 = q^{2m'},X_2,..,X_N). $$
\end{theorem}

This provides partial proof of the conjecture by Park \cite{S1}, which we restate here:
\begin{conjecture} \label{partial_surg_conj} The Partial Surgery Theorem holds for any link with a well-defined GM series whenever the Laplace transform converges. Similarly for $\infty$-surgery. 
\end{conjecture}

We can also easily perform surgery on all components of a link at once and acquire a relatively compact formula. If $L$ is an $N$-component link $L = \cup_i L_i$,  recall that $\{p_i/r_i\}_{i=1,\ldots,N}$ Dehn surgery on $L$ yields a closed 3-manifold $M$ with,
$$H_1(M) = \Z^N/\tilde{Q}\Z^N$$
where,
$$\tilde{Q}_{ij} = \biggl\{ \begin{matrix}
    p_i &i=j\\
    r_i lk(L_i,L_j) & i\neq j
\end{matrix}$$
and the $Spin^c$-structures are identified with,
$$Spin^c(M) = \Z^N/\tilde{Q}\Z^N +\delta,$$
where
$$\delta_i = \frac{1}{2}(1+r_i)+\frac{r_i}{2} \biggl(\sum_{j\neq i}lk(L_j,L_i)\biggr) \mod \Z . $$

\begin{proposition}{(Rational Laplace Transform for Links)}\label{laplace_transform_links}
    Let $\ket{L,\mathbf{a}}$ be wavefunctions in $\mathcal{H}(T^2)^{\otimes N}$ of the form \eqref{general_link_wavfn}, 
    $$\ket{L,\mathbf{a}} = \underset{}{\sum_{\mathbf{n}\in \mathbf{a} + \Z^N}}  q^{-(Q\mathbf{n},\mathbf{n})} \cdot  \prod^N_{i=1} Y^{n_i}_i X_i^{(Q n)_i}F_L(X;q).$$
    Additionally, suppose $Q_f = Q+diag(\frac{p_1}{r_1}, \ldots,\frac{p_N}{r_N} )$ is invertible (with $gcd(p_i,r_i) =1$ for all $i=1,\ldots,N$).
    Then, we have, 
    \be \label{rational_link_surg}\left(\otimes^N_{i=1} \bra{\mathbb{S},\mathbf{a}_i} \gamma_{p_i/r_i}^\dagger s\right) \ket{L} = \mathcal{L}^\alpha_{Q_f} \left( \prod^N_{i=1}(X_i^{\frac{1}{2r_i}} - X_i^{-\frac{1}{2r_i}}) \cdot \tilde{F_L}\right)\ee
    whenever the left-hand side converges. Where the multivariable Laplace transform is defined by, 
    $$\mathcal{L}^\alpha_{Q}  (X^\mathbf{\mu}) = \delta^{Q \mathbb{Z}^N} (\mathbf{\mu} - \frac{\alpha}{\mathbf{r}}) \cdot q^{-(\mu,Q^{-1} \mu)}$$    
    and $\alpha = Bk(\mathbf{a})$. Here, $Q$ an invertible $N\times N $ matrix and vector division is taken element wise, $(\frac{\alpha}{\mathbf{r}})_i = (\frac{\alpha_i}{r_i})$.  In the case that $\ket{L}$ is the wavefunction for an $N$-component plumbed link, $L$, and $M = S^3_{\frac{p_1}{r_1},\ldots,\frac{p_N}{r_N}}(L)$ is a weakly negative definite 3-manifold, then this is the $\hat{Z}$ invariant of $S^3_{\frac{p_1}{r_1},\ldots, \frac{p_N}{r_N} } (L)$ up to an overall $\pm q^c$ factor,
    $$\Hat{Z}_\alpha (S^3_{\frac{p_1}{r_1},\ldots, \frac{p_N}{r_N} } (L)) \cong \mathcal{L}^\alpha_{Q_f} \left( \prod^N_{i=1}(X_i^{\frac{1}{2r_i}} - X_i^{-\frac{1}{2r_i}}) \cdot F_L\right) $$
    where $\alpha$ takes values in $Spin^c(M) \cong \Z^N/\tilde{Q}\Z^N +\delta$ described above and $F_L$ is the GM series of $L$.     
\end{proposition}
\begin{proof}
    We are interested in the amplitudes, $\bra{\mathbb{S},\mathbf{a}_i} \gamma_{p_i/r_i}^\dagger s^\dagger \ket{n_i, m_i -(Qn)_i}$. Repeating essentially the same calculation we have done several times by now, we find, 
$$\bra{\mathbb{S},\mathbf{a}_i} \gamma_{p_i/r_i}^\dagger s^\dagger \ket{n_i, m_i -(Qn)_i} \cong \sum_{\epsilon_i = \pm1} \epsilon_i q^{-n_i^2 \frac{p_i}{r_i} - 2n_i \frac{\alpha_i}{r_i}} \delta_{m_i - \frac{\alpha_i}{r_i} +\frac{\epsilon_i}{2 r_i} , (Qn)_i + \frac{p_i}{r_i} n_i} \delta^\mathbb{Z}(n_i) . $$
Plugging this into the left-hand side of \eqref{rational_link_surg}, we find, 
$$\left(\otimes^N_{i=1} \bra{\mathbb{S},\mathbf{a}_i} \gamma_{p_i/r_i}^\dagger s\right) \ket{L} = \underset{ \mathbf{n} \in \mathbb{Z}^N, \mathbf{m}\in \frac{1}{2}\Z ^N}{\sum_{\epsilon \in \{\pm1 \}^N}} \sigma(\epsilon) f^\mathbf{m}_L(q) \cdot q^{-(Q_f \mathbf{n} , \mathbf{n}) -2 (\mathbf{\frac{\alpha}{r}}, \mathbf{n})} \cdot  \delta_{\mathbf{m} -\mathbf{\frac{\alpha}{r}} +\mathbf{\frac{\epsilon}{2r}} , Q_f \mathbf{n}} \cdot  \delta^{Q_f\mathbb{Z}^N}(\mathbf{m} -\mathbf{\frac{\alpha}{r}}+\mathbf{\frac{\epsilon}{2r}}).$$
On the right-hand side, we have, 
\begin{align*}
   \mathcal{L}^\alpha_{Q_f} \left( \sum_{\epsilon \in \{\pm1 \}^N} \sigma(\epsilon) X^{\mathbf{\frac{\epsilon}{2r}}} \cdot \tilde{F_L}\right)  = &  \underset{\mathbf{m}\in \frac{1}{2}\Z ^N}{\sum_{\epsilon \in \{\pm1 \}^N} }\sigma(\epsilon) \cdot f^\mathbf{m}_L \cdot \delta^{Q_f\mathbb{Z}^N}(\mathbf{m} -\mathbf{\frac{\alpha}{r}}+\mathbf{\frac{\epsilon}{2r}}) \cdot q^{- (\mathbf{m} +\mathbf{\frac{\epsilon}{2 r}},Q_f^{-1} (\mathbf{m} +\mathbf{\frac{\epsilon}{2 r}}))}.
\end{align*}
Therefore, the expressions match up to an overall $\pm q^{c}$ factor. For the second part of the statement, simply apply Proposition \ref{plumbed_link_wavfn}.
\end{proof}

This proves the conjecture by Park \cite{S2} in the case of integer surgery on plumbed links. This also provides a generalization to rational surgery coefficients. Furthermore, we conjecture the above surgery formula for $\hat{Z}$ invariants holds for any link with a well-defined GM series so long as the Laplace transform converges. 

An immediate application of the above discussion is given by the satellite operation, which we now describe. Let $P$, $K$ be any two knots. Let $X = S^3 \setminus int(\nu(K))$ denote the exterior of $K$. Additionally, let $P$ be a knot inside the solid torus $S^1\times D^2$ and let $Y = S^1\times D^2 \setminus int(\nu(K)) $ denote its complement inside the solid torus. If $C_P(K)$ denotes the satellite of $K$ with pattern $P$, then the complement of $C_P(K)$ in $S^3$ is given by gluing the boundary of $X$ to the exterior boundary of $Y$ via an $S$ transformation (that is, $\bra{X}S\ket{Y}$ where the inner product takes place in the described factor of $\mathcal{H}(T^2)^{\otimes 2} \ni \ket{Y} $). To define a knot $P$ inside the solid torus, it is sufficient to consider a link in $S^3$ consisting of 2 components, one of these being an unknot, and then take the complement of the unknot component. We call the knot sitting inside $S^3$, which defines the pattern $P'$. More generally, we can allow for the pattern $P$ to be a link of $N-1$ components, and $P'$ will therefore have $N$ components.\footnote{We would like to thank Sunghyuk Park for sharing his notes on the satellite operation with us for this description.}

Then, for such a gluing, we repeat the calculation leading to \eqref{partialsurgwavfn}, except the only change is, of course, the replacement $\delta_{|m'|,1/2} \mathrm{sgn}(m') \rightarrow f^{m'}_K(q)$. This gives us:

\begin{theorem}[Satellite Formula]\label{satellite}
    Let $P'$ be a plumbed link with an unknot component $P'_1$ and GM series,
    $$F_{P'} = \sum_{n\in \frac{1}{2}\Z} f^n_{P'}(X_2,X_3,\ldots) \cdot X_1^n $$
    and $K$ a plumbed knot with GM series,
    $$F_K = \sum_{m\in \Z+\frac{1}{2}} f^m_{K} \cdot X^m,$$
    then the satellite of $K$ with pattern $P$ (defined by taking the complement of $P'_1$), $C_P(K)$ has corresponding GM series given by:
\be \label{sat_1} 
F_{C_P(K)}(X_2,\ldots,X_N;q) = \sum_{m\in \Z+\frac{1}{2}} f^{m}_K \prod_{k\neq 1} X^{m lk(L_1,L_k)}_k  F_{P'}(X_1 = q^{2m},X_2,\ldots,X_N).
\ee 
Equivalently, we may write this as, 
\be \label{sat_2}
F_{C_P(K)}(X_2,\ldots,X_N;q) = \sum_{n\in \frac{1}{2}\Z} f^n_{P'} (X_2,X_3,\ldots) \cdot F_K(X = q^{2n } \prod_{k\neq 1} X^{lk(L_1,L_k)}_k).
\ee     
\end{theorem}
The two different ways of writing the GM series for the satellite come from interchanging two convergent sums. In general, for two general knots/links with well-defined GM series, it may be that only one or neither of the formulas above converges. Thus, we are led to the following conjecture,
\begin{conjecture}
    For any knot $K$ and any pattern $P$ with pattern-defining link $P'$ with a well-defined GM series, the GM series of the satellite $C_{P}(K)$ is given by either \eqref{sat_1} or \eqref{sat_2} (so long as one of them converges). If they both converge, then they are equal.  
\end{conjecture} 

\subsubsection{Example: Torus Links and Cabling}\label{torus_links_section}

Let us try to derive some useful specific instances of Theorem \ref{satellite}. We will provide explicit formulas for cabling with an $(tp,sp)$ torus link as well as a conjecture for Whitehead doubling (and an infinite class of its generalizations). 

In \cite{GM}, a closed formula was found for the GM series of torus knots $(r,s)$ with $gcd(r,s)=1$. Here, we will generalize their formulas to $(sp,tp)$ torus \textit{links}. Our strategy will be to derive a surgery diagram for the complement of any such torus link, and then carry out partial surgeries on its wavefunction per the previous section. 

\begin{figure}[H] 
\centering
\includegraphics[width=0.7\textwidth]{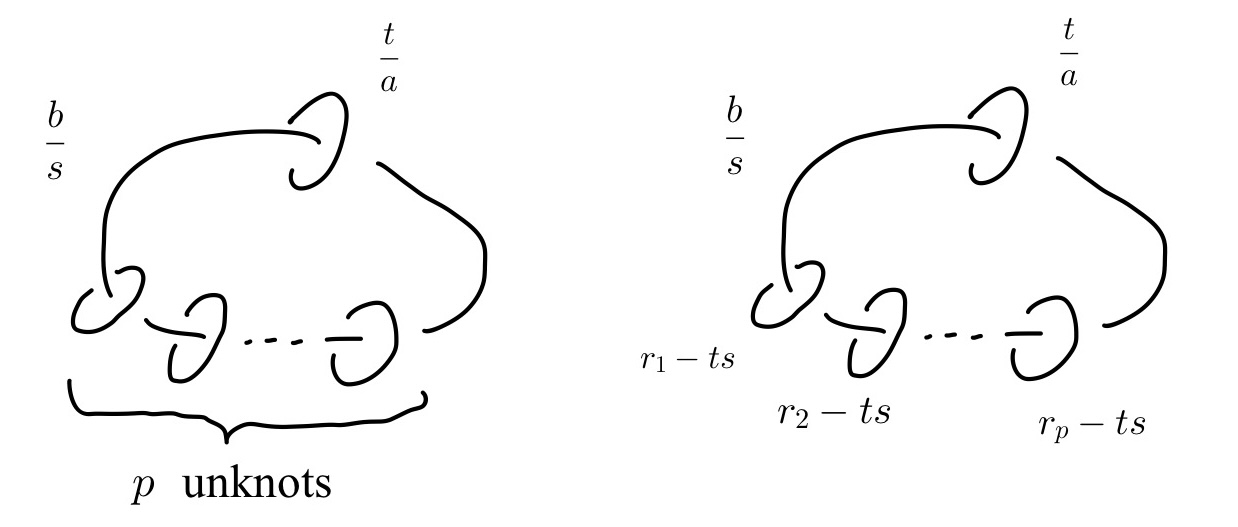}
\caption{Surgery diagram for the complement of a $T(sp,tp)$ torus link in $S^3$ (left) and the diagram for surgery on a torus link (right). We provide the picture on the right for completeness, and it should be understood that the link component $L_i$ ($i\in \{1,2,..,p\}$) has rational framing $r_i$. On the left, we exclude the framing to indicate that each unknot is the complement of an unknot (no surgery is being done).  }
\label{torus_links}
\end{figure}

\begin{proposition}
    Let $T(sp,tp)$ denote a torus link and $M$ its complement in $S^3$. Let $a,b$ be solutions to Bezout's identity $sa-tb = 1$. Then $M$ has the surgery description in figure \ref{torus_links}.
\end{proposition}
\begin{proof}
    As a torus link, $T(sp,tp)$ has the obvious description of wrapping the $(sp,tp)$ curve in $\pi_1(T^2)$. Inspired by the techniques in \cite{KS}, we embed the torus in $S^3$ and put two $\infty$-framed unknots wrapping the meridian and longitude of a slightly larger and slightly smaller torus as in figure \ref{inf_fram}. Call them $O_l$ and $O_m$.

\begin{figure}[H] 
\centering
\includegraphics[width=0.4\textwidth]{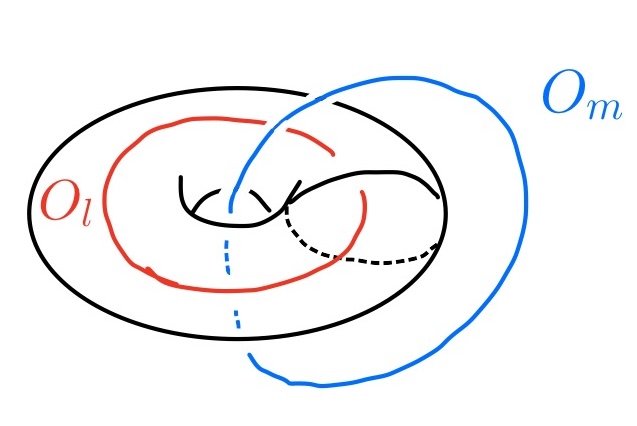}
\caption{$O_l$ and $O_m$ (both $\infty$ framed) in relation to the torus. It is to be understood that the torus link lies in a thickened version of the $T^2$ in the figure. }
\label{inf_fram}
\end{figure}

By assumption, there exists integers $a,b$ such that $\begin{pmatrix}
    s&b\\
    t&a
\end{pmatrix}\in SL(2,\mathbb{Z})$. Additionally, $\frac{t}{s}$ has some continued fraction $(k_1,k_2,\ldots,k_n)$:
$$\frac{t}{s} = k_1+\frac{1}{k_2+\frac{1}{k_3+\ldots}}.$$
Therefore, we can begin by doing $-k_1$ Rolfsen twists along $O_m$. Recall that torus links have the braid description of having $s p$ strands with $tp$ sweeps along them. Without loss of generality, we assume there are $sp$ strands running transversely to $O_m$. As such, the twists turn the $T(sp,tp)$ torus link into a $T(sp,tp-k_1ps)$ torus link. A key observation here is that $O_l$ does \textit{not} leave its place in the core of the torus. The framings change per figure \ref{Rolfsencalc}:
\begin{align*}
&O_l: \infty +k_1 = \infty \\
&O_m: \frac{1}{k_1}.
\end{align*}
Notice that:
$$\frac{s}{t-sk_1} = k_2+\frac{1}{k_3+\frac{1}{k_4+\ldots}}.$$
So that the natural next step is to do $-k_2$ Rolfsen twists on $O_l$. Since the longitude and meridian directions are interchangeable, we similarly find the resulting torus link to be $T(p(s-k_2(t-k_1s)),p(t-k_1 s))$. Again, $O_m$ stays put, and the framings change to:
\begin{align*}&O_l: \frac{1}{\frac{1}{\infty}+k_2} = \frac{1}{k_2}\\
&O_m: \frac{1}{k_1}+k_2. \end{align*}
We proceed until the torus link is completely unknotted into a disjoint union of $p$ unknots. The final framings will be:
\begin{align*}&O_l: \frac{1}{\frac{1}{\frac{1}{\cdots+ k_{n-2}}+k_{n-1}}+k_n} = \frac{1}{(k_n,k_{n-1},\ldots,k_2)}\\
&O_m: \frac{1}{\frac{1}{\cdots+ k_{n-2}}+k_{n-1}}+k_n = (k_n,\ldots,k_1).\end{align*}
We have already seen automorphisms $\begin{pmatrix}
    s&b\\
    t&a
\end{pmatrix}\in SL(2,\mathbb{Z})$ decompose into\footnote{Note this corresponds to a continued fraction with negative signs instead, but one can always invert $k_i\rightarrow-k_i$ in the appropriate places to move between the two descriptions.}:
$$\begin{pmatrix}
    s&b\\
    t&a
\end{pmatrix} = S\tau^{k_1}_-S\tau^{k_2}_-S\cdots\tau^{k_n} = \begin{pmatrix}
    k_1&1\\
    -1&0
\end{pmatrix}\begin{pmatrix}
    k_2&1\\
    -1&0
\end{pmatrix}\cdots\begin{pmatrix}
    k_n&1\\
    -1&0
\end{pmatrix} .$$
Taking the transpose of both sides, one finds:
$$\frac{b}{s} = (k_n,\ldots,k_1).$$
Similarly taking the transpose and then inverting $S\tau^{k_1}_-$ on the right gives:
$$\frac{a}{t} = (k_n,\ldots,k_2).$$
The resulting surgery diagram is then that of figure \ref{torus_links}.
\end{proof}

Once we have a plumbed description for the $T(sp,tp)$ links, the formalism we have developed thus far allows immediate and simple calculations for the gluing. That is, we have,
\begin{proposition}\label{toruslinksGM}
    The GM series of the $T(tp,sp)$ torus link is given by,
    $$F_{T_{tp,sp}}(X_1,\ldots,X_p;q) = \sum_{n\in \mathbb{Z}+\frac{st(p-1) +1}{2}} f^n_{T_{sp,tp}}(q) \prod^p_{i=1}X_i^n$$
    
    $$f^n_{T_{sp,tp}}(q) = q^{\frac{n^2}{st}}\sum_{\epsilon,\epsilon' = \pm1} \epsilon \epsilon' \mathrm{sgn}(-2n +\epsilon t +\epsilon' s )^p \begin{pmatrix}
        \frac{|\epsilon t+\epsilon' s -2n|+pst}{2st}-1\\
        p-1
    \end{pmatrix} \delta^\mathbb{Z}\left(\frac{\epsilon t+\epsilon' s+pst -2n}{2st}\right).$$
\end{proposition}
\begin{proof}
The wavefunction of the tree link displayed in figure \ref{torus_links} can immediately be read off from \eqref{treelinks},
$$\ket{L,\mathbf{a}} = \sum_{n\in \Z^{p+2}+\mathbf{a}} q^{(Qn,n)} Y^n (X_0^{\frac{1}{2}}-X_0^{-\frac{1}{2}})^{-p} X^{Qn} \in \otimes^{p+1}_{i=0}\mathcal{H}(T^2)$$
where we take $0$ to be the index of the center vertex and $p+1$ to be the index of the leaf with surgery framing $\frac{t}{a}$. The result is then immediately given by taking the $Y$ constant term of,
$$\bra{\mathbb{S},\mathbf{a}_0}\otimes\bra{\mathbb{S},\mathbf{a}_1} (\gamma^\dagger_{b/s}\otimes\gamma^\dagger_{t/a})\ket{L,\mathbf{a}}.$$
Notice that the labels $\mathbf{a}_1$ and $\mathbf{a}_0$ on the unknots being surgered must correspond to the trivial cohomology element since the result of the gluings is $S^3$.
\end{proof}

Though we have done partial surgeries in the calculation above, it is worth mentioning that these do \textit{not} fall under the Partial Surgery Theorem. That is, we have performed more general rational surgeries. Such gluings for arbitrary wavefunctions will generically yield link complements in arbitrary 3-manifolds $M$, not $S^3$. These are perfectly fine wavefunctions, but fall outside the definition of the GM series. Here, we know the end result is a link complement in $S^3$ by construction, so there are no issues. Moreover, we emphasize that our formalism allows for easy computations of such cases, therefore generalizing the Partial Surgery Theorem for link complements. 

However, if we want a cabling formula, we still have some more work to do. Let us say we wish for the pattern link $P$ to be a torus link $T(tp,sp)$ where $sp$ strands wrap along the interior of $S^1\times D^2$ (note that this is \textit{not} symmetric in $s$ and $t$, we are picking a preferred direction). Then the unknot $L_0$, whose complement is the $S^1\times D^2$, wraps around the $sp$ strands. 

For instructive purposes, let us see how this works in the simple case of $s=1$ links. The surgery diagram is simply a $p$-valency vertex with central framing $-1/t$. Applying the Partial Surgery Theorem for $-\frac{1}{r}$, we find, 
$$\mathcal{L}_{-1/r,X_0} \left(\frac{X^{1/2r}_0-X^{-1/2r}_0}{(X^{1/2}_0 -X^{-1/2}_0)^{p-1}}\right) = \sum_{m,\epsilon}\epsilon f^m_{p+1} \cdot q^{rm^2+\epsilon m} \prod^p_{i=1}X^{r m +\epsilon/2}_i. $$
We separate $r = r' +t$, and rewrite the above as,
$$\mathcal{L}_{-1/r,X_0} \left(\frac{X^{1/2r}_0-X^{-1/2r}_0}{(X^{1/2}_0 -X^{-1/2}_0)^{p-1}}\right)= \mathcal{L}_{-1/r',X_0}\left( (X^{1/2r'}_0-X^{-1/2r'}_0)\cdot \sum_m f^m_{p+1}X_0^m q^{tm^2}\prod^p_{i=1}X_i^{tm}\right)$$
which implies the GM series of the pattern defining link $P'$ for this case is given by,
\be \label{t,tp case}F_{P'}(X) = \sum_m f^m_{p+1}X_0^m q^{tm^2}\prod^p_{i=1}X_i^{tm}\ee
$$f^m_{p+1} = \mathrm{sgn}(m)^p \binom{\frac{p}{2}-\frac{3}{2}+m}{p-2} \delta^{\mathbb{Z}}\left(m-\frac{p}{2}-\frac{1}{2}\right).$$

In the case of general $s$, $-\frac{1}{r}$ surgery on $L_0$ should return us a torus link yet again, namely $T(p(t+rs), sp)$. Using this observation, we see the following result,
\begin{lemma} \label{patterndefiningtoruslink}
    The GM series for the pattern defining link $P'$ of the $P = T(tp,sp)$ torus link is given by (up to an overall $\pm q^c$ factor),
    $$F_{P'} \cong \sum_{m\in \frac{1+sp}{2}} f^m_{P'}(X) \cdot X^m_0$$
    
    $$f^m_{P'}(X) = \left[ \sum_{\epsilon=\pm1} \epsilon \cdot \mathrm{sgn}(m-\frac{\epsilon}{2})^p\binom{\frac{|2m -\epsilon|+ps}{2s}-1}{p-1} \delta^\mathbb{Z}\left( \frac{2m-\epsilon-ps}{2s}\right)\right] \cdot q^{m^2 \frac{t}{s} }X^{mt}$$
    Where $X$ is shorthand for $\prod^p_{i=1}X_i$.
\end{lemma}
\begin{proof}
    Let $P'$ be the pattern defining link for $P=T(tp,sp)$. Assume its GM series to be of the form,
$$F_{P'} = \sum_{m\in \frac{1}{2}\Z} f^m_{P'}(X) X^m_0.$$
Applying the Partial Surgery Theorem for $-\frac{1}{r}$ surgery on the $L_0$ component, we find,
$$\sum_{m \in \frac{1}{2} \Z}\sum_{\epsilon' =\pm 1} \epsilon' f^m_{P'}(X) q^{rm^2 +\epsilon' m} X^{s(rm+\epsilon'/2)}$$
where $X$ is a placeholder for $\prod^p_{i=1}X_i$. Using Proposition \ref{toruslinksGM}, we may set the above equal to,
$$F_{T_{p(t+ rs),ps}} = \sum_{n\in \mathbb{Z}}\sum_{\epsilon,\epsilon'} \epsilon \epsilon ' \mathrm{sgn}(2n-p)^p\binom{|2n-p|+p-1}{p-1} q^{\cdots} X^{-stn +\epsilon t/2  +pst/2}\cdot  X^{rs(-sn+\epsilon/2 +ps/2)+ \epsilon's/2  } $$
which we may rewrite as,
$$\sum_{m \in \frac{1}{2}\Z}\sum_{\epsilon,\epsilon'=\pm1} \epsilon \epsilon ' \mathrm{sgn}(m-\frac{\epsilon}{2} )^p\binom{\frac{|2m/s -\epsilon/s|+p}{2}-1}{p-1} \delta^\mathbb{Z}\left( \frac{2m-\epsilon-ps}{2s}\right)q^{m^2 (r+\frac{t}{s})+m\epsilon'} X^{m (rs + t)+ s\epsilon'/2}.$$
Therefore, we directly deduce,
$$f^m_{P'}(X) = \left[ \sum_{\epsilon=\pm 1} \epsilon \cdot \mathrm{sgn}(m-\frac{\epsilon}{2} )^p \binom{\frac{|2m -\epsilon|+ps}{2s}-1}{p-1} \delta^\mathbb{Z}\left( \frac{2m-\epsilon-ps}{2s}\right)\right] \cdot q^{m^2 \frac{t}{s} }X^{mt}  $$
which implies the result.
\end{proof}

Again, for instructive purposes, let us see how this agrees with \eqref{t,tp case}. Consider just $F^+_{P'}$ (only positive exponents). The $s=1$ case of Lemma \ref{patterndefiningtoruslink} reads,
$$f^m_{P'}(X) = \left[ \sum_\epsilon \epsilon \binom{\frac{2m -\epsilon+p}{2}-1}{p-1} \right]\delta^\mathbb{Z}\left( \frac{2m-1-p}{2}\right) \cdot q^{m^2 t }X^{mt}.$$
It follows from Pascal's identity that,
$$\binom{\frac{p}{2}-\frac{3}{2}+m}{p-2} = \binom{\frac{p}{2}-\frac{1}{2}+m}{p-1}-\binom{\frac{p}{2}-\frac{3}{2}+m}{p-1}.$$
So, indeed, they agree as required.

Now, we come to our main result,
\begin{theorem}[Cabling Formula]
    Let $K$ be a plumbed knot with GM series $F_K (X_0) = \sum_n f^n_K(q)X^n_0$. Then the $(tp,sp)$-cabling of $K$ ($gcd(s,t)=1$) is given by (up to an overall $\pm q^c$ factor),
    $$F_{C_{(tp,sp)}(K)} = \sum_{m\in \frac{1+sp}{2}} f^m_{t,s,p}(X) \cdot F_K(X_0 = q^{2m} \prod^p_{i=1}X^s_i) $$
    $$f^m_{t,s,p}(X) = \left[ \sum_\epsilon \epsilon \cdot \mathrm{sgn}(m-\frac{\epsilon}{2} )^p \binom{\frac{|2m -\epsilon|+ps}{2s}-1}{p-1} \delta^\mathbb{Z}\left( \frac{2m-\epsilon-ps}{2s}\right)\right] \cdot q^{m^2 \frac{t}{s} }(\prod^p_{i=1}X_i)^{mt}.$$
    Note that by $(tp,sp)$-cabling, we mean the pattern torus link is embedded in the curve $(tp,sp)$ in $T^2 \subset S^1 \times D^2$, where the $(0,1)$ cycle is contractible in $S^1 \times D^2$. 
\end{theorem}
\begin{proof}
    Simply apply Theorem \ref{satellite} and Lemma \ref{patterndefiningtoruslink}.
\end{proof}

We are also led to immediately conjecture, 
\begin{conjecture}
    The cabling formula above holds for any knot with a well-defined GM series, so long as the right-hand side converges.
\end{conjecture}

As a special case of the above conjecture, let us specify $p=1,s=2,t=2r+1$, so that we have $(2r+1,2)$-cabling. Again, for $F^+$, one can easily check,
$$F_{C_{(2r+1,2)}(K)} = \sum_{m\geq0}(-1)^m q^{(m+\frac{1}{2})^2 \frac{2r+1}{2}}X^{(m+\frac{1}{2})(2r+1)} F_K(X_0 = q^{2m+1}X^2).$$
This agrees with the results of \cite{John} and \cite{John2} whenever $K = 4_1$. 

We now move to a more speculative formula for the Whitehead double and its generalizations. 

\subsubsection{Examples: Twist Links and Whitehead Doubles}
Let us define,
\begin{align*}
&[n]= \frac{q^{\frac{n}{2}}-q^{-\frac{n}{2}}}{q^{\frac{1}{2}}-q^{-\frac{1}{2}}}\\
&[n]! = \prod^n_{i=1}[i]\\
& \begin{bmatrix}
    n\\k
\end{bmatrix} = \frac{[n]!}{[k]![n-k]!}.
\end{align*}

Then, from \cite{S2}, we may write the closed form of the GM series for the Borromean rings as, 
$$F_{Bor}(X_1,X_2,X_3) = \sum_{n\geq 0}(-1)^n q^{-\frac{3n^2+n}{2}} (q^{n+1},q)_n^2 \Phi_n(X_1)\Phi_n(X_2)\Phi_n(X_3) $$
where,
$$\Phi_n(X) = \sum_{k\geq 0} X^{n+k +\frac{1}{2} } \begin{bmatrix}
    2n+k\\2n
\end{bmatrix}.$$
Using the Partial Surgery conjecture \ref{partial_surg_conj}, we perform $r$ Rolfsen twists on one of the components of the Borromean rings to find the twist link $K_r$ depicted in figure \ref{twist_link_fig}.

\begin{figure}[H] 
\centering
\includegraphics[width=0.4\textwidth]{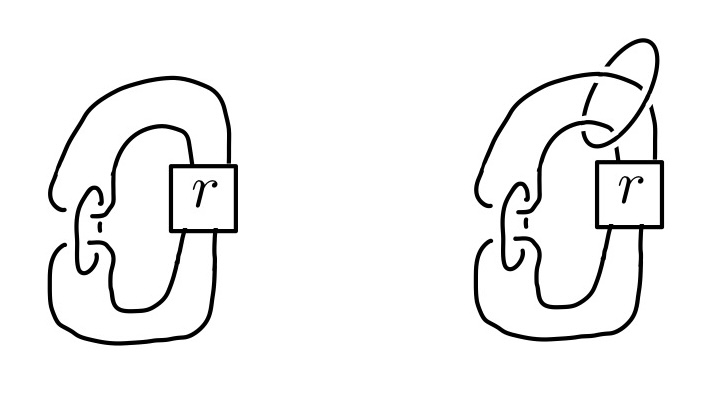}
\caption{The twist link $K_r$ is depicted on the left, while its pattern-defining link $K_r'$ is shown on the right. The box with $r$ inside indicates $r$ full twists.}
\label{twist_link_fig}
\end{figure}

Then one can solve for the pattern-defining link, $K_r'$, using topological invariance (and carrying out the same procedure as the previous section). The result for $F_{K_r'}$ is given by the Borromean rings formula above with the replacement, 
$$\Phi_n(X) \longmapsto \Phi^{r}_n(X) $$
$$\Phi^{r}_n(X) = \sum_{k\geq0} q^{(n+k+\frac{1}{2})^2r} X^{n+k+\frac{1}{2}}\begin{bmatrix}
    2n+k\\2n
\end{bmatrix}$$
in one of the $X_i$'s (does not matter which, Borromean rings are symmetric). We may rewrite the GM series for $F_{K_r}$ as,
$$F_{K_r'} = \underset{m\in \mathbb{Z}+\frac{1}{2}}{\sum_{m\geq0}} f^m_{K_r'}(X_1,X_2)\cdot X_3^m$$
\be f^m_{K_r'}(X_1,X_2) = q^{m^2r} \sum^{\lfloor m-\frac{1}{2}\rfloor}_{n\geq 0} \begin{bmatrix}
    n+m-\frac{1}{2}\\2n
\end{bmatrix} (-1)^n q^{-\frac{3n^2+n}{2}} (q^{n+1},q)_n^2 \Phi_n(X_1)\Phi_n(X_2).\ee 
Applying the satellite formula, we find that the GM series for the satellite of $K$ with pattern the $K_r$ link is, 
\be \label{double_1} F_{C_{K_r}(K)}(X_1,X_2) = \sum_{m}f^m_{K_r'}(X_1,X_2) F_K(X=q^{2m})\ee
or, 
\be \label{double_2}F_{C_{K_r}(K)}(X_1,X_2) = \sum_n f^n_{K}(q) F_{K_r'}(X_1,X_2,X_3=q^{2n}) \ee
where, again, we conjecture whichever of these (\eqref{double_1} or \eqref{double_2}) converges (if any) yields the correct formula. Of special note is the case $r=0$, in which case the satellite operation yields the \textit{Whitehead double}. At $r=0$, we see \eqref{double_2} is manifestly divergent and we must use \eqref{double_1} instead. Specifically, \eqref{double_1} naively converges if $F_K (X=q^\lambda)$ converges for all $\lambda$, in which case we expect $F_K (X=q^\lambda) = J_\lambda(q)$, for $J_\lambda(q)$ the colored $SU(2)$ Jones polynomial in representation $\lambda$. 

\subsubsection{Example: Seifert Fibrations over \texorpdfstring{$\Sigma_g$}{Sigma}}\label{seifert_sigma_g_section}
In Section \ref{seifert_section}, we found a simple formula for Seifert fibrations over $S^2$. Here, we will extend this formula to all genii $g$ using Conjecture \ref{general_link_wavefun_conj} and our partial surgery formalism described above. First, recall that $M\left(b,g; \frac{p_1}{r_1},\frac{p_2}{r_2},\ldots,\frac{p_d}{r_d}\right)$ has a surgery diagram displayed in figure \ref{genus_g_seifert_fig}.  

\begin{figure}[H] 
\centering
\includegraphics[width=0.6\textwidth]{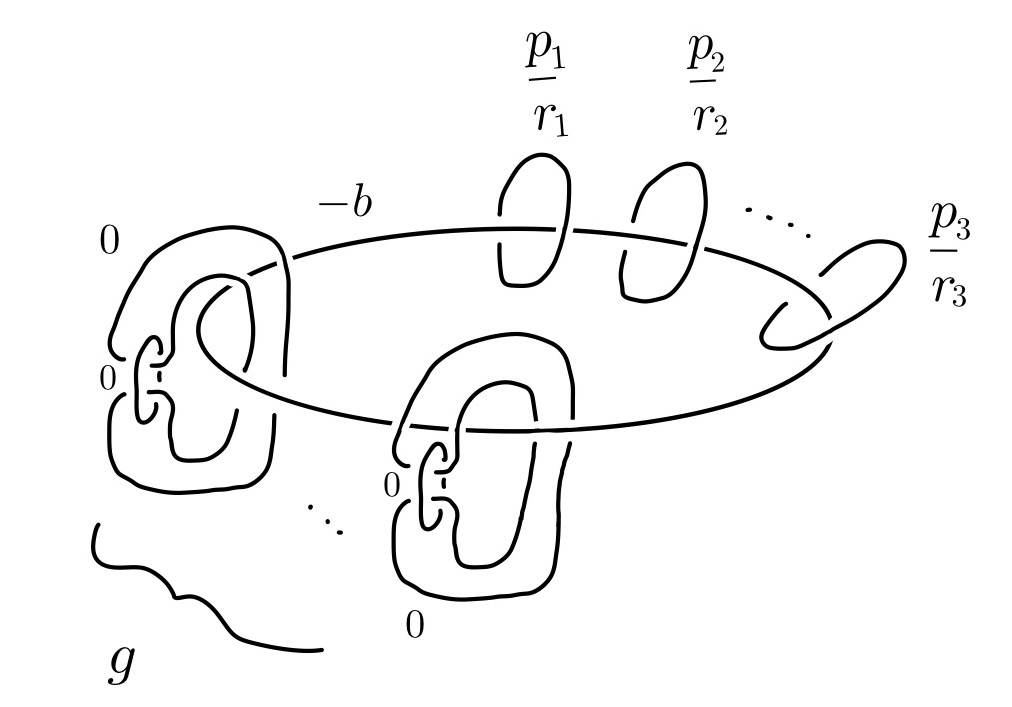}
\caption{Surgery diagram for $M\left(b,g; \frac{p_1}{r_1},\frac{p_2}{r_2},\ldots,\frac{p_d}{r_d}\right)$ as defined in Section \ref{seifert_section}. See, for instance, \cite{Hansen}.}
\label{genus_g_seifert_fig}
\end{figure}

We will denote by $\ket{L_{g,d},\mathbf{a}} \in \mathcal{H}(T^2)^{\otimes(d+1)}$ the wavefunctions for the link complement whose surgery (specified by figure \ref{genus_g_seifert_fig}) yields $M\left(b,g; \frac{p_1}{r_1},\frac{p_2}{r_2},\ldots,\frac{p_d}{r_d}\right)$. That is, $-b, \frac{p_1}{r_1},\ldots, \frac{p_d}{r_d}$ surgery on $L_{g,d}$ yields  $M\left(b,g; \frac{p_1}{r_1},\frac{p_2}{r_2},\ldots,\frac{p_d}{r_d}\right)$. 

\begin{figure}[H] 
\centering
\includegraphics[width=0.6\textwidth]{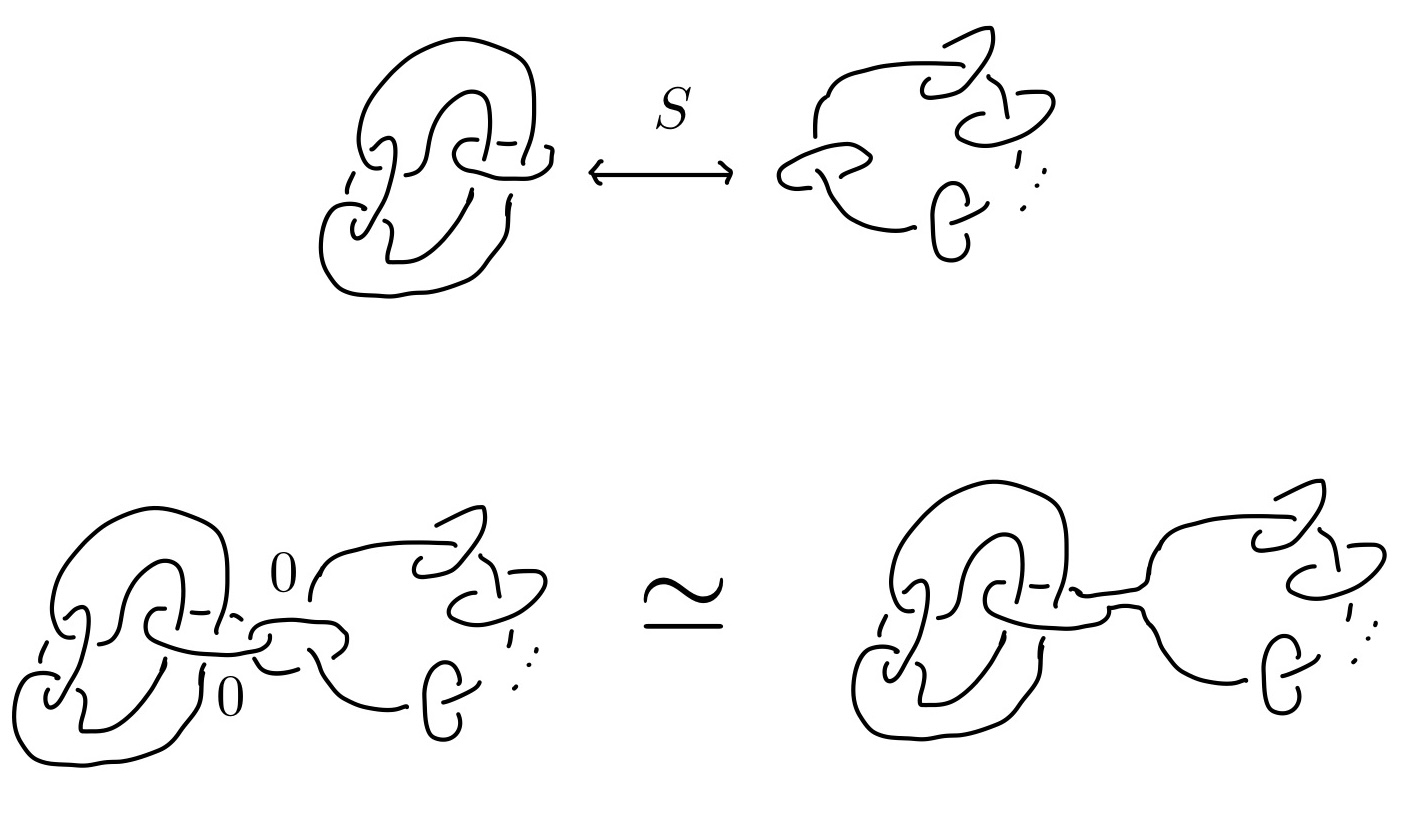}
\caption{$S$ transform between the two link components shown above yields a link complement of the form in figure \ref{genus_g_seifert_fig}}
\label{bor_rings_genus_g_fig}
\end{figure}

Let $\ket{Bor,\mathbf{a}} \in \mathcal{H}(T^2)^{\otimes3}$ denote the wavefunctions for the Borromean rings and $\ket{Bor_{0,0},(\mathbf{0},\mathbf{0},\mathbf{a}_3)} = \frac{1}{4} \bra{\mathbb{S},\mathbf{0}}\otimes \bra{\mathbb{S},\mathbf{0}} \ket{Bor,(\mathbf{0},\mathbf{0},\mathbf{a}_3)}$ the result of $0$ surgery on 2 of its components (it does not matter which, Borromean rings are symmetric). Then topological invariance requires that

$$\bra{Bor_{0,0},(\mathbf{0},\mathbf{0},\mathbf{0})}s \cdot S\ket{L_{0,d+1},(\mathbf{a_1,\ldots,a_d,0})} = \ket{L_{1,d},\mathbf{a}}.  $$

This is schematically portrayed in figure \ref{bor_rings_genus_g_fig}. Thankfully, \cite{S2} found a closed formula for $F_{Bor}$ in the form of  the `inverse Habiro series.' For us, it will only be important to mention,
$$F_{Bor}(X_1,X_2,X_3) = \sum_{m_1,m_2,m_3\in \mathbb{Z}+\frac{1}{2}} f^{m_1,m_2,m_3}_{Bor}(q) \cdot  X_1^{m_1} X_2^{m_2},X_3^{m_3}$$
\begin{align*}
&f^{m_1,m_2,m_3}_{Bor}(q) = f^{\sigma(m_2,m_1,m_3)}_{Bor}(q) \hspace{3mm} \forall\sigma\in S_3\\
&f^{+\frac{1}{2},m_2,m_3}_{Bor}=\mathrm{sgn}(m_2)\mathrm{sgn}(m_3).
\end{align*}

The linking matrix of $Bor$ is trivial, implying by means of Conj. \ref{general_link_wavefun_conj},
$$\ket{Bor,\mathbf{a}} = \underset{m\in \Z^3+\frac{1}{2}}{\sum_{n\in \Z^3+\mathbf{a}}} f^{\mathbf{m}}_{Bor} Y^\mathbf{n} X^{\mathbf{m}}.$$

Then, it follows,
\begin{align*}
\ket{Bor_{0,0},(\mathbf{0},\mathbf{0},\mathbf{a}_3)} & = \frac{1}{4} \bra{\mathbb{S},\mathbf{0}}\otimes \bra{\mathbb{S},\mathbf{0}} \ket{Bor,(\mathbf{0},\mathbf{0},\mathbf{a}_3)} \\ &= \frac{1}{4} \sum_{n-\mathbf{a}_3,m-\frac{1}{2}\in \Z} Y^n (\sum_{\epsilon_1,\epsilon_2=\pm1} \epsilon_1 \epsilon_2 f^{\epsilon_1 \frac{1}{2},\epsilon_2 \frac{1}{2},m}_{Bor} X^m)\\
&=\frac{1}{4} \sum_{n-\mathbf{a}_3,m-\frac{1}{2}\in \Z} Y^n (\sum_{\epsilon_1,\epsilon_2=\pm1} \epsilon_1^2 \epsilon_2^2f^{ \frac{1}{2}, \frac{1}{2},m}_{Bor} X^m)\\
&=\sum_{n-\mathbf{a}_3,m-\frac{1}{2}\in \Z} Y^n \mathrm{sgn}(m) X^m \\
&= \sum_{n-\mathbf{a}_3,m-\frac{1}{2}\in \Z} Y^n f^m_{3} X^{-m}.
\end{align*}

\begin{remark}
    In this step, we have picked the labels corresponding to the trivial cohomology element on the 0-surgeries of the components of the Borromean rings. Note, this does \textit{not} have to be the case. We do not, as of the writing of this paper, understand how to precisely move to non-trivial $\mathrm{Spin}^{c}$-structures whenever the summand is freely generated (that is, $\mathbb{Z}$). In any case, our wavefunctions $\ket{B_{0,0}}$ should, in truth, depend on this choice of $\mathrm{Spin}^{c}$-structure, as we are not in a manifold with trivial Betti number. We should also mention that factors of $2$ will show up in link surgeries, but they are easily traceable through Kirby invariance. As such, we do not mention them elsewhere in this text.
\end{remark}

The relevant tree link for the gluing displayed in figure \ref{bor_rings_genus_g_fig} is given by,
\begin{align*}
\ket{L_Q,\mathbf{a}} & = \underset{}{\sum_{n\in \Z^{d+2} + \mathbf{a}}} q^{-(Q\mathbf{n},\mathbf{n})} Y^{n_0}_0X_0^{n_1+n_2+\cdots+n_{d+1}}(X^{\frac{1}{2}}_0-X^{-\frac{1}{2}}_0)^{-d}\prod^{d+1}_{i=1}Y_i^{n_i}X^{n_0}_i\\
&=\underset{m_0\in \frac{1}{2}\Z}{\sum_{n\in \Z^{d+2} + \mathbf{a}}} q^{-(Q\mathbf{n},\mathbf{n})} Y^{n_0}_0 f^{m_0}_{d+2} X_0^{-m_0+n_1+n_2+\cdots+n_{d+1}} \prod^{d+1}_{i=1}Y_i^{n_i}X^{n_0}_i .
\end{align*}

After applying the $S$ transform, we see,
\be
\label{bor_ring_S_transf} \bra{B_{0,0},(\mathbf{0},\mathbf{0},\mathbf{0})}s \cdot S \cdot Y_i^{n_i}X_i^{n_0}\ket{0} = q^{2n_i n_0}f^{-n_i}_{3}. 
\ee
Let $i=d+1$ be the index of the leaf being glued to the Borromean rings component via the $S$ transform. Let $Q'$ be the linking matrix associated with the tree link resulting from the removal of the $d+1$ leaf. Similarly, let $\mathbf{n'} = (n_0,n_1,\ldots,n_d)$ be the new summation index. Then we see from \eqref{bor_ring_S_transf},
\begin{align*}
  & \bra{B_{0,0},(\mathbf{0},\mathbf{0},\mathbf{0})}s\cdot S\ket{L_Q,(\mathbf{a}',\mathbf{0})}   \endline
   = & \sum_{\mathbf{n'}\in \Z^{d+1}+\mathbf{a'},n_{d+1}} q^{-(Q'\mathbf{n'},\mathbf{n'})} Y^{n_0}_0 \left((X^{\frac{1}{2}}_0-X^{-\frac{1}{2}}_0)^{-d} f^{-n_{d+1}}_{3} X_0^{n_{d+1}}\right) X_0^{n_1+n_2+\cdots+n_d} \prod^{d}_{i=1}Y_i^{n_i}X^{n_0}_i \endline
   = & \sum_{\mathbf{n'}\in \Z^{d+1}+\mathbf{a'}} q^{-(Q'\mathbf{n'},\mathbf{n'})} Y^{n_0}_0 \left((X^{\frac{1}{2}}_0-X^{-\frac{1}{2}}_0)^{-d} (X^{\frac{1}{2}}_0-X^{-\frac{1}{2}}_0)^{-1}\right) X_0^{n_1+n_2+\cdots+n_d} \prod^{d}_{i=1}Y_i^{n_i}X^{n_0}_i.
\end{align*}
The expression above is naively nonsensical since multiplication is not defined in $\mathcal{H}(T^2)$. However, we can define multiplication in this case since both elements have well-defined inverses (in the sense described in Appendix \ref{high_rank_expansions}). As such, to the product of inverses, we assign the inverse of the product.  Concretely, we have that for any $p,q\in\mathbb{Z_+}$, we have,
$$(X^{\frac{1}{2}}-X^{-\frac{1}{2}})^{-p}(X^{\frac{1}{2}}-X^{-\frac{1}{2}})^{-q} \equiv  (X^{\frac{1}{2}}-X^{-\frac{1}{2}})^{-q-p} = \sum_{m} f^{m}_{p+q+2}X^{-m}.$$
Therefore,
$$\ket{L_{1,d},\mathbf{a}'} = \sum_{\mathbf{n'}\in \Z^{d+1}+\mathbf{a}'} q^{(Q'\mathbf{n'},\mathbf{n'})} Y^{n_0}_0 \cdot (X^{\frac{1}{2}}_0-X^{-\frac{1}{2}}_0)^{1-(d+2)}  \cdot X_0^{n_1+n_2+\cdots+n_d} \prod^{d+1}_{i=1}Y_i^{n_i}X^{n_0}_i.$$
Here, by $\mathbf{a}'$ we mean $(\mathbf{0},\mathbf{a}')$, where we take the $\mathbf{0}$ label on the $\Q/\Z$ factors of $H^1(M;\Q/\Z) = (\Q/\Z)^{b_1} \oplus Tor$. We can easily repeat this procedure $g$ times on a $ d+g$-pointed star graph to find,

$$\ket{L_{g,d},\mathbf{a}'} = \sum_{\mathbf{n'}\in \Z^{d+1}+\mathbf{a}'} q^{(Q'\mathbf{n'},\mathbf{n'})} Y^{n_0}_0 \cdot (X^{\frac{1}{2}}_0-X^{-\frac{1}{2}}_0)^{1-(d+2g)}  \cdot X_0^{n_1+n_2+\cdots+n_d} \prod^{d+1}_{i=1}Y_i^{n_i}X^{n_0}_i.$$

Now, for the $\Hat{Z}$ invariant of $M\left(b,g; \frac{p_1}{r_1},\frac{p_2}{r_2},\ldots,\frac{p_d}{r_d}\right)$, we repeat exactly the same procedure as the $g=0$ case in Section \ref{seifert_section} to find:
\begin{conjecture}
    Let $M= M(b,g;\frac{p_1}{r_1},\ldots,\frac{p_d}{r_d})$ denote the orientable Seifert fibration over $\Sigma_g$ of degree $b$ with singular fiber data $(\frac{p_1}{r_1},\ldots,\frac{p_d}{r_d})$. Let $\eta = -b -\sum^d_{i=1}\frac{r_i}{p_i}$ and suppose $\eta<0$ Fix a $Spin^c$ structure representative $(0,\alpha) \in \mathbb{Z}^{2g+\delta_0,b}\oplus TorH_1$ with $\alpha = (\alpha_0,\ldots,\alpha_d)$ and define $\xi_\alpha = \alpha_0 - \sum_i \frac{\mathbf{a_i}}{p_i}$. Then, up to an overall $\pm q^c$ factor, the $\hat{Z}$ invariant of $M$ is given by, 
    $$\Hat{Z}_{\alpha}(M(b,g;\frac{p_1}{r_1},\ldots,\frac{p_d}{r_d});q) \cong \sum_{n\in \mathbb{Z}}q^{-\eta n^2 -n \xi_\alpha} \cdot \Psi^n_{\Vec{p},\Vec{r},\alpha,g}$$
    where, 
    $$\Psi^n_{\Vec{p},\Vec{r},\alpha,g} = \underset{i=1,\ldots,d}{\sum_{\epsilon_i=\pm1}} \prod_i \epsilon_i \delta^\mathbb{Z}\left( \frac{2r_i n +\alpha_i+\epsilon_i}{2p_i}\right)\cdot \mathrm{sgn}(m)^d \binom{\frac{d}{2}+g-2+|m|}{d+2g-3}\delta^\mathbb{Z}\left(m+\frac{d}{2}\right)\bigg|_{m = \eta n_0 + \frac{\xi_\alpha}{2}- \sum_i\frac{\epsilon_i}{2p_i}}.$$
\end{conjecture}

\begin{remark}
    As evidence for the above conjecture, we specify to the case where $M$ is a degree $p$ $S^1$-bundle over $\Sigma_g$
$$S^1 \overset{p}{\rightarrow} M \rightarrow\Sigma_g.$$
If $T[M]$ denotes the 3d $\mathcal{N}=2$ SCFT resulting from compactification of $2$ coincident $M5$ branes on $M$, its superconformal half-index taken with certain 2d $\mathcal{N}= (0,2)$ boundary conditions is expected to coincide with that of a theory with $\mathcal{N}=2$ vector multiplet with Neumann b.c. and level $p$ CS term and $2g+1$ $\mathcal{N}=2$ chirals in the adjoint representation ($1+g$ with Neumann b.c. and $g$ with Dirichilet b.c.). The superconformal half index of this theory (which we emphasize is \textit{not} itself $T[M]$ \cite{3d3d}) agrees with the conjecture above \cite{GPPV}. 
\end{remark}

\section{Higher Rank \texorpdfstring{$\hat{Z}$}{Z}-TQFT}
$\hat{Z}$ for higher rank gauge groups was studied by Park in \cite{P3}. In this section, we present the TQFT for the higher rank $\hat{Z}$. While the previous sections focused on the $SL(2,\C)$ gauge group and the rational extension of the $\mathfrak{sl}_2$ non-commutative torus algebra, the framework generalizes naturally to arbitrary simply laced sem-simple Lie algebras. Since many of the techniques and calculations parallel those in the $\mathfrak{sl}_{2}$ case, we will omit them in this section. Consequently, this section will serve as a concise overview of the entire paper. 

We first fix some notation. Let $\mathfrak{g}$ be the semi-simple Lie algebra associated to the complexified gauge group of $G_\mathbb{C}$ Chern--Simons theory. 
\begin{itemize}
    \item By $R$, we mean the root system associated to $\mathfrak{g}$.
    \item $P$ and $P^\vee$ will denote the weight and coweight lattice of $R$, respectively. 
    \item $Q$ and $Q^\vee$ will denote the root and coroot lattice of $R$, respectively.
    \item $W$ will be the Weyl group of $G$. 
    \item $\{w_i\}$ will denote the fundamental weights and $\{\alpha_i\}$ will denote simple roots. Their duals will be given by $w_i^\vee$ and $\alpha_i^\vee$. 
    \item $m$ will denote the minimum integer such that $(P,P^\vee) = \frac{1}{m}\mathbb{Z}$.
\end{itemize}

\subsection{Rules of \texorpdfstring{$H^{1}(\cdot, \Q/\Z \otimes Q )$}{H1(.,Q/ZxQ)} decorated \texorpdfstring{$\mathrm{Spin}$}{Spin} TQFT}

In the case of higher rank gauge group, the $\hat{Z}^{\Q/\Z}(M)$ invariants are decorated by $H^{1}(M,\Q/\Z \otimes Q)$ and a $\mathrm{Spin}$-structure. Here we write down the rules of $H^{1}(\cdot, \Q/\Z \otimes Q )$ decorated $\mathrm{Spin}$ TQFT, thus generalizing the discussion in Section \ref{tqftrules}. Below, $K$ will denote some background field, but for the purposes of our discussions, in the remainder of the paper, we take $K$ to be the algebraic closure of the field of formal Laurent series in $q$, $\mathbb{C}((q)) $, so that $K =\mathbb{C}((q))^{alg} = \C_q$.

\begin{definition} We say that $Z$ is a decorated, oriented 3d $\mathrm{Spin}$ TQFT decorated by $H^{1}(\cdot,\Q/\Z\otimes Q)$ if $Z$ is the following collection of data subject to the following rules: 
\begin{enumerate}
    \item \textbf{State spaces.} To every oriented 2-manifold $\Sigma$ (possibly with punctures), $Z$ assigns a (possibly infinitely generated) $K$-module $Z(\Sigma)$.
    \item \textbf{Disjoint unions.} For disjoint surfaces, the $K$-module factorizes as, $Z(\Sigma_1 \bigsqcup \Sigma_2) = Z(\Sigma_1) \otimes Z(\Sigma_2)$.
    \item \textbf{Empty surface.} The $K$-module associated with the empty surface is the ground field. That is $Z(\Sigma) = K$.
    \item \textbf{Grading.} $Z(\Sigma)$ is graded by $H^{1}(\Sigma, \Q/\Z\otimes Q)$.
    \item \textbf{States.} To the tuple $(M,\mathfrak{s},h, \varphi)$, where $M$ is a three-manifold with boundary, $\mathfrak{s}$ is a $\mathrm{Spin}$-structure on $M$, $h \in H^{1}(M,\Q/\Z\otimes Q) $, and $\varphi$ is a choice of framing on $M$, we assign a vector (equivalently wavefunction), $$Z(M,\mathfrak{s},h, \varphi) = \ket{M,\mathfrak{s},h, \varphi} \in Z(\partial M).$$
    \item \textbf{Grading of states.} $Z(M,\mathfrak{s},h, \varphi)$ belongs to the $(h\vert_{\partial M} + 2 \varphi\vert_{\partial M}^{*}\circ i^{*}(\mathfrak{s})  \rho )$-graded subspace of $Z(\partial M)$, where $i$ is the inclusion map from the frame bundle of $\partial M$ to the frame bundle of $M$, and $\rho$ is the Weyl vector.
    \item \textbf{Mapping class group action.} To every oriented 2-manifold $\Sigma$, $Z$ assigns a representation of $MCG(\Sigma)$ on $Z(\Sigma)$: $$Z(\cdot):MCG(\Sigma) \rightarrow GL(Z(\Sigma)).$$ In fact, $$Z(M,\mathfrak{s}',h', \gamma \cdot \varphi) = Z(\gamma) \cdot Z(M,\mathfrak{s},h, \varphi), $$ 
    where $\mathfrak{s}', h'$ are the corresponding $\mathrm{Spin}$-structure and first homology element under the attachment of the mapping cylinder $M_\gamma$ to $M$. Due to this property, we will often ignore the framing of wavefunctions and consider only 0-framed wavefunctions:
    $$\ket{M,\mathfrak{s},h} = \ket{M,\mathfrak{s},h,0}.$$
    \item \textbf{Pairing.} $Z$ is endowed with a bilinear map $\langle \cdot |\cdot \rangle_\Sigma: Z(\Sigma)\times Z(\Sigma) \rightarrow K$. 
    \item \textbf{Gluing rule.} Suppose $M_{1}$ and $M_{2}$ are three-manifolds with boundary. Suppose $\Sigma \subset \partial M_{1}$ and $\Sigma \subset \partial M_{2}$. Suppose $M$ is obtained by gluing $M_{1}$ and $M_{2}$ along $\Sigma$ using a diffeomorphism in mapping class group element $\gamma$. Then 
$$\ket{M, \mathfrak{s}, h, \varphi} = \langle M_{1}, \mathfrak{s}\vert_{M_1}, h\vert_{M_2}|  \gamma | M_{2}, \mathfrak{s}\vert_{M_2}, h\vert_{M_2}\rangle_\Sigma.$$
\end{enumerate}
\end{definition}

\subsection{Torus State Space and \texorpdfstring{$SL(2,\mathbb{Z})$}{SL(2,Z)} automorphisms}

The quantized torus algebra associated to $\mathfrak{g}$ ,$QT_\mathfrak{g}$ is given by,
$$QT_\mathfrak{g} = \frac{\langle X^\mu ,Y^\lambda\rangle_{\mu\in P,\lambda\in P^\vee}}{(X^\mu Y^\lambda = q^{-4(\mu,\lambda)} Y^\lambda X^\mu)}.$$
This is an algebra over $\mathbb{C}_q$. This algebra carries a natural $SL(2,\mathbb{Z})$ action given by, 
\begin{align}
 \begin{pmatrix}
1 & 1 \\
0 & 1 
\end{pmatrix} &\leftrightarrow \tau_{+}: \hspace{0.5cm} X^{w_i} \mapsto  Y^{w_i^\vee} X^{w_i} q^{-2(w_i,w_i)} \hspace{1cm}  Y^\lambda \mapsto Y^\lambda, \\
\begin{pmatrix}
1 & 0 \\
1 & 1 
\end{pmatrix} &\leftrightarrow \tau_{-}: \hspace{0.5cm}  X^\mu \mapsto X^\mu \hspace{3.2cm}   Y^{w_i^\vee}  \mapsto q^{2(w_i,w_i)} X^{w_i}  Y^{w_i^\vee}.
\end{align}
Our conventions are such that when $\mathfrak{g} = \mathfrak{sl}_2$, $(w_1,w_1) = \frac{1}{2}$ and the relations above become that of Section \ref{sl2z_section}. Moreover, the set $\{Y^\lambda, X^\mu\}_{\mu\in P,\lambda\in P^\vee}$ provides a basis for $QT_\mathfrak{g}$ as a free $\mathbb{C}_q$-module. It will prove useful to know the action of $SL(2,\mathbb{Z})$ on this basis. The action of generators $\tau_{\pm}$ on the basis above is described in the following lemma. 
\begin{lemma}\label{tau_mat_lemma} Let $\tau_-$ and $\tau_+$ be as above, then,
\begin{align*}
   \tau_-^{p}(Y^\lambda X^\mu ) = & q^{-2p(\lambda,\lambda)} Y^\lambda X^{\mu+p\lambda}, \\
   \tau_+^{a}(Y^\lambda X^\mu ) = & q^{-2a(\mu,\mu)} Y^{\lambda +\mu a}X^{\mu}.
\end{align*}
\end{lemma}
\begin{proof}
    We will prove only the $\tau_+$ relation, as the proof for $\tau_{-}$ is almost identical. Denote $X_i = X^{w_i}$ and $Y_i = Y^{w_i^\vee}$. Then notice that,
$$\tau_+(X^\mu) = \prod^r_{i=1}\tau_+(X_i)^{\mu_i} = q^{-2\sum_i \mu_i(w_i,w_i)} \prod^r_{i=1}(Y_i X_i)^{\mu_i}.$$
Repeatedly using the $q$ commutation relation, we see that, 
$$(Y_i X_i)^{\mu_i} = q^{-2(\mu_i-1)\mu_i (w_i,w_i)}Y^{\mu_i}X^{\mu_i}.$$
As we commute the $Y$'s all the way to the left, we notice that for $j<i$,
$$X^{\mu_j}_jY^{\mu_i}_{i} = q^{-4\mu_i \mu_j (w_j,w_i)}Y^{\mu_i}_{i} X^{\mu_j}_j.$$
Putting these together, we get,
\begin{align*}
 \tau_+(X^\mu) = q^{-2 \sum^r_{i=1} \mu_i^2 (w_i,w_i) -4 \underset{j<i}{\sum^r_{i,j=1}} \mu_i\mu_j (w_i,w_j)}Y^\mu X^\mu  = q^{ -2(\mu,\mu)}Y^\mu X^\mu . 
\end{align*}
Then, 
$$\tau_+^a(X^\mu) = \tau^{a-1}_+(q^{-2(\mu,\mu)}Y^\mu X^\mu ) = \tau^{a-2}_+(q^{-4(\mu,\mu)}Y^{2\mu} X^\mu ) = \cdots =q^{-2a(\mu,\mu)}Y^{a\mu} X^\mu .$$
The statement follows immediately.
\end{proof}
For a general element of $SL(2,\Z)$, the action of the above basis is given by the following lemma.
\begin{lemma}\label{sl2z_act_lemma}
    Let $\gamma = \begin{pmatrix}
        b&a\\
        p&r
    \end{pmatrix} \in SL(2,\mathbb{Z})$. Then,
    $$\gamma(Y^\lambda X^\mu) = q^{-2(\mu,\mu)ab -4 (\mu,\lambda)ap -2(\lambda,\lambda)pr} Y^{\lambda r+\mu a } X^{\lambda p + \mu b}.$$    
\end{lemma}
\begin{proof}
    After repeatedly applying Lemma \ref{tau_mat_lemma}, this is a standard proof by induction.
\end{proof}

We aim to extend this algebra to define a module, which will then be used to construct the vector space of our theory on a torus. There is a very natural way to do this. The weight lattice is usually written as,
$$P = \mathbb{Z}w_1 + \mathbb{Z}w_2 +\cdots +\mathbb{Z}w_r,$$
such that $\mu \in P$ is always of the form $\mu = \sum_i \mu_i w_i$ for $\mu_i \in \mathbb{Z}$. Then the obvious extension is given by,
$$P_F = P\otimes_\mathbb{Z}F = Fw_1 + Fw_2 +\cdots +Fw_r,$$
where $F=\mathbb{Q}$ or $\mathbb{R}$. Then the relevant algebra of operators is given by,
$$\boxed{\mathcal{O}^\mathfrak{g}_{F} = QT^{F}_\mathfrak{g} = \frac{\langle X^\mu ,Y^\lambda\rangle_{\mu,\lambda\in P_F}}{(X^\mu Y^\lambda - q^{-2(\mu,\lambda)} Y^\lambda X^\mu)},}$$
where we have relabeled $q^2\rightarrow q$ for convenience. The commutation relation easily extends to $P_F$ by inspection. For the $SL(2,\mathbb{Z})$ action, we use Lemma \ref{sl2z_act_lemma} to extend the action.
$$\gamma(Y^\lambda X^\mu) = q^{-(\mu,\mu)ab -2 (\mu,\lambda)ap -(\lambda,\lambda)pr} Y^{\lambda r+\mu a } X^{\lambda p + \mu b}, \hspace{4mm} \mu,\lambda \in P_F.$$
For our immediate purpose, unless otherwise stated, we will take $F=\mathbb{Q}$. See section \ref{torus_hilbertspace} for a discussion of this dichotomy. 

Now, we can construct the $\mathcal{O}^\mathfrak{g}_F$-module, $$\hat{\mathcal{O}}^\mathfrak{g}_F = \mathbb{C}_q \langle \{  Y^\lambda X^\mu\}_{\lambda,\mu \in P_F}\rangle,$$
which consists of $\C_q$-linear combinations of elements $Y^\lambda X^\mu$ (allowing for infinite sums). Under the $SL(2,\Z)$ action described above, $\hat{\mathcal{O}}^\mathfrak{g}_F$ becomes a $SL(2,\Z)$ representation. Furthermore, we have:
\begin{lemma}
$\hat{\mathcal{O}}^\mathfrak{g}_\mathbb{Q}$ forms a faithful representation of $SL(2,\mathbb{Z})$ under the action,
$$\gamma(Y^\lambda X^\mu) = q^{-2(\mu,\mu)ab -2 (\mu,\lambda)ap -(\lambda,\lambda)pr} Y^{\lambda r+\mu a } X^{\lambda p + \mu b}, \hspace{4mm} \mu,\lambda \in P_F.$$
\end{lemma}
\begin{proof}
    The set of elements $\{ Y^\lambda X^\mu\}_{\mu,\lambda\in P_\mathbb{Q}}$ are a basis for $\mathcal{O}^\mathfrak{g}_\mathbb{Q}$ as a free $\mathbb{C}_q$-module, so the statement follows from discussion above.
\end{proof}

With this, we are ready to define our vector space $\mathcal{H}(T^{2})$, which is, 
$$\mathcal{H}(T^{2}) = \hat{\mathcal{O}}^\mathfrak{g}_F,$$
and is spanned by elements of the form,
\be \label{states_general_g} \ket{\lambda,\mu} \equiv Y^\lambda X^\mu \ket{0}\hspace{4mm} \mu,\lambda \in P_F,\ee
where, again, $\ket{0}$ is some vacuum state.

We give the bilinear map by specifying it on the basis of $\mathcal{H}(T^{2})$,
$$\bra{\lambda_1,\mu_1}\ket{\lambda_2,\mu_2} = \delta_{\lambda_1,\lambda_2}\delta_{\mu_1,\mu_2}.$$
Equivalently, we can take formal contour integrals that pick out constant term expressions in both $X$ and $Y$, 
$$\bra{\Psi_1}\ket{\Psi_2} = \oint \frac{dX}{2\pi i X}\oint \frac{dY}{2\pi i Y} \Psi_1(X,Y)^\dagger \Psi_2(X,Y),$$
where the states $\ket{\Psi} = \Psi(X,Y)\ket{0}$ are given by $\Psi(X,Y) = \sum_{\mu,\lambda\in P_\mathbb{Q}} \psi_{\lambda,\mu}(q) Y^\lambda X^\mu $ so that,
$$\Psi(X,Y)^\dagger = \sum_{\mu,\lambda\in P_\mathbb{Q}} \psi_{\lambda,\mu}(q) X^{-\mu} Y^{-\lambda} .$$
Following our discussion in the previous sections, we must pick a suitable wavefunction for the solid torus. \footnote{Suitable here meaning invariant under the right $SL(2,\mathbb{Z})$ subgroup.} Arguing as in Section \ref{solid_torus_sec}, we pick, 
\be \label{highranksolidtorus} \ket{\mathbb{S}, \mathfrak{s}, h} = \sum_{\substack{ \lambda\in Q \\ w\in W }}  \epsilon(w) \ket{\lambda + 2 \rho \mathfrak{s}(m) + h(m) ,w(\rho)},\ee 
where $\rho = \sum_i w_i$ is the Weyl vector and $\epsilon(w)$ is the signature of $w\in W$. We will, for the remainder of this section, unify $(h,\mathfrak{s})$ into one label $\mathbf{a}$ as in Section \ref{spinstructures_review}, so that we may write, 
$$Z(\mathbb{S}, \mathbf{a}) =\ket{\mathbb{S}, \mathbf{a}} = \sum_{\substack{ \lambda\in \mathbf{a} + Q \\ w\in W }}  \epsilon(w) \ket{\lambda,w(\rho)}.$$

\subsection{Amplitudes and States in \texorpdfstring{$\mathcal{O}^\mathfrak{g}_\mathbb{Q}$}{OgQ}}
In section \ref{amplitudesandstates} we showed that for the $\mathfrak{g} = \mathfrak{sl}_2$ case, given $\mathcal{H}(T^{2})$ endowed with the $SL(2,\mathbb{Z})$ action, the solid torus wavefunction $\ket{\mathbb{S},h,\mathfrak{s}}$, and topological invariance \textbf{uniquely} determines the amplitudes for plumbed 3-manifolds. Though we will not carefully document the calculations as before, the amplitude structure for general $\mathfrak{g}$ is also uniquely determined by these choices.

\begin{proposition}
     Let $L_M$ be an $N$-component tree link specified by the linking matrix $M$ (taken with 0 diagonal). $Z(S^3 \setminus \nu(L_M),\mathbf{a})$ is uniquely determined by topological invariance and takes the form,
\be \label{treelink_general_g} 
\ket{L_{M}, \mathbf{a} } = \sum_{\lambda \in \mathbf{a} + Q^N} q^{-(M\lambda,\lambda)} \prod^N_{i=1} Y^{\lambda_i}_{i} X_i^{(M\lambda)_i} \left(\sum_w \epsilon(w)X_i^{w(\rho)}\right)^{1-\mathrm{deg}(v_{i})} .
\ee 
Moreover, integer surgery on every component recovers the higher rank plumbing formula for the $\Hat{Z}$ invariant. That is, if $X$ is the closed 3-manifold resulting from $(p_1,p_2,..,p_N)$ surgery on $L_M$, then
\begin{align}
  Z(X,\mathbf{a}) =\Hat{Z}^\mathfrak{g}_\alpha(X;q) & = \oint  \prod_{i}   \frac{dX_i}{2\pi i X_i}  \prod_{i} \left(\sum_w \epsilon(w)X_i^{w(\rho)}\right)^{2-\mathrm{deg}( v_i)}  \sum_{ \ell \in  \alpha + M^{f} Q^{N }} q^{-(\ell,(M^f)^{-1}\ell)} \prod_{i} X_{i}^{\ell},
\end{align}
where $Bk(\mathbf{a}) =\alpha = (\alpha_1,\alpha_2,\ldots,\alpha_N)$ denotes the generalized $Spin^c$ structure, $N$ the number of vertices, and $M^f$ is the linking matrix with framing included $M^f = M + diag(p_1,p_2,..,p_N)$. 
\end{proposition}

These invariants were originally defined and studied in \cite{P3}, where higher rank GM series (which we denote by $F^\mathfrak{g}_F$ for a knot $K$) were also derived from the plumbing formula. For instance, the higher rank GM series for all torus knots $T(t,s)$ was found in \cite{P3}. From our point of view, the generalization to $(ps,pt)$ torus links is easily computable via,
$$F^\mathfrak{g}_{T(pt,ps)} = C.T._Y \left[ \bra{\mathbb{S},\mathbf{0}}\otimes\bra{\mathbb{S},\mathbf{0}} (\gamma^\dagger_{b/s}\otimes\gamma^\dagger_{t/a})\ket{L,\mathbf{0}}\right],$$
where $L$ and the framings are as in Section \ref{torus_links_section}. More generally, we can also compute wavefunctions of any plumbed knot \textit{or} link $L$ (that yield a convergent $q$-series), which are restricted to be of the form,
$$\ket{L,\mathbf{a}} = \underset{}{\sum_{\lambda \in Q^N+\mathbf{a}}} q^{-(M\lambda,\lambda)} \cdot \prod^N_{i=1} Y^{\lambda_i}_{i}  X_i^{(M\lambda)_i} F^\mathfrak{g}_L(X;q).$$

The above wavefunctions are invariant under Rolfsen moves. This can be seen by repeating the arguments of Section \ref{partial_surg_section}. A rational surgery formula for the GM series of knots was conjectured in \cite{P3} (and proved for plumbed knot complements). It is given by, 
$$\Hat{Z}^\mathfrak{g}_\alpha (S^3_{p/r}(K),q) \cong \mathcal{L}^\alpha_{p/r}\left((\sum_{w\in W}\epsilon(w)X^{\frac{w(\rho)}{r}})\cdot F^\mathfrak{g}_K\right),$$
and the Laplace transform here is defined by,
$$\mathcal{L}^\alpha_{p/r} (X^\mu) = q^{(\mu,\mu)\frac{r}{p}} \cdot \delta^Q\left(\frac{r\mu-\alpha}{p}\right),$$
where, analagously to the $\mathfrak{sl_2}$ case with $\delta^\mathbb{Z}$, we define,
$$\delta^Q(x) = \begin{cases}
    1 &x\in Q\\
    0 & \text{else.}
\end{cases} $$

We will show that the $SL(2,\mathbb{Z})$ action on our vector space leads to this surgery formula as well. First, however, we determine the higher rank gluing formula. Let $K_1$, $K_2$ be knots with well-defined higher rank GM series,
$$F^\mathfrak{g}_{K_i} = \sum_{\mu\in Q+\rho} f^\mu_{K_i}(q) \cdot X^\mu,$$
then the result for gluing is as expected:
\begin{proposition}
    Let $\gamma = \begin{pmatrix}
        r&a\\ p&b
    \end{pmatrix} \in SL(2,\mathbb{Z})$ be the element of $MCG(T^2)$ determining the 3-manifold, $M$, resulting from the gluing of two knot complements: $M = (S^3 \setminus\nu K_2)\cup_\gamma (S^3 \setminus\nu K_1)$. Furthermore, assume $K_1$ and $K_2$ are plumbed knots so that their wavefunctions are well-defined. Then,
    \be \label{higher_rank_splicing_eq} Z(M,\mathbf{a}) = \bra{K_2,\mathbf{a}_2}\gamma\ket{K_1,\mathbf{a}_1} \cong \boxed{ \sum_{\mu_1,\mu_2}f^{\mu_2}_{K_2}\cdot f^{\mu_2}_{K_2} \cdot q^{-\frac{r}{p} (\mu_2,\mu_2) + \frac{2}{p}(\mu_2,\mu_1) -\frac{b}{p} (\mu_1,\mu_1)} \cdot \delta^Q\left( \frac{r \mu_2 -\mu_1 - \alpha}{p}\right).} \ee

    Whenever the right-hand side converges. Here, $\mathbf{a}_1$ and $\mathbf{a}_2$ are the restrictions of $\mathbf{a}$ to the corresponding knot complement and $Bk(\mathbf{a}) = \alpha \in Spin^c_\mathfrak{g}(M) \cong Q/pQ + \rho \cdot p\cdot \delta_{0,r\mod 2}$. Of course, in the case where $M$ has a weakly negative plumbing description, we have,
$$\hat{Z}^\mathfrak{g}_{\alpha}(M) = Z(M,\mathbf{a}) .$$
\end{proposition}
\begin{proof}
    The proof is essentially identical to the one in Theorem \ref{splicing_theorem}.
\end{proof}

The surgery formula of \cite{P3} is then immediately implied by letting $K_1 = O$, where, 
$$F_O(X) = \sum_{w\in W}\epsilon(w) X^{w(\rho)}.$$
Then, the result is: 
\begin{corollary}
     Let $K$ be a plumbed knot with well-defined wavefunctions and GM series $F_K^\mathfrak{g}$, and let $\gamma = \begin{pmatrix}
        b&a\\ p&r
    \end{pmatrix} \in SL(2,\mathbb{Z})$ be the element of $MCG(T^2)$ determining the 3-manifold, Then,
    \be Z(M,\mathbf{a}) = \bra{K_2,\mathbf{a}_2}\gamma\ket{\mathbb{S},\mathbf{a}_1} =\mathcal{L}^\alpha_{p/r}\left((\sum_{w\in W}\epsilon(w)X^{\frac{w(\rho)}{r}})\cdot F^\mathfrak{g}_K\right),  \ee
    where again, $Bk(\mathbf{a}) = \alpha$. In particular, if $M$ is a weakly negative definite plumbed 3-manifold, then, 
    $$Z(M,\mathbf{a}) \cong \Hat{Z}^\mathfrak{g}_\alpha (M) .$$
\end{corollary}
\begin{proof}
    Notice that, 
$$\mathcal{L}^\alpha_{p/r}\left((\sum_{w\in W}\epsilon(w)X^{\frac{w(\rho)}{r}})\cdot F^\mathfrak{g}_K\right) = \underset{\mu}{\sum_{w\in W}} \epsilon(w) f^\mu_K \cdot q^{-\frac{r}{p}(\mu,\mu) -\frac{2}{p}(\mu,w(\rho)) -(w(\rho),w(\rho))\frac{1}{pr}} \cdot \delta^Q\left( \frac{r \mu_2 +w(\rho) - \alpha}{p}\right). $$
Since $(w(\rho),w(\rho)) = (\rho,\rho)$, we see this agrees with \eqref{higher_rank_splicing_eq} up to an overall $q$ factor.
\end{proof}

\begin{example}[Seifert Manifolds]
 Just as in the $\mathfrak{g}= \mathfrak{sl}_2$  case, we can also compute the amplitude for Seifert manifolds using the method in Section \ref{seifert_section}. Let $M= M(b,g;\frac{p_1}{r_1},\ldots,\frac{p_d}{r_d})$ denote the orientable Seifert fibration over $\Sigma_g$ of degree $b$ with singular fiber data $(\frac{p_1}{r_1},\ldots,\frac{p_d}{r_d})$. Let $\eta = -b -\sum^d_{i=1}\frac{r_i}{p_i}$ denote its Euler number and suppose $\eta<0$. Fix a $Spin^c$ structure representative $(0,\alpha) \in \left(\mathbb{Z}^{2g+\delta_0,b}\oplus TorH_1\right)\bigotimes Q+\delta \cdot \rho$ with $\alpha = (\alpha_0,\ldots,\alpha_d)$ and define $\xi_\alpha = \alpha_0 - \sum_i \frac{\alpha_i}{p_i}$. Then, 

$$\boxed{ \Hat{Z}^\mathfrak{g}_\alpha\left( M\left(b,g;\frac{p_1}{r_1},\ldots,\frac{p_d}{r_d}\right) \right) \cong \sum_{\mu \in Q}q^{-\eta (\mu,\mu) -(\mu, \xi_\alpha)} \cdot \prescript{}{\mathfrak{g}}{\Psi}^\mu_{\Vec{p},\Vec{r},\alpha,g},}$$
where, 
$$\prescript{}{\mathfrak{g}}{\Psi}^\mu_{\Vec{p},\Vec{r},\alpha,g} = \underset{i=1,\ldots,d}{\sum_{w_i\in W}} \prod_i \epsilon(w_i) \delta^Q\left( \frac{r_i \mu +\alpha_i+w_i(\rho)}{2p_i}\right)\cdot f^{\eta \mu  + \frac{\xi_\alpha}{2}- \sum_i\frac{w_i(\rho)}{p_i}}_{d+2g+2}, $$
and $f^\mu_p$ is as in \eqref{high_rank_f} in Appendix \ref{high_rank_expansions}. The case $g=0$ is a Theorem. For $g>0$ this is a conjecture based on the proposal in Section \ref{seifert_sigma_g_section}.    
\end{example}

\section{Conclusion and Future Directions}

\subsection{Towards DAHA}
The algebras of operators we have used to (partially) define the decorated TQFTs in this paper are a certain extension of $t=1$ Spherical Double Affine Hecke Algebras (sDAHA). It is natural to ask: how is the full sDAHA (at arbitrary $t$) related to our TQFTs? In this section, we provide some clues to the answer. 

From the superstring perspective, the $\hat{Z}$ Gukov--Manolescu invariants are described by the following brane construction. For a 3-manifold $M$, the M-theory spacetime is:
$$    S^1 \times T^*M \times TN$$
with $TN$ the Taub-NUT spacetime being twisted by $S^1$ as one goes around $z_1\rightarrow q z_1$, $z_2 \rightarrow \tau^{-1} z_2$. For now, we take $\tau=q$. As in \cite{AS}, we denote the generators of these rotations as $S_1$ and $S_2$, respectively, when represented in the M-theory Hilbert space. One then places $N$ M5-branes on:
$$S^1 \times M \times D^2.$$
One finds the $U(N)$/$SU(N)$ $\hat{Z}$ invariants \footnote{As usual, the center of mass mode of the $M5$-brane stack decouples and the $U(N)$ and $SU(N)$ partition functions are the same up to a factor.} for closed 3-manifolds by computing the $M5$ brane index \cite{GPPV},

\be \Hat{Z}_a(M;q) = Tr_{Z_a} (-1)^F q^{S_1 - S_2}\ee 

where $Z_a$ denotes the space of M2-M5 BPS bound states where the boundary of the M2s are restricted to wrap $a$ 1-cycles in $M$. Equivalently, if we partially twist along $M$ with the residue R-symmetry of the worldvolume theory and compactify on $M$, one finds $T[M]$ in the infrared, a 3d $\mathcal{N}=2$ SCFT whose signed partition function on $S^1\times D^2$ (with certain 2d $\mathcal{N}= (2,0)$ boundary conditions on $\d D^2$) is just the so-called half-index (which is $\hat{Z}$ up to a normalization).

So far, we have been working with the ``unrefined'' DAHA $\tau = t q = q$, where the wavefunctions admit significant simplifications. More importantly, from the M-theory viewpoint, we only expect such a refinement by $\tau = q^{\beta+1}$ to exist in the instance that the 3-manifold $M$ admits a semi-free action by $U(1)$ (that is, $M$ is Seifert). Such 3-manifolds allow us to define a nowhere vanishing vector field $V$ in $T^*M$, that generates this $U(1)$ action. Since the M5-branes (whose index computes $\hat{Z}$) wrap the Lagrangian submanifold $M\subset T^* M$, $V$ allows us to see the action of a global symmetry $U(1)$ rotating the fibers of $T^* M$. This is a global symmetry of the worldvolume theory in the sense that it is a symmetry of spacetime acting transverse to the M5-brane directions.  It is this $U(1)$-symmetry that defines a refined supersymmetric index. That is, it allows us to find mutually commuting charges $S_1'$ and $S_2',$ by mixing $S_1,S_2$ and this new $U(1)$-symmetry associated to Seifert manifolds (see \cite{AS} for the details of this). These refined indices are the refined $\hat{Z}$ invariants of \cite{GPPV},

\be \label{index_def_Zhat}\hat{Z}_a(M;q,t) = Tr_{Z_a} (-1)^F q^{S'_1} \tau^{- S'_2}.\ee 

Let us restrict our attention to the solid torus, $M = S^1\times D^2$. As in \cite{AS}, we may reduce on the M-theory circle to Type IIA superstring theory, where \eqref{index_def_Zhat} is now a count of D2-D4 BPS particles.\footnote{We remark that this cannot be done in general. Reducing on the M-theory circle is a delicate procedure \cite{GHP}, and it is somewhat unexpected that this trick works here. We thank Sergei Gukov for pointing this out to us.} This system is more tractable, and the calculation can be directly completed. One finds, 

\be \label{refsolidtorus} \hat{Z}_0 (\mathbb{S};X,q,t) = \prod_{m=0}^\beta \prod_{1\leq i < j \leq N} \left( q^{-m/2} X_j ^{1/2}X_i ^{-1/2}-q^{m/2} X_i ^{1/2}X_j ^{-1/2} \right)\ee

where $N$ is the number of coincident M5-branes giving rise to $U(N)$ refined Chern--Simons dual to the refined open topological A-model. For the remainder of this section, we will work with $SU(2)$ gauge group for simplicity, in which case, \eqref{refsolidtorus} is,

$$\hat{Z}_0 (\mathbb{S}_0;X,q,t) = \prod_{m=0}^\beta \left( q^{-m/2} X ^{-1/2}-q^{m/2} X ^{1/2} \right) =q^{-\frac{\beta(\beta+1)}{4}}X^{-(\beta+1)/2} (X,q)_{\beta+1}.$$
It is worth noting that, 
$$\hat{Z}_0 (\mathbb{S}_0;X,q,t) \hat{Z}_0 (\mathbb{S}_0;X^{-1},q,t) = q^{-\frac{\beta(\beta+1)}{2}} \frac{(X,q)_\infty (X^{-1},q)_\infty}{(Xqt,q)_\infty (X^{-1}qt,q)_\infty}.$$
The earlier sections of this paper teach us that we must define a wavefunction for the solid torus in the full $\mathfrak{sl}_2$ $sDAHA$ algebra. That is, let $\ddot{H}$ denote the Double Affine Hecke Algebra \cite{C2} of type $A_1$. By $s\ddot{H}$, we will then mean its spherical subalgebra. There, $X$ and $Y$ do not $q$-commute any longer. Rather, they obey
$$XY = q^{-2}YX T^2.$$
Note that at $t=1$, $T^2=1$, and we recover the unrefined results. The $SL(2,\mathbb{Z})$ action on $sDAHA$ is the same as before (see Section \ref{sl2z_section}), which implies that the basis $Y^n X^m \ket{0}$ are no longer permuted by $SL(2,\mathbb{Z})$ automorphisms if $t$ is generic. It is convenient for us to work with a family of elements in $s\ddot{H}$ for which the elements of the mapping class group act as permutations up to a $q$-factor. To this end, we define,
$$\psi_{n,0} = e Y^n e$$ 
where $e$ is the indempotent used to define $s\ddot{H} = e \ddot{H}e$ (for a semi-simple lie algebra $\mathfrak{g}$ with Weyl group $W$, this would be $e = \frac{\sum_{w\in W}t_w T_w}{\sum_{w\in W}t_w^2}$ ). Let $\gamma = \begin{pmatrix}
    r&a\\p&b
\end{pmatrix} \in SL(2,\mathbb{Z})$. Then, we define,
$$\psi_{nr,np} = q^{n^2p r}\gamma(\psi_{n,0}).$$
Since $gcd(r,p)=1$, this uniquely specifies all elements $\psi_{n,m} \in s\ddot{H}$. Then, we define the states,
$$\ket{n,m} = \psi_{n,m}\ket{0}.$$
For instance, some simple examples of these elements in the PBW basis would be,
\begin{align*}
 q^{-2p}\psi_{1,p} &= q^{-p} \cdot \tau^p_-(\psi_{1,0}) = e \cdot  X^p Y\cdot e \\
 q^{-8} \psi_{2,2} &=  e\cdot X^2 Y^2 -\tilde{t} X^{-2}T \cdot e \\
 q^{-16} \psi_{2,4} &=   e\cdot X^4Y^2 -\tilde{t} X^4 T -\tilde{t} q^{-4} X^2T\cdot e\\
 q^{-2}\psi_{-1,-1} &=   e\cdot X^{-1} Y^{-1} -\tilde{t} XY (T-\tilde{t})\cdot e\\
 q^{-18}\psi_{3,3} &=   e\cdot X^3 Y^3 - \tilde{t} X^3(Y+Y^{-1})T -\tilde{t} q^{-4} XY^{-1} T\cdot e\\
 q^{-36} \psi_{3,6} &=   e\cdot X^6 Y^3 - \tilde{t} X^6(Y+Y^{-1})T -\tilde{t} q^{-4} X^4Y^{-1} T  -\tilde{t} q^{-8}(X^4YT +X^{2} Y^{-1}T -\tilde{t} X^4Y)  \\ & \hspace{0.7cm} - \tilde{t}q^{-12} (Y^{-1} T -\tilde{t} X^2Y)\cdot e,
\end{align*}
where $\tilde{t} = t-t^{-1}$. Due to the growing complexity of these elements in the PBW basis, one is tempted to conclude that they may be more suitable than the PBW basis in defining wavefunctions inside the torus Hilbert space. That is, by construction, these elements fulfill 
$$\gamma \ket{n,m} = q^{-n^2pr -2m n ap -m^2 ab}\ket{nr+ma,pn+mb}.$$
At this point, we can do the most naive thing imaginable, which is to repeat our procedure for the unrefined $\hat{Z}$ theory by considering a wavefunction for the solid torus of the form
\be \label{wrong_ref_torus_wavfn}\ket{\mathbb{S}} = \sum_{n,m} c_m \psi_{n,m}\ket{0},\ee 
where $c_m \in \mathbb{C}_{q,t}$ are the coefficients of a formal $X$-expansion of,
$$\hat{Z}_0 (\mathbb{S}_0;X,q,t) = q^{-\frac{\beta(\beta+1)}{4}}X^{-(\beta+1)/2} (X,q)_{\beta+1} = \sum_{m\in I} c_m X^m $$
where both the index $I$ and $c_m$ depend explicitly on $\beta$. In accordance with Section \ref{solid_torus_sec}, we must specify an inner product, or metric, on the Hilbert space if we wish to compute any 3-manifold invariants. The elements $\psi_{m,n}$ are manifestly linearly independent over $s\ddot{H}$, so we can partially define a metric by requiring the states $\psi_{n,m}\ket{0}$ to be orthonormal, 
$$\bra{n,m}\ket{n',m'} = \delta_{n,n'}\delta_{m,m'}.$$
More generally, the elements must be orthogonal, but of course, we can absorb the overall factor into our wavefunctions if necessary. This consists of a \textit{partial} definition for the inner product as $\psi_{n,m}$ are not a basis for $s\ddot{H}$ (treated as a free $\mathbb{C}_{q,t}$-module).

We can now compute the amplitude for lens spaces $L(p,1)$. We find 
\begin{align*}
Z(L(p,1),\mathbf{a})& =\bra{\mathbb{S},\mathbf{a}_1}\tau_-^p  \ket{\mathbb{S},\mathbf{a}_2}  \\
&= \sum_{n,n',m',m} q^{-n^2p -2n\alpha}  c_m c_{m'} \bra{n',m'}\ket{n,m+np+\alpha}\\
&= \sum_{n,m',m} q^{-n^2p -2n\alpha}  c_m c_{m'} \delta_{m',m+np+\alpha} \\
& \label{refined_lp1}\cong \oint\frac{dX}{2\pi i X} \frac{(X,q)_\infty (X^{-1},q)_\infty}{(Xqt,q)_\infty (X^{-1}qt,q)_\infty} \cdot \sum_{n\in\mathbb{Z}}q^{-\frac{(pn+\alpha)^2}{p}} X^{np+\alpha}.
\end{align*}
We note that this is precisely the half-index of $T[L(p,1)]$, a 3d $\mathcal{N}=2$ theory specified by a vector multiplet (Neumann b.c.) with level $p$ CS term and an adjoint chiral (also Neumann b.c.) with $R$ charge 2 with boundary $2d$ $\mathcal{N}=(2,0)$ theory canceling the anomaly inflow from the CS term \cite{GPPV}\footnote{Strictly, speaking it is the index divided by Cartan contributions from both the vector and chiral multiplets.}. So, we recover the correct L(p,1) amplitude for $\Hat{Z}$. 

However, in the approach described above, we have ignored a fatal flaw: \textit{the solid torus wavefunction is not invariant under Dehn twists}. Therefore, the theory resulting from this choice of wavefunction will be manifestly non-topological. That is, a general ansatz for $\ket{\mathbb{S}}$ is
$$\ket{\mathbb{S}} = \sum_{n,m} S^{n,m} \psi_{n,m}, \hspace{4mm} S^{n,m}\in \mathbb{C}_{q,t}.$$
Recalling the argument in Section \ref{solid_torus_sec}, topological invariance (up to a $q$-factor) requires invariance under Dehn twists,
$$\tau_+^a \ket{\mathbb{S}} = q^{f(a)}\ket{\mathbb{S}}$$
\be \label{dehn_twist_invariance} \Leftrightarrow S^{n,m} = S^{m} \cdot q^{-m n-f(\frac{n}{m})} \ee 
for some $S^m$. One immediately sees that \eqref{wrong_ref_torus_wavfn} violates this directly. However, if we start from this requirement, it is impossible to recover the correct expression for $L(p,1)$ (term by term, at least). 

We are led to conclude that the problem of realizing the decorated TQFT associated with refined $\hat{Z}$ invariants (and thus generalizing the results of this paper) via $s\ddot{H}$ can be broken down to:
\begin{enumerate}
    \item Extending the spherical Double Affine Hecke Algebras $s\ddot{H}$ to some rational extension, $s\ddot{H}^{ext}$, which include rational exponents of $X,Y$ and obeys $$s\ddot{H}^{ext} \underset{t\rightarrow1}{\longrightarrow }\mathcal{O}^{\mathfrak{sl}_2}_\mathbb{Q}.$$ This will serve as the algebra of operators in the refined theory, $\mathcal{O} = s\ddot{H}^{ext} $, and therefore defines the torus Hilbert space $Z_0(T^2)$ by taking the $\mathbb{C}$-span of all $\mathbb{C}_{q,t}$ linearly independent elements in $s\ddot{H}^{ext}$.
    \item Specifying a wavefunction for the solid torus $\ket{\mathbb{S}} \in Z_0(T^2)$, which is invariant, up to a $q$ factor, under Dehn twists, $\tau_+^a \ket{\mathbb{S}} = q^{f(a)}\ket{\mathbb{S}}$ and which recovers the unrefined solid torus wavefunction from Section \ref{solid_torus_sec} under $t\rightarrow1$.
    \item Specifying a metric $\bra{\cdot}\ket{\cdot}:s\ddot{H}^{ext}\times s\ddot{H}^{ext} \rightarrow\mathbb{C}_{q,t}[[q]]$ and ensuring the amplitude $\bra{\mathbb{S}}\tau_-^p  \ket{\mathbb{S}}$ recovers the $L(p,1)$ amplitude up to a $q$ factor. 
\end{enumerate}

We are tempted to conjecture that a solution to the 3 constraints above exists and is in fact unique. Not only that, but by the same arguments in this paper, this starting point would uniquely constrain the amplitudes of all Seifert manifolds with $b_1 =0$. This theory would also immediately contain line operators (that is, Wilson loops), whose expectation values in $S^3$ would recover those of refined Chern--Simons theory. Moreover, one would be able to compute refined knot invariants (with integral coefficients) in any background Seifert manifold (where the knots are allowed to wrap any non-trivial homotopy cycles).\footnote{ These Wilson loop operators exist in the unrefined theory as well and have a rich topological and algebraic structure, which we will explore in an upcoming paper.} It would be very interesting to define and solve the refined theory in the future, not only due to its interesting algebraic and representation theoretic structures, but also because it seems to be easier to categorify. We hope to explore this and related problems in the future.

%%%%%%%%%%%%%%%%%%%%%%%%%%%%%%%

\subsection{Line Operators and Complex Extensions}
The $\hat{Z}$-TQFTs we have developed in this work are expected to provide a non-perturbative completion of $G_{\C}$ Chern—Simons theory. Moreover, from the $M$-theory perspective discussed in the previous subsection, we can consider $M$-branes in the $M$-theory space-time configurations, which can be viewed as line defects inside the three-manifold. As such, we expect the $\hat{Z}$-TQFTs to have line operators. 

In Chern--Simons theory with non-abelian gauge field $A$, the Wilson loop operators are defined as,
$$W_\lambda (\gamma) = \Tr_\lambda \mathcal{P}exp\left( i \oint _\gamma A\right).$$ 
Recall that $X$ and $Y$ are holonomies around the generating cycles of the 2-torus. Thus, if we wish to consider Wilson loop operators, wrapping the meridian cycle of the knot complement, the appropriate object to consider in $\mathcal{H}(T^2)$ is, 
$$W_\lambda (X) = \Tr_\lambda(X),$$
where we allow $\lambda$ to be any representation of the Lie algebra $\mathfrak{g}$ (so long as $\Tr_\lambda(X) \in \mathcal{H}(T^2)$). Then, we may define an operator that instantiates wrapping a Wilson loop operator around the longitude of a knot complement, 
$$W_\lambda : \mathcal{H}(T^2) \rightarrow \mathcal{H}(T^2),$$
whose action is given as follows,
$$W_\lambda \ket{n,m} = W_\lambda Y^n X^m\ket{0} =Y^n X^m \chi _\lambda(X)\ket{0}, $$
where $\chi_\lambda (\cdot)$ is the character of the representation. Thus, for instance, the unknot complement with a Wilson line charged under $\lambda$ would be given by,
$$W_\lambda \ket{\mathbb{S},\mathbf{a}}.$$
Let $\mathfrak{g = sl_2}$ for simplicity and let $\lambda$ denote the highest weight of a Verma module representation of $\mathfrak{sl_2}$. In this case, we have, 
$$W_\lambda \ket{\mathbb{S},\mathbf{a}} = \sum_{n\in \Z+\mathbf{a}}Y^n X^\lambda. $$
From general principles, we expect that the amplitude, 
\be \label{wilson_line_knot_complement}\bra{\mathbb{S},\mathbf{a}_2} W_\lambda ^\dagger S^\dagger \ket{K,\mathbf{a}_1},\ee
for generic $\lambda$, should recover the GM series for the knot complement $F_K$. Computing this amplitude, we find, 
$$\bra{\mathbb{S},\mathbf{a}_2} W_\lambda ^\dagger S^\dagger \ket{K,\mathbf{a}_1} = \delta^\Z(\mathbf{a}_1+\lambda) \sum_{m\in \Z+\frac{1}{2}}q^{2 m\lambda} f^m_K(q),$$
where we have taken $\mathbf{a}_2 = (0,\mathfrak{s}_{\frac{1}{2}})$. It is perhaps not too surprising that introducing non-local defects affects the $\mathrm{Spin}^c$-structures of the 3-manifolds in question, but the details of this we leave for future work. For now, we take this to be the case and set $\mathbf{a}_1 = -\lambda \mod \Z$. Whenever $\lambda \in \frac{1}{2}\Z$, the quantity above recovers the colored Jones polynomial of $K$. If instead we take $\lambda$ to be a generic complex number ($\lambda \in \C$), then we can recover the GM series $F_K$ by redefining $X=q^{2\lambda}$. In a way, this discussion draws the first connection to the work of \cite{S2}. More importantly, however, it suggests the complex extension of the quantum torus algebra $QT$ rather than the rational extension. In this paper, we needed the $\Q$-extension to write down functorial cutting and gluing rules. However, in this context, it may be that a more complete framework involves enlarging $\Q \rightarrow \C$ and consequently the construction of a $H^1(\cdot, \C/\Z)$ decorated $\mathrm{Spin}$ TQFT with functorial cutting and gluing rules. In this scenario, the module associated with the torus would instead be,
$$\boxed{\mathcal{H}(T^2) = \hat{\mathcal{O}}^\mathfrak{g}_\C.}$$
The $SL(2,\Z)$ action would remain the same, so this would automatically be a representation of $SL(2,\Z)$. It is an interesting future problem to explore other implications of this.

%%%%%%%%%%%%%%%%%%%%%%%%%%%%%%%%%%%%

\subsection{Other Future Prospects}
In light of this work, several interesting open questions remain:
\begin{enumerate}
    \item \textbf{Extension to other TQFTs:} In this work, we established rules for decorated $\mathrm{Spin}$ TQFTs. Which other known TQFTs fall within this construction? For instance, we suspect that $U(1|1)$ Chern--Simons theory, which computes the Reidemeister torsion invariants (see \textit{e.g.}, \cite{Mik}), and the TQFT that computes Akhmecher--Johnson--Krushkal \cite{AJK} fit within this framework. It would be valuable to verify this explicitly.
    \item \textbf{General boundaries:} Can we extend the discussion in this paper to three-manifolds with arbitrary boundaries? In particular, can we find the mapping class group representations associated with higher-genus Riemann surfaces?
    \item \textbf{Quantum moduli spaces and skein algebras:} A large portion of this paper relied on an extension of the well-known quantum torus algebra, which arises from the novel quantization of the moduli space of flat connections on the torus. Could the quantizations of the moduli space of flat connections for more general Riemann surfaces prove useful in studying $\hat{Z}$ invariants of Heegaard splittings? How do skein algebras fit into this picture, and can they help us understand $\hat{Z}$ for Heegaard splittings?
    \item \textbf{Quantum Modularity:} Do the explicit formulas for $\hat{Z}$ of Seifert manifolds shine any light on quantum modularity of the $\hat{Z}$ invariants?
    \item \textbf{Line operators:} As mentioned in the previous section, we know that line operators are part of this theory. Is there a state-line operator correspondence? If so, it would be interesting to establish a precise state-operator correspondence using the results developed in this work. This would enable us to study the braiding and fusion of operators colored by infinite-dimensional representations of the Lie algebra of the gauge group and determine general structural aspects of this theory.  
\end{enumerate}

\section*{Acknowledgement}
We would like to thank John Chae, Davide Passaro, Lara San Martin Suarez, and Josef Svoboda for helpful discussions. We extend special thanks to Sergei Gukov for many insightful discussions and mentorship, and to Sunghyuk Park for insightful discussions and for sharing his notes on cabling and satellites. This work is supported by a Simons Collaboration Grant on New Structures in Low-Dimensional Topology. The work of MJ is supported by the Walter Burke Institute for Theoretical Physics, the U.S. Department of Energy, Office of Science, Office of High Energy Physics, under Award No. DE-SC0011632.

\appendix

\section{Framed Solid Torus Comparison} 

In this appendix, we compare the solid torus wavefunctions from \cite{GM} and this work. In \cite{GM}, the authors determined the $\Hat{Z}$ invariant for $\mathbb{S}_{p/r}$ with a certain decoration via a plumbing formula. That is, equation (87) from \cite{GM} reads (up to an overall $q$ factor due to framing anomaly),

$$\Hat{Z}_\alpha(\mathbb{S}_{p/r};X,n,q) = \sum_{\epsilon=\pm 1}\epsilon q^{- \frac{(p n + \alpha)^{2}}{p r}}\cdot X^{\frac{n p + \alpha}{r} + \frac{\epsilon}{2r}} \cdot \delta^\mathbb{Z}\left( \frac{p n}{r}+\frac{\alpha}{r}+ \frac{\epsilon}{2r} +\frac{1}{2}\right),$$
where $\delta^\mathbb{Z}(x) = \delta_{0,x\mod{\mathbb{Z}}}$. 
The wavefunction for the framed solid torus given in \ref{solid_torus_sec} can be written as, 
\begin{align*}
    \ket{\mathbb{S}_{p/r} , \mathbf{a} } =\gamma_{p/r}\cdot \ket{\mathbb{S},\mathbf{a}}= & \sum_{\epsilon = \pm 1} \sum_{k \in \Z}  \epsilon q^{- \tfrac{a b}{4} - a p \epsilon (k+\mathbf{a}) - p r (k+\mathbf{a})^{2} } Y^{ \frac{a \epsilon}{2} + r (k + \mathbf{a})} X^{\frac{b \epsilon}{2} + p(k + \mathbf{a}) }.
\end{align*}
The label $\mathbf{a}$ can, in general, take any value in $\Q/\Z$. The agreement with \cite{GM} will occur with specific values of $\mathbf{a}$, namely, the values that come from restrictions of closed 3-manifolds resulting from gluing of the torus boundary. That is, suppose $\mathbf{a}$ is such that $\mathbf{a} = \frac{b x}{p} + \frac{a}{2}$ for some $x \in \frac{r+1}{2} + \Z$, and let $m = k r + \frac{a(2 x + r + \epsilon)}{2}$. Therefore, $k + \mathbf{a} = \frac{m}{r} + \frac{b x}{p} - \frac{a(2 x + \epsilon)}{2r} $. Now, we can write $\ket{\mathbb{S}_{p/r} , \mathbf{a} } $ as, 
\begin{align*}
& \ket{\mathbb{S}_{p/r} , \mathbf{a} } \endline = & \sum_{\epsilon = \pm 1} \sum_{m \in \Z} \epsilon q^{- \frac{a(b r - a p)}{4 r} - \frac{m^{2}p}{r} - \frac{ 2 m x (b r - a p)}{r} - \frac{(b r - a p )^{2} x^{2}}{p r} }  Y^{ m + \frac{(b r - a p) x}{p}} X^{ \frac{ m p }{ r} + \frac{(b r - a p)(2 x + \epsilon)}{2 r} } \delta^{\Z}\left( \frac{2 m-a(2x+ r + \epsilon)}{2r} \right),   \endline
 = & \sum_{\epsilon = \pm 1} \sum_{m \in \Z} \epsilon q^{- \frac{a}{4 r} - \frac{(m p + x)^{2}}{p r} }  Y^{ m + \frac{ x}{p}} X^{ \frac{ m p + x}{ r} + \frac{ \epsilon}{2 r} } \delta^{\Z}\left( \frac{2 m-a(2x+ r + \epsilon)}{2r} \right) .
\end{align*}
In the second equality, we have used $ b r - a p =1$. Since $gcd(p,r)=1$ and $ x+ \frac{r + \epsilon}{2} \in \Z$, we have,
\begin{align*}
    \delta^{\Z}\left( \frac{2 m-a(2x+ r + \epsilon)}{2r} \right) = & \delta^{\Z}\left(  \frac{p m }{r} - \frac{a p (2 x + r + \epsilon)}{2 r}  \right) \endline
   =  & \delta^{\Z}\left(  \frac{p m }{r} - \frac{(b r -1) ( x + \frac{r + \epsilon}{2})}{ r}  \right) \endline
   = & \delta^{\Z}\left(  \frac{p m }{r} + \frac{x}{r} + \frac{\epsilon}{2r} + \frac{1}{2}  \right).
\end{align*}
In the second equality, we have used $ b r - a p =1$, in the third equality we have used that $ x+ \frac{r + \epsilon}{2} \in \Z$. Using this, we get, 
\begin{align*}
 \ket{\mathbb{S}_{p/r} , \mathbf{a} } = &   \sum_{\epsilon = \pm 1} \sum_{m \in \Z} \epsilon q^{- \frac{a}{4 r} - \frac{(m p + x)^{2}}{p r} }  Y^{ m + \frac{ x}{p}} X^{ \frac{ m p + x}{ r} + \frac{ \epsilon}{2 r} }  \delta^{\Z}\left(  \frac{p m }{r} + \frac{x}{r} + \frac{\epsilon}{2r} + \frac{1}{2}  \right).
\end{align*}
Letting $x=\alpha$, we find an equality up to a factor of $q$:
$$\hat{Z}_\alpha(\mathbb{S}_{p/r};X,n,q) = q^{a/4r}\sum_{m\in \Q}\ket{0,m} \braket{{n+\frac{\alpha}{p},m}, {\mathbb{S}_{p/r} , \mathbf{a} } }.$$

\section{Inverses in the Algebra of Operators}\label{high_rank_expansions}
\subsection*{Definition}
For this section, we are only concerned with the inverses of elements depending only on $X$. As such, we use the notation $$\mathcal{C}^\mathfrak{g}_R \equiv \mathcal{O}^\mathfrak{g}_R \vert_{Y=1}.$$
The corresponding space $\Hat{\mathcal{C}}^\mathfrak{g}_R = \hat{\mathcal{O}}^\mathfrak{g}_R|_{Y=1}$ is an $\mathcal{C}^\mathfrak{g}_R$-module, but has no internal multiplication that can be defined. Despite this, it is easy to make sense of an inverse element in $r \in \Hat{\mathcal{C}}^\mathfrak{g}_R$ of $e$, so long as $e \in \mathcal{C}^\mathfrak{g}_R$. We do this by requiring that inverses be compatible with the Weyl group action on $\Hat{\mathcal{C}}^\mathfrak{g}_R$. Explicitly,
\begin{definition}\label{weylcompatibleinverse}
    If $e \in \mathcal{C}^\mathfrak{g}_R$ is any element, we say $r \in \Hat{\mathcal{C}}^\mathfrak{g}_R$ is an inverse of $e$ in $\Hat{\mathcal{C}}^\mathfrak{g}_R$ if the following conditions are satisfied,
    $$e\cdot r = 1$$
    $$w(e\cdot r) = w(e) w(r) =1$$
    for all $w\in W$. 
    
\end{definition}

It is helpful to also define another $\mathcal{C}^\mathfrak{g}_R$-module, 
$$\ddot{\mathcal{C}}^\mathfrak{g}_R = \C_q ((\{X^\mu \}_{\mu \in P_R})).$$
This is actually a field of power series in the $X$ variables, which fits inside $\hat{\mathcal{C}}^\mathfrak{g}_R$ as a $\mathcal{C}^\mathfrak{g}_R$-modules. 
$$\mathcal{C}^\mathfrak{g}_R \subset \ddot{\mathcal{C}^\mathfrak{g}_R} \subset \Hat{\mathcal{C}^\mathfrak{g}_R}$$
Now, we are ready to prove the main result of this section, 
\begin{lemma}
    Suppose $e \in \mathcal{C}^\mathfrak{g}_R$ is Weyl symmetric or antisymmetric. Then, $e$ has a unique inverse (in the sense above) in $\Hat{\mathcal{C}^\mathfrak{g}_R}$
\end{lemma}
\begin{proof}
    Without loss of generality, assume $e$ is anti-symmetric. Consider the map,
$$a: \ddot{\mathcal{C}^\mathfrak{g}_R} \rightarrow \Hat{\mathcal{C}^\mathfrak{g}_R}$$
$$c\mapsto \frac{1}{|W|}\sum_{w\in W} \epsilon(w) w(c). $$
Since $\ddot{\mathcal{C}^\mathfrak{g}_R}$ is a field, we can pick the honest inverse $e^{-1}$ of $e$ in $\ddot{\mathcal{C}^\mathfrak{g}_R}$. Let us define $r = a(e^{-1})$. Now, compute the quantity,
$$e \cdot a(e^{-1}) = \frac{1}{|W|}\sum_{w\in W} \epsilon(w) e\cdot  w(e^{-1}) .$$
If $e$ is Weyl antisymmetric, then $w(e) = \epsilon(w) e$, so:
$$e \cdot a(e^{-1}) = \frac{1}{|W|}\sum_{w\in W}w(e)\cdot w(e^{-1}) = 1.$$
The first condition $e\cdot r=1$ is verified. For the second, observe that it is sufficient for $w(r) = \epsilon (w) r$, but this is true by construction. Now, we show uniqueness. Suppose $r'$ is another inverse of $e$ in $\Hat{\mathcal{C}^\mathfrak{g}_R}$. It is antisymmetric, so it must have some pre-image under $a$, $a(x) = r'$. This gives, 
$$e\cdot a(x-e^{-1})=0 .$$
But we know, 
$$e\cdot a(x-e^{-1})=\frac{1}{|W|}\sum_{w\in W}w(e)w(x-e^{-1}) = \bigg(\frac{1}{|W|}\sum_{w\in W}w\bigg)(ex -1).$$
The kernel of the map in big brackets is the subspace of all Weyl antisymmetric elements. Therefore, $x= e^{-1} +s$ for $s$ some symmetric element. Hence, the two inverses are the same.
\end{proof}

As a relevant example for a big chunk of the paper, we mention the result below, 

\begin{corollary}
Let $\mathfrak{g=sl_2}$. If we write the inverse of the element $(X^{\frac{1}{2}}-X^{-\frac{1}{2}})^p \in \mathcal{C}^\mathfrak{sl_2}_R$ in $\hat{\mathcal{C}}^\mathfrak{g}_R$ as $(X^{\frac{1}{2}}-X^{-\frac{1}{2}})^{-p}$, then we have, 
$$(X^{\frac{1}{2}}-X^{-\frac{1}{2}})^{-p} = \sum_{m\in \frac{1}{2} \Z}f^{p+2}_m \cdot X^m$$
$$f^{p+2}_m = \frac{1}{2}\mathrm{sgn}(m)^p\binom{\frac{p}{2}-1+|m|}{p-1}\delta^\Z\left(m+\frac{p}{2}\right).$$
\end{corollary}
\begin{proof}
    Obvious.
\end{proof}

\section{Higher Rank Expressions}
In this section, we study and define inverse elements of the form,
$$\left(\sum_{w\in W}\epsilon(w)X^{w(\rho)} \right)^{2-d}$$
in the algebra of operators we have defined. First, recall that we have the ring, 
$$\mathcal{O}^\mathfrak{g}_{K} = QT^{K}_\mathfrak{g} = \frac{\langle X^\mu ,Y^\lambda\rangle_{\mu,\lambda\in P_K}}{(X^\mu Y^\lambda - q^{-2(\mu,\lambda)} Y^\lambda X^\mu)}.$$
As in the $\mathfrak{g} = \mathfrak{sl}_2$ case, we define, 
\be\label{high_rank_f}
\left(\sum_{w\in W}\epsilon(w)X^{w(\rho)} \right)^{2-d}=\sum_{\mu\in P} f^\mu_d X^{-\mu}.
\ee
We wish for an explicit expression for $f^\mu_d$; To this end, we consider the Konstant partition function, $\mathcal{K}(\mu)$, which has as its definition,
$$\left(\sum_{w\in W}\epsilon(w)X^{w(\rho)} \right)^{-1} \overset{P_+ \hspace{2mm }expansion}{=} \sum_{\mu\in Q\cap P_+} \mathcal{K}(\mu) X^{-\mu-\rho}.$$
We will denote this expansion in positive weights only by $=_+$ and work only with this formal series. At the end, we will antisymmetrize to find $f^\mu_d$. As such, we have, 
\begin{align*}
\left(\sum_{w\in W}\epsilon(w)X^{w(\rho)} \right)^{-p} & =_+X^{-p \rho}\sum_{\mu_i \in Q \cap P_+} \mathcal{K}(\mu_1)\cdots\mathcal{K}(\mu_p) X^{-\mu_1-\mu_2-\cdots-\mu_p}  \\
& = X^{-\rho p}\sum_\mu X^{-\mu}\cdot \underset{\mu_1+\cdots+\mu_p =\mu}{\sum_{\mu_i}}\prod^p_{i=1}\mathcal{K}(\mu_i)\\
&=\sum_{\mu \in (Q+p\rho)\cap P_+}X^{-\mu} \cdot \tilde{f}^\mu_{p+2},
\end{align*}
where we have set, 
$$\tilde{f}^\mu_{p+2} = \underset{\mu_1+\mu_2+\cdots+\mu_p+p \rho=\mu }{\sum_{\mu_i \in Q\cap P_+}} \prod_{i=1,\ldots,p}\mathcal{K}(\mu_i).$$
Antisymmetrizing the above expression in agreement with definition \ref{weylcompatibleinverse}  yields,
$$\left(\sum_{w\in W}\epsilon(w)X^{w(\rho)} \right)^{-p}=\sum_{\mu \in P} f^{\mu}_{p+2} X^{-\mu},$$
where we have redefined, 
$$\boxed{f^\mu_{p+2} = \frac{1}{|W|}\sum_{w\in W} \epsilon(w)\cdot \delta^{(Q+p\rho )\cap P_+}(w(\mu)) \cdot \tilde{f}^{\mu}_{p+2} .}$$

\section{Notations and Conventions}\label{notationsandconventions}
\begin{enumerate}
    \item Let $S$ be a some set, then $\delta^S(x)$ represents the indicator function on that set. That is, $\delta^S(x)=1$ if $x\in S$ and $\delta^S(x)=0$ otherwise. 
    \item In the context of a link complement $M = S^3 \setminus\nu(L)$, $\mathfrak{s}_0$ represents the $\mathrm{Spin}$ structure evaluating to 0 on all meridians of $M$. Similarly, $\mathfrak{s}_{\frac{1}{2}}$ represents the Spin structure evaluating to $\frac{1}{2}$ on all meridians. 
\end{enumerate}

\end{document}